\documentclass[11pt,letterpaper]{article} %STOC Cfp requires letter-size paper format.
\usepackage[T1]{fontenc}
\usepackage[margin=1in]{geometry}

\usepackage{graphicx}
\usepackage{amsmath,amssymb,amsthm,mathtools}
\usepackage{lmodern}
\usepackage{microtype}
\usepackage{xcolor}
\usepackage{enumitem}
\usepackage{aliascnt}
\usepackage{authblk}
\usepackage[hyphens]{url}
\usepackage{listings}
\usepackage{subcaption}
\usepackage{algorithm}
\usepackage{algpseudocode}
\algrenewcommand\algorithmicrequire{\textbf{Input:}}
\algrenewcommand\algorithmicensure{\textbf{Output:}}

\usepackage{booktabs,tabularx}
\usepackage[hidelinks]{hyperref}
\usepackage[nameinlink,noabbrev]{cleveref}
\newif\ifpapercomments
\papercommentstrue
\newif\ifpapertodos
\papertodostrue

\definecolor{CommentBlue}{HTML}{1D4ED8}
\definecolor{CommentPurple}{HTML}{7E22CE}
\definecolor{CommentTeal}{HTML}{087F8C}
\definecolor{CommentRed}{HTML}{B91C1C}
\definecolor{CommentOrange}{HTML}{B45309}
\definecolor{CommentMagenta}{HTML}{A21CAF}

\DeclareRobustCommand{\inlinecomment}[3][CommentBlue]{%
  \ifpapercomments
    {\color{#1}\small\sffamily[\textbf{#2}: #3]}%
  \fi
}
\DeclareRobustCommand{\margincomment}[3][CommentBlue]{}

\newcommand{\DeclareAuthorComments}[3]{%
  \expandafter\newcommand\csname #1comment\endcsname[1]{%
    \inlinecomment[#3]{#2}{##1}}%
  \expandafter\newcommand\csname #1margin\endcsname[1]{%
    \margincomment[#3]{#2}{##1}}%
}
\DeclareAuthorComments{sl}{SL}{CommentBlue}
\DeclareAuthorComments{pv}{PV}{CommentTeal}
\DeclareAuthorComments{cl}{Claude}{CommentOrange}
\DeclareAuthorComments{ys}{YS}{CommentMagenta}
\DeclareAuthorComments{gpt}{GPT}{CommentPurple}

\newif\ifpaperchanges
\paperchangestrue
\DeclareRobustCommand{\add}[1]{\ifpaperchanges{\color{blue}\relax#1}\else#1\fi}
\DeclareRobustCommand{\del}[1]{\ifpaperchanges{\color{gray}\relax#1}\fi}
\DeclareRobustCommand{\gptadd}[1]{\add{#1}}
\DeclareRobustCommand{\gptdel}[1]{\ifpaperchanges\del{[#1]}\space\fi}

\newcommand{\gptoldlabel}[1]{}

\makeatletter
\newenvironment{gptremoved}{%
  \par\medskip\begingroup
  \let\add\@firstofone\let\del\@firstofone
  \renewcommand{\margincomment}[3][gray]{}%
  \renewcommand{\inlinecomment}[3][gray]{\textnormal{\footnotesize[##2: ##3]}}%
  \color{gray}\small\noindent\textit{Previous version (superseded).}\par\nobreak\smallskip
}{%
  \par\smallskip\noindent\textit{End of previous version.}\par\medskip
  \endgroup
}
\makeatother


\newlist{drafttodos}{enumerate}{1}
\setlist[drafttodos]{%
  label={\(\square\)\,\textsf{T\arabic*.}},ref={T\arabic*},
  leftmargin=3.7em,labelsep=0.6em,itemsep=5pt,topsep=5pt,parsep=0pt}

\newtheorem{theorem}{Theorem}[section]

\newaliascnt{definition}{theorem}
\newtheorem{definition}[definition]{Definition}
\aliascntresetthe{definition}
\crefname{definition}{definition}{definitions}
\Crefname{definition}{Definition}{Definitions}
\newaliascnt{lemma}{theorem}
\newtheorem{lemma}[lemma]{Lemma}
\aliascntresetthe{lemma}
\crefname{lemma}{lemma}{lemmas}
\Crefname{lemma}{Lemma}{Lemmas}
\newaliascnt{corollary}{theorem}
\newtheorem{corollary}[corollary]{Corollary}
\aliascntresetthe{corollary}
\crefname{corollary}{corollary}{corollaries}
\Crefname{corollary}{Corollary}{Corollaries}
\newaliascnt{proposition}{theorem}

\aliascntresetthe{proposition}
\crefname{proposition}{proposition}{propositions}
\Crefname{proposition}{Proposition}{Propositions}
\newaliascnt{assumption}{theorem}

\aliascntresetthe{assumption}
\crefname{assumption}{assumption}{assumptions}
\Crefname{assumption}{Assumption}{Assumptions}
\newaliascnt{remark}{theorem}

\aliascntresetthe{remark}
\crefname{remark}{remark}{remarks}
\Crefname{remark}{Remark}{Remarks}
\newaliascnt{observation}{theorem}
\newtheorem{observation}[observation]{Observation}
\aliascntresetthe{observation}
\crefname{observation}{observation}{observations}
\Crefname{observation}{Observation}{Observations}

\newcommand{\E}{\mathbb{E}}
\newcommand{\Var}{\mathrm{Var}}
\DeclareMathOperator*{\union}{\bigcup}

\newcommand{\calD}{\mathcal{D}}
\newcommand{\calE}{\mathcal{E}}
\newcommand{\calR}{\mathcal{R}}
\newcommand{\calU}{\mathcal{U}}
\newcommand{\calM}{\mathcal{M}}
\newcommand{\calP}{\mathcal{P}}

\newcommand{\vecc}{\mathbf{c}}
\newcommand{\vecx}{{\mathbf{x}}}
\newcommand{\vecy}{{\mathbf{y}}}

\newcommand{\open}{\mathrm{open}}
\newcommand{\phys}{\mathrm{phys}}
\newcommand{\virt}{\mathrm{virt}}
\newcommand{\rest}{\mathrm{rest}}
\newcommand{\orig}{{}}

\newcommand{\cost}{\mathrm{cost}}
\newcommand{\pos}[1]{\left(#1\right)_{+}}
\newcommand{\Disp}{\mathcal D}

\newcommand{\ip}[2]{\langle #1,#2\rangle}
\newcommand{\norm}[1]{\lVert #1\rVert}

\title{Iterative Randomized Rounding for $k$-Median in General and Euclidean Metrics}

\DeclareRobustCommand{\repeatthanks}[1]{\footnotemark[#1]}
\newcommand{\authorspace}{2em}
\author{
    Jaros{\l}aw Byrka\thanks{University of Wroc{\l}aw, \texttt{\{jby, yongho, phanu.vajanopath\}@cs.uni.wroc.pl}}
    \hspace{\authorspace}
    Yuhao Guo\thanks{Tsinghua University, \texttt{\{yh-guo23, wcz23\}@mails.tsinghua.edu.cn}}
    \hspace{\authorspace}
    Yang S. Hu\thanks{Carnegie Mellon University, \texttt{yanghu2@andrew.cmu.edu}}% Yang: The middle name is intentional :)
    \hspace{\authorspace}
    Shi Li\thanks{Nanjing University, \texttt{\{shili@, zaixuanwang@smail.\}nju.edu.cn}}
    \\
    Yongho Shin\repeatthanks{1}
    \hspace{\authorspace}
    Phanu Vajanopath\repeatthanks{1}
    \hspace{\authorspace}
    Chengzhang Wan\repeatthanks{2}
    \hspace{\authorspace}
    Zaixuan Wang\repeatthanks{4}
}

\date{}

\begin{document}

\maketitle

\begin{abstract}
We study the classic $k$-median problem and refine the iterative rounding framework recently introduced by Byrka, Guo, Hu, Li, Wan, and Wang~\cite{BGHL2026}, which yields a $(2+\varepsilon)$-approximation.

Our algorithm opens $k+O(1)$ facilities and achieves an expected connection cost of at most $1.9999$ times the value of the standard LP relaxation. Although the numerical improvement over $2$ is small, it establishes a conceptual result: the integrality gap of the standard relaxation is strictly below $2$ when a constant number of extra facilities is allowed. This partially answers an open question of Li and Svensson~\cite{LS16}. Combined with their reduction, our result yields a $1.9999$-approximation algorithm for $k$-median.

Our approach builds on several technical ingredients. First, we develop a preprocessing routine that makes the LP solution $\mathrm{poly}(\varepsilon)$-integral, at the cost of a $(1+\varepsilon)$ factor in the connection cost and $O(1)$ extra facilities. Unlike the preprocessing in \cite{BGHL2026}, our conversion is unconditional: it requires no special treatment of clients whose fractional facilities are concentrated. Next, we show that the main rounding algorithm can open facilities one at a time, eliminating the batch processing used in \cite{BGHL2026}. Finally, the improvement over 2 comes from a tighter analysis for clients whose distances to their facilities are not all equal, and a complementary algorithm for the case in which these distances are nearly uniform.

We also show that the iterative rounding algorithm admits an improved
analysis in Euclidean metrics. We obtain a $1.71955$-approximation for
Euclidean $k$-median and a $1.4561$-approximation for Euclidean
uncapacitated facility location (UFL).
For every $p\ge2$ and a sufficiently small $\varepsilon>0$, we further
obtain a $\bigl(\frac{3^p+2}{3}+\varepsilon\bigr)$-approximation for the Euclidean $k$-clustering problem with the cost function being the $p$-th power of the distance.
In particular, taking $p=2$ yields an
$(11/3+\varepsilon \approx 3.667 + \varepsilon)$-approximation for Euclidean $k$-means,
improving upon the $(3+\ln2+\varepsilon \approx 3.693 + \varepsilon)$-approximation of
Anand, Charikar, Cohen-Addad, Gao, Grandoni, Lee, Sharma, and van Wijland~\cite{ACCGGLSW26}.
Our approximation ratios for
Euclidean $k$-median, $k$-means, and UFL are strictly below the
respective general-metric hardness bounds of $1+2/e \approx 1.736$, $1+8/e \approx 3.943$,
and $1.463$~\cite{JMS02,GK99}.
The Euclidean $k$-median bound relies on a geometric inequality proved via a nonnegativity certificate constructed from positive-semidefinite kernels.
\end{abstract}

% \paragraph{AI Disclosure Statement:} We used ChatGPT to assist us in the development of various detailed arguments across the paper. We also used Claude to assist us in preparing a presentation of the most technical details of the arguments. The authors verified the correctness and originality of all content, including references.

%\ifpapertodos
%  \input{todos}
%  \clearpage
%\fi

\clearpage
\tableofcontents
\clearpage

\section{Introduction}\label{sec:intro}

In the $k$-median problem, we are given a set $C$ of clients, a set $F$
of facilities, a metric $d$ on $C\cup F$, and a positive integer $k$.
The task is to open a set $O\subseteq F$ of at most $k$ facilities and
connect every client to its nearest open facility. The goal is to
minimize the total connection cost
\[
  \cost(O)=\sum_{j\in C}d(j,O),
  \qquad\text{where } d(j,O)=\min_{i\in O}d(i,j).
\]
$k$-median is one of the most basic clustering problems. It has also
served as a testbed for many techniques in approximation algorithms,
including LP rounding, primal-dual algorithms, greedy algorithms, and
local search. The problem is NP-hard, and we focus on polynomial-time
algorithms that compute $\alpha$-approximate solutions for a constant
$\alpha$.

We first review previous results in Section~\ref{sec:intro-previous}.
We then state our results in Section~\ref{sec:intro-results}, briefly
survey further related work in Section~\ref{sec:intro-related}, and give
an informal description of our approach in
Section~\ref{sec:intro-overview}.

\subsection{Previous results}\label{sec:intro-previous}

Many of the algorithms discussed below use Lagrangian relaxation. They
replace the hard constraint of opening at most $k$ facilities with a
uniform opening cost for each facility, yielding an
uncapacitated facility location problem with uniform opening costs. A Lagrangian multiplier preserving
(LMP) $\alpha$-approximation algorithm for this relaxation returns a solution
whose connection cost plus $\alpha$ times its total facility opening cost
is at most $\alpha$ times the optimal total cost.
Such an
algorithm can be used to produce a randomized solution that opens $k$
facilities in expectation and has expected connection cost at most $\alpha$ times
the optimal $k$-median cost. Without loss of generality, this distribution
can be supported on at most two solutions; the resulting convex combination
is called a \emph{bi-point solution}.

\paragraph{$k$-median.}
The first constant-factor approximation algorithm for $k$-median was
obtained by Charikar, Guha, Tardos, and Shmoys~\cite{CGTS02}. Their LP-rounding algorithm achieved a factor of $6\frac23$. Jain and Vazirani~\cite{JV01} then gave
a $6$-approximation by combining a primal-dual LMP $3$-approximation
algorithm with a factor-$2$ rounding of bi-point solutions. Charikar and
Guha improved this to $4$~\cite{CG99}. Jain, Mahdian, Markakis, Saberi, and Vazirani~\cite{JMMSV03} gave a
greedy LMP $2$-approximation, and Arya, Garg, Khandekar, Meyerson, Munagala, and Pandit~\cite{AGKMMP04} gave a
$(3+\varepsilon)$-approximation based on local search.

Further progress came from better bi-point rounding. Li and
Svensson~\cite{LS16} gave a bi-point rounding algorithm with factor
$(1+\sqrt3)/2+\varepsilon \approx1.366+\varepsilon$
that opens $k+O(1)$ facilities. They also introduced a reduction that
turns such a pseudo-approximation into a true approximation that opens
at most $k$ facilities, with only a multiplicative $1+\varepsilon$ loss. Together, these
yielded a $(1+\sqrt3+\varepsilon \approx 2.732 + \varepsilon)$-approximation. The bi-point
rounding factor was improved further in~\cite{BPRST17,GPST23}.
Independently, Cohen-Addad, Grandoni, Lee, and Schwiegelshohn~\cite{CGLS23} gave an LMP algorithm for
facility location with factor slightly below $2$. These lines of work led to approximation factors of $2.675$, $2.67059$, and
$2.613$~\cite{GPST23}.

Recently, Cohen-Addad, Grandoni, Lee, Schwiegelshohn, and
Svensson~\cite{CGLSS25} avoided the loss from bi-point rounding
altogether and obtained a $(2+\varepsilon)$-approximation. They adapted
the greedy algorithm of~\cite{JMMSV03} so that it opens only slightly
more than $k$ facilities while keeping the factor $2+\varepsilon$. The
number of extra facilities in their approach is
$O(\log n/\varepsilon^2)$, which is too large for the reduction
of~\cite{LS16}. They therefore developed a separate algorithm for
instances whose cost is highly sensitive to the number of open
facilities. The work of~\cite{BGHL2026} provided an alternative route to
the same factor via iterative randomized rounding of the standard LP
relaxation. Its rounding procedure is an LMP $2$-approximation for $k$-median
and can be converted into an algorithm that opens
$k+O_\varepsilon(1)$ facilities with only a $1+\varepsilon$ loss.

Factor $2$ is a natural barrier for LP-based algorithms. The integrality
gap of the standard LP relaxation is at least $2$~\cite{LS16}.  However, the known instances establishing this lower bound
cease to exhibit a gap of $2$ when the integral solution may open $k+1$
facilities. The reduction of~\cite{LS16} allows algorithms to exploit
a constant number of extra facilities.

\paragraph{Euclidean $k$-median.}
Stronger guarantees are possible when the clients and facilities
lie in a Euclidean space. If the dimension is a fixed constant, the problem admits a
polynomial-time approximation scheme~\cite{ARR98,KR07,CFS19}. If the
dimension is part of the input, the problem is APX-hard. It is NP-hard
to approximate within a factor of $1.07$, or $1.015$ when facilities may
be opened anywhere in the space~\cite{CK19,CKL22}. Under the Johnson
Coverage Hypothesis, these bounds increase to $1.27$ and $1.08$,
respectively~\cite{CKL22}. On the algorithmic side, Ahmadian, Norouzi-Fard, Svensson, and Ward~\cite{ANSW20} gave a $2.633$-approximation by running the
primal-dual algorithm of~\cite{JV01} in a way that respects the
cardinality constraint, avoiding the loss of bi-point rounding.
Cohen-Addad, Esfandiari, Mirrokni, and Narayanan~\cite{CEMN22} improved this to $2.406$. Both
guarantees were later superseded by the $(2+\varepsilon)$-approximation
of~\cite{CGLSS25} for general metrics. The coverage-based reduction
behind the lower bound of $1+2/e$ for general metrics does not carry
over to Euclidean spaces. Before the concurrent work discussed below, it remained open whether
Euclidean $k$-median admits an approximation factor better than
$1+2/e$~\cite{CEMN22}.

\paragraph{Euclidean $k$-means.}
The $k$-means objective replaces each connection distance by its square.
In Euclidean space of fixed dimension, the problem admits a
polynomial-time approximation scheme~\cite{CFS19}, whereas it is
APX-hard when the dimension is part of the input~\cite{CKL22}.
Ahmadian, Norouzi-Fard, Svensson, and Ward~\cite{ANSW20} gave a
primal-dual approximation with factor about $6.357$.
Cohen-Addad, Esfandiari, Mirrokni, and Narayanan~\cite{CEMN22}
improved this to $5.912+\varepsilon$.
For the continuous version, in which facilities may be opened anywhere
in the space, Charikar, Cohen-Addad, Gao, Grandoni, Lee, and van Wijland~\cite{CCGGLW26} obtained a
$(4+\varepsilon)$-approximation via non-monotone dual fitting.
The work of~\cite{ACCGGLSW26} subsequently improved this to
$3+\ln2+\varepsilon$ using spectral analysis. The continuous version admits a standard reduction to the discrete
version: for every fixed $\varepsilon>0$, one can construct in
polynomial time a finite set of candidate facilities such that
restricting facilities to this set increases the optimum by at
most a factor of $1+O(\varepsilon)$~\cite{ANSW20}.
For the discrete setting considered here, the iterative rounding
framework of~\cite{BGHL2026} gives an LMP $\frac{11}{3}$-approximation,
but the loss in its preprocessing limits the final approximation
factor to $4+\varepsilon$.

\paragraph{Hardness of approximation.}
Jain, Mahdian, and Saberi~\cite{JMS02} showed that it is NP-hard to approximate $k$-median and $k$-means within factors better than $1+2/e\approx1.736$ and $1+8/e\approx3.943$, respectively. These bounds hold in the setting above, where the
facilities are given separately from the clients.

\paragraph{Facility location.}
Closely related to $k$-median is uncapacitated facility location (UFL).
UFL replaces the constraint on the number of open facilities
with a location-specific opening cost for each facility. The first
constant-factor approximation algorithm for UFL used LP rounding and was due to
Shmoys, Tardos, and Aardal~\cite{STA97}. For general metrics, the best known approximation ratio is
$1.488$, obtained by Li~\cite{Li13} building on Byrka and
Aardal~\cite{byrka2010optimal}. It is close to the lower bound of $\gamma_{\text{UFL}}$
shown by Guha and Khuller~\cite{GK99}, where $\gamma_{\text{UFL}} \approx 1.463$ is the unique solution to $\gamma = 1 + 2e^{-\gamma}$. For $k$-median the LMP factor
matters most. The greedy algorithm of~\cite{JMMSV03} is an LMP
$2$-approximation, and~\cite{CGLS23} improved this factor slightly. For Euclidean UFL, Lee and Shin~\cite{LS25} strengthened the bifactor approximation of Byrka and Aardal~\cite{byrka2010optimal}, which achieves a $(\gamma,1+2e^{-\gamma})$ guarantee for every $\gamma\ge1.6774$ in general metrics. They showed that, for every fixed $\gamma$ in this range, there exists a positive constant $\varepsilon$, depending only on $\gamma$, such that Euclidean UFL admits a $(\gamma,1+2e^{-\gamma}-\varepsilon)$-bifactor approximation. Here the two coefficients apply to the facility opening and connection costs, respectively. They further incorporated this geometric improvement into Li's framework~\cite{Li13}, obtaining a unifactor approximation that improves Li's ratio of approximately $1.488$ by a small positive constant.

\subsection{Our results}\label{sec:intro-results}

%\ifpaperchanges
%\gptmargin{Restored the previous text in gray for comparison. The blue passages give the revision; old commented drafts are also preserved in the source.}
%\fi
%\gptremovedfile{removed-text/intro-01}

\begin{theorem}\label{thm:intro-main}
There exists a randomized polynomial-time $1.9999$-approximation
algorithm for metric $k$-median: it returns at most $k$ facilities
and its expected cost is at most $1.9999\,\mathrm{OPT}$.
\end{theorem}

The improvement over the $(2+\varepsilon)$-approximations
of~\cite{CGLSS25,BGHL2026} is small. It nevertheless shows that factor
$2$ is not the limit of LP-based rounding, and it identifies the
structure of fractional solutions for which the factor-$2$ analysis
of~\cite{BGHL2026} is tight. Before the
final reduction of~\cite{LS16}, our algorithm is a pseudo-approximation
with respect to the standard LP relaxation. It opens
$k+O(1)$ facilities, and its expected cost is at most
$1.9999$ times the LP value. Consequently, the
integrality gap of the standard relaxation is strictly below $2$ when a
constant number of extra facilities is allowed. This partially answers the question of Li and Svensson
about the integrality gap when $k+O(1)$ facilities may be
opened~\cite{LS16}.
The final combination and its error budget are given in
Section~\ref{sec:final-analysis}.

\begin{theorem}\label{thm:euclidean-kmedian}
    There exists a randomized polynomial-time $1.71955$-approximation algorithm for Euclidean $k$-median.
\end{theorem}

We obtain this result via a more accurate argument about the expected decrease of the potential function. It turns out that the expected change of the potential can be expressed via a quadratic form on distances to facilities serving a client. We believe the argument about the sign of this quadratic form, interpreted as a property of the Euclidean space, can be of independent interest. Then, by a previously known reduction, we also obtain the following result:

\begin{theorem}\label{thm:euclidean-ufl}
    There exists a randomized polynomial-time $1.4561$-approximation algorithm for Euclidean UFL.
\end{theorem}

Finally, we also get the following improved result:

\begin{theorem}\label{thm:euclidean-kmeans}
    For every $p\ge2$ and a sufficiently small $\varepsilon>0$, there exists
a randomized polynomial-time
$\left(\frac{3^p+2}{3}+\varepsilon\right)$-approximation algorithm
for the Euclidean $k$-clustering problem with the cost function being the $p$-th power of the distance.
In particular, taking $p=2$ gives an
$(11/3+\varepsilon \approx 3.667 + \varepsilon)$-approximation for Euclidean $k$-means.
\end{theorem}

The LMP analogue of this result was already shown in~\cite{BGHL2026}, but there the preprocessing step (needed for the version opening exactly $k$ facilities) involved a loss of factor $(2^p)$ for the type-3 clients. The current preprocessing step avoids this problem and allows for a smooth transformation of the LMP algorithm into a true approximation algorithm.

Note that all the approximation ratios in Theorems~\ref{thm:euclidean-kmedian}, \ref{thm:euclidean-ufl} and \ref{thm:euclidean-kmeans} are strictly smaller than the known lower bounds on the approximability of the problems in general metrics, which are $1+2/e \approx 1.736, \gamma_{\text{UFL}}\approx 1.463$ and $1 + 8/e \approx 3.943$ respectively.

\subsection{Further related work}\label{sec:intro-related}

The literature on clustering is very broad. We briefly discuss
some further results most relevant to our work.

\paragraph{Approximation in FPT time.}
Some works allow fixed-parameter tractable (FPT) running time
$f(k,\varepsilon)\cdot\mathrm{poly}(n)$.
% The hardness threshold of $1+2/e$ still holds under the Gap Exponential
% Time Hypothesis (Gap-ETH)~\cite{CGKLL19}.
In this setting, a $(1+2/e+\varepsilon)$-approximation
is achievable for every fixed $\varepsilon>0$~\cite{CGKLL19}. When the facilities coincide with the
clients, the standard NP-hardness bound is $1+1/e$. In this case, Anand and
Lee~\cite{AL24} gave an FPT-time $1.546$-approximation, strictly improving
on $1+2/e$.

\paragraph{Concurrent work.}
% NOTE (authors): decide how to describe the timing of our work relative
% to these results, and whether to add a sentence comparing the factors.
Several works giving better-than-$2$ approximations for metric $k$-median appeared in
October 2026. Buhai, Cohen-Addad, Fleischmann, Gupta, Jub\'es Monforte,
Pittu, Stouras, Svensson, and Vogl gave a randomized polynomial-time
$1.95$-approximation in two papers. The first replaces the linear bids
in the greedy algorithm of~\cite{JMMSV03} with general bidding
functions, which gives a $1.9499$-LMP approximation for facility
location. A smooth variant of this algorithm opens $k+O(\log n/\varepsilon^2)$
facilities and costs at most $(1.9499+\varepsilon)$ times the
optimum~\cite{BCFGJPSSV26-bids}. The second gives a
$(1.89+\varepsilon)$-approximation for stable instances, whose optimal cost increases sufficiently
when one fewer facility is opened~\cite{BCFGJPSSV26-stable}. The two algorithms are combined using the framework of~\cite{CGLSS25}.

In the same month, OpenAI released a collection of mathematical
manuscripts produced by an internal model. Two of them concern metric
$k$-median. The first claims a deterministic
$(1+2/e+\varepsilon)$-approximation~\cite{OAI-threshold}, which would
match the hardness of~\cite{JMS02}. Such an algorithm would also yield a $(\gamma_{\text{UFL}}+\varepsilon \approx 1.463 + \varepsilon)$-approximation for uncapacitated facility location, matching the hardness threshold up to an arbitrarily small additive term.\footnote{A blackbox reduction only gives a $(\gamma_{\text{UFL}}+\varepsilon)$-approximation for the uncapacitated facility location problem with uniform facility costs. However, the rounding algorithm of \cite{OAI-threshold} preserves the marginal probabilities of opening facilities, and thus works for UFL with arbitrary facility costs as well. }
The result also builds on the iterative
rounding framework of~\cite{BGHL2026}, with a modified, distance-tapered
rounding graph and a different per-client potential, and it forces facilities with large deletion sets to open
during batched iterations. The second claims a randomized
$(2-\sigma)$-approximation for an absolute constant
$\sigma>0$~\cite{OAI-recovery} and builds on the approach
of~\cite{CGLSS25}.

\subsection{Overview of the techniques}\label{sec:intro-overview}

\paragraph{The algorithm in a nutshell.}
Our algorithm has three steps.
\begin{enumerate}[itemsep=2pt]
\item Solve the standard LP relaxation. Preprocess its solution so that
  every fractional opening is a multiple of a small constant $\delta$
  and the total opening is at most $k+3$.
\item Round the preprocessed solution in two independent ways. The
  first is a variant of the iterative rounding of~\cite{BGHL2026}. It
  repeatedly opens a randomly chosen fractional facility and closes the
  facilities that relied on it as part of their nearest unit of
  opening. A slight bias in the random choices keeps the number of open
  facilities at $k+O_\varepsilon(1)$. The second is a filtering and
  bundling rounding in the style of Charikar and Li~\cite{CL12}, which
  opens at most $k+3$ facilities.
\item Return the cheaper of the two solutions. This is a
  pseudo-approximation that opens $k+O_\varepsilon(1)$ facilities, and
  the reduction of Li and Svensson~\cite{LS16} turns it into a solution
  with at most $k$ facilities.
\end{enumerate}

The two roundings in Step 2 have complementary strengths. Fix a client
$j$ and consider the distances from $j$ to the facilities that serve it
in the fractional solution. The factor-$2$ analysis of iterative
rounding is tight only when these distances are all about the same.
When they are spread out, the analysis charges $j$ more than it actually
pays, and we show that the gap is proportional to the spread. When the
distances are concentrated, iterative rounding may really lose a factor
of $2$, but then $j$ behaves like an element in maximum coverage. The
second rounding opens one of $j$'s fractional facilities with
probability at least about $1-1/e$. Otherwise it serves $j$ through a
nearby representative at about three times the distance, for a factor
of about $1+2/e$ overall. Both kinds of clients may occur in one
instance, and we cannot choose a rounding separately for each client.
We therefore measure the spread of the whole instance by a single
number and show that the cheaper of the two solutions is always within
a factor strictly below $2$ of the LP value. The rest of this section
explains each part in more detail.

\paragraph{Notation.}
Fix a fractional solution $(\vecx,\vecy)$. After splitting facilities
into colocated copies, every client $j$ is assigned to a set $F_j$ of
facilities with $y(F_j)=1$, namely its nearest unit of opening. Write
$d_i=d(i,j)$ and
\[
  V_j=\sum_{i\in F_j}y_id_i,\qquad
  \Disp_j=\sum_{i,i'\in F_j}y_iy_{i'}|d_i-d_{i'}|,\qquad
  M_j=\sum_{i\in F_j}y_i|d_i-V_j|.
\]
Thus $V_j$ is the fractional connection cost of $j$, $\Disp_j$ is the
pairwise dispersion of its distances, and $M_j$ is their mean absolute
deviation. These quantities satisfy $M_j\le\Disp_j\le2V_j$, and
$\Disp_j=0$ if and only if all facilities in $F_j$ are at the same
distance from $j$.

\paragraph{Preprocessing.}
As in~\cite{BGHL2026}, we
first make the fractional solution $\delta$-integral for a constant
$\delta$ (Theorem~\ref{thm:preprocessing}). Our preprocessing separates
the nearest $1-\varepsilon$ units of each representative client's
opening, its \emph{head}, from the next $\varepsilon$ units, its
\emph{tail}. It replaces tail mass by small atoms and rounds all atoms
by a single pipage rounding.  Heads are rounded depending on the states of their corresponding tails. Roughly speaking, if some facility in a tail is open, then we tend to open fewer facilities in the head. This will avoid the factor of 2 loss in~\cite{BGHL2026}. So far, these techniques only guarantee that the total fractional opening is at most $k$ in expectation. To control the total opening with high probability, we slightly scale up the $y$-values of atoms shared by several tails before rounding, and the resulting saving in the heads keeps the total
opening at most $k+3$ with high probability. The expected connection
cost of every client increases by at most a $1+O(\varepsilon)$ factor.

We briefly explain why the preprocessing step is needed. What we actually need is a uniform positive lower bound on the nonzero assignment values $x_{ij}$; once this holds, $\delta$-integrality of the opening values $y_i$ is not essential. Establishing this lower bound is itself the main difficulty in preprocessing, and our procedure additionally guarantees $\delta$-integrality of the $y_i$'s, a stronger property that may be useful in future applications.
To enforce the facility budget, our later rounding algorithm slightly
biases its random choices, and preprocessing controls the accumulated
effect of this bias on connection costs. Without a lower bound on
positive opening values, a client's fractional support could be removed
in arbitrarily many small steps, allowing a small bias at each step to
produce a large overall loss. The $\delta$-integrality bounds the number
of such steps affecting any client by $1/\delta$, so we can choose the
bias small enough to keep the total cost increase within a
$1+O(\varepsilon)$ factor. 

\paragraph{Algorithm A: iterative rounding and where the factor $2$ is lost.}
We first recall the LMP algorithm of~\cite{BGHL2026}. Every fractional
facility $i$ receives incoming edges from its nearest one unit of
fractional opening, so that every facility has fractional in-degree
$1$. In each iteration, a pivot $q$ is chosen with probability
proportional to its current opening $z_q$ and opened. Each out-neighbor
$i$ of $q$ is closed with probability $w_{qi}/z_q$, where $w_{qi}$ is
the amount of $q$'s opening in the nearest unit of $i$.
This keeps the
expected number of open facilities unchanged. For a client $j$, the
analysis tracks the set $T\subseteq F_j$ of support facilities that are
still fractional, together with the distance $b$ from $j$ to its
nearest integrally open facility. It shows that the potential
\begin{equation} \label{eq:intro:potential}
  \Phi(T,b)=2\sum_{i\in T}y_id_i+\bigl(1-y(T)\bigr)b
\end{equation}
does not increase in expectation, which gives
$\E[d(j,O)]\le2V_j$.

This argument has slack in two places. When an outside pivot closes a
support facility $i$, the new backup distance is bounded by
$2d_i+r_i$. Here $r_i$ is the smallest distance from $j$ to a surviving
support facility that is not entirely contained in the nearest unit of
$i$, and this bound is often loose. When a support
facility $q$ itself is chosen, the surviving support facilities with
$2d_i<d_q$ remain useful. We show that the total slack in an iteration
is at least a quarter of the current pairwise dispersion
(Lemma~\ref{lem:combined-saving}). We also show that the dispersion of
the surviving support decays slowly enough. Together, these facts give
\begin{equation}\label{eq:intro-dispersion}
  \E[d(j,O)]\le2V_j-\frac{\Disp_j}{12}
\end{equation}
(Theorem~\ref{thm:client-dep}). In words, the factor $2$ can be tight only for
clients whose fractional distances are nearly equal.

Bound~\eqref{eq:intro-dispersion} holds for the LMP algorithm, which
opens $y(F) \leq k +O(1)$ facilities only in expectation.
We then convert the LMP algorithm into one that opens
$k+O_\varepsilon(1)$ facilities. The conversion of~\cite{BGHL2026}
selects many facilities in each iteration and forces facilities with
large negative imbalance to open. In contrast, our conversion keeps a
single operation per iteration (Section~\ref{sec:single-pivot-conversion}).
Each fractional copy carries one token and each open facility carries
$1/\delta$ tokens. An operation chooses a token, opens its owner, and
closes the fractional copies whose nearest $1/\delta$ tokens include it.
Operations that would close too much opening mass are split into a few
smaller operations. To compensate for the resulting drift, we slightly
increase the sampling probability of operations that decrease the
total opening. An exponential function of the total opening is then a
supermartingale, and the number of open facilities exceeds the initial
opening by more than $O_\varepsilon(1)$ only with small probability.

Splitting and biased sampling perturb the cost analysis. We isolate a
sufficient condition, a \emph{local-charge certificate}, under which a
potential argument survives the conversion with only a $1+\varepsilon$
loss (Theorem~\ref{thm:conversion-local-charges}). For every token,
the potential change of every partial operation must be bounded by an
opening charge plus deletion charges for the closed tracked facilities.
The opening charge is counted once per fragment and each deletion
charge exactly once, and the total of all charges must be nonpositive.
The main difficulty is that the dispersion saving does not
fit this framework directly. Splitting an operation can count the loss
of the same pairwise dispersion credit in several fragments. We
therefore use the potential
\[
  P_j(T,B)=\min_{U\subseteq T}\left\{\Phi(U,B)-\frac{\Disp(U)}{12}\right\},
\]
which may discard part of the surviving support. When a support
facility opens, the potential moves to a subset that excludes it. This
retains dispersion and makes the opening charge nonpositive, so
repeating it over fragments is harmless. When a support facility $i$ is
closed, the lost dispersion is charged to an endpoint credit of $i$
proportional to $\sum_{v\in U}|d_i-d_v|$. A pair with both endpoints
closed is then overcharged, never undercharged. The endpoint credits
total $\Disp(U)/6$, and together with the remaining terms they are paid
for by the saving of $\Disp(U)/4$ from
Lemma~\ref{lem:combined-saving}. Hence the weighted balance holds
(Lemma~\ref{lem:conversion-disp-cert}). Consequently,
the converted algorithm opens $k+O_\varepsilon(1)$ facilities and
satisfies $\E[d(j,O)]\le(1+\varepsilon)(2V_j-\Disp_j/12)$ for every
client $j$.

\paragraph{Algorithm B: rounding when distances are concentrated.}
For these clients, we propose the second algorithm.
For every client $j \in C$ and a radius $R$, let $T_j(R) := \sum_{i \in F_j : d_i > R} y_i d_i$ denote the fractional connection cost restricted to the facilities at distance more than $R$ from $j$.
We then choose for each client $j \in C$ a radius $R_j$ that minimizes $R + 5e \, T_j(R)$ among the radii whose ball contains at least half of its fractional opening.
Let $A_j := R_j + 5e \, T_j(R_j)$ be the value of this objective function for client $j$.
This value exhibits a dispersion-dependent bound $A_j \leq V_j+\sqrt{10eV_jM_j}+\frac{5e}2M_j$ (\Cref{lem:assurance-bound2}).

The second rounding algorithm then follows filtering and bundling approaches inspired by Charikar and Li~\cite{CL12}.
It first selects representatives from the clients in the increasing order of $A_j$, and greedily matches these representatives in the increasing order of their distances.
Each representative receives a bundle of facilities within half the distance to its nearest other representative.
The dependent rounding scheme of~\cite{srinivasan2001distributions} then applies here to guarantee that the marginals are preserved with negative correlation introduced, and that at least one facility in each matched pair of bundles is open while the total number of open facilities in the rounded solution $O \subseteq F$ is (almost) identical to the total opening from the input fractional solution.
This leads to a bound $\E[d(j,O)]\le(1-1/e)V_j+(3/e)A_j$ (\Cref{lem:robust}).
Together with the above bound on $A_j$, we derive that $|O| \leq k + 3$ and
\begin{equation}\label{eq:intro-robust}
  \E[d(j,O)]\le\left(1+\frac2e\right)V_j
  +\sqrt{\frac{90V_jM_j}{e}}+\frac{15}{2}M_j 
  \text{ for every $j \in C$}
  \quad
  \text{(\Cref{thm:robust})}.
\end{equation}

\paragraph{Combining the two algorithms.}
The complementarity must be exploited at the level of complete
solutions, since one cannot choose a rounding procedure separately for
each client. We run both algorithms on the same fractional solution and
return the cheaper solution. Let $V=\sum_jV_j$, $\Disp=\sum_j\Disp_j$,
and $s=\Disp/V\in[0,2]$. Since $\sum_jM_j\le\Disp$, the Cauchy--Schwarz
inequality gives $\sum_j\sqrt{V_jM_j}\le\sqrt{V\Disp}$. Summing
\eqref{eq:intro-dispersion} and \eqref{eq:intro-robust} over clients
shows that the expected cost of the cheaper solution is at most
\[
  V\cdot\min\left\{2-\frac s{12},\;
  1+\frac2e+\sqrt{\frac{90s}{e}}+\frac{15}{2}s\right\}.
\]
The first term decreases in $s$, and the second increases from
$1+2/e<2$. The two terms cross at $s=s^\star\approx0.0018867$, the
positive solution of $\frac{91}{12}s+\sqrt{90s/e}=1-\frac2e$. Hence the
expression is at most $\alpha^\star V$, where
$\alpha^\star:=2-s^\star/12\approx1.9998428$.

\paragraph{From $k+O(1)$ facilities to $k$ facilities.}
The robust algorithm runs on the same preprocessed solution and opens at
most $k+3$ facilities. Returning the cheaper of the two solutions gives
a pseudo-approximation that opens $k+O_\varepsilon(1)$ facilities and
has expected cost at most $(\alpha^\star+O(\varepsilon))$ times the LP
value. Finally, the reduction of Li and Svensson~\cite{LS16} turns it
into an algorithm that opens at most $k$ facilities, which proves
Theorem~\ref{thm:intro-main}, since $\alpha^\star<1.9999$.

\paragraph{An improved LMP guarantee for Euclidean $k$-median.}
% \paragraph{An improved LMP guarantee in Euclidean space.}
Our result for Euclidean $k$-median uses the same iterative randomized rounding algorithm, but a different analysis that exploits the geometry of the facilities serving each client.

Fix a client $j\in C$, and adapt the potential defined in \eqref{eq:intro:potential} to a target approximation ratio $\alpha$ as follows:
\[
    \Phi(T,b)=\alpha\sum_{i\in T}y_i d_i+(1-y(T))b,
\]
where $T\subseteq F_j$ is the set of fractionally open facilities supporting $j$ tracked by the analysis, and $b$ is the backup distance from $j$ to its nearest integrally open facility. The analysis discards facilities $i\in T$ whenever $\alpha d_i>b$, since replacing their contribution by the backup only decreases the potential.

We show that a sufficient condition for this potential to be a supermartingale is the following pairwise inequality:
\[
\sum_{i,i'\in T}y_i y_{i'}
\bigl(d_i+d(i,i')-\alpha d_i\bigr)_+
\le
\sum_{i,i'\in T}y_i y_{i'}
(\alpha d_i-d_{i'})_+
\quad \text{(\Cref{lem:nonstop-lmp})}.
\]
The two sides represent competing effects of rounding.
Intuitively, the left-hand side bounds the positive contribution associated with the support facilities to the change in potential.
When a pivot $q\in F$ closes a support facility $i\in T$, the mass served by $i$ may move to a backup whose distance from $j$ is at most $d(j,q)$.
By the construction of $\{w_{qi}\}_{q\in F}$ using the nearest-unit rule, we can match the pivot mass $q$ that closes $i$ to the mass from $T$ and the remaining mass at the existing backup for the client $j$. Whenever mass at $q$ is matched to support mass at $i'$, we have $d(i,q)\le d(i,i')$. 
The triangle inequality then bounds the new backup distance by $d_i+d(i,i')$, leading to the left-hand side.
On the other hand, the right-hand side represents the decrease in potential when a support facility $i'\in T$ opens: the analysis moves mass served by facilities $i\in T$ with $\alpha d_i>d_{i'}$ to the new backup $i'$. Terms involving the existing backup are handled separately.

Crucially, the inequality need only hold after summing over all pairs. A general metric allows many pairs $(i,i')\in T\times T$ to attain $d(i,i')=d_i+d_{i'}$ simultaneously, as in a star metric.
On the contrary, Euclidean geometry imposes angular constraints that prevent this equality from holding for every pair among three or more nonzero vectors from $j$.
The angular constraints allow us to establish the pairwise inequality for a target approximation ratio $\alpha$ below the general-metric hardness threshold $1+2/e$.

Symmetrizing the difference between the two sides gives a kernel $K$ for which the desired inequality is
\begin{equation}\label{eq:intro:kernel-goal}
    \sum_{i,i'\in T}y_i y_{i'}K(d_i,d_{i'},d(i,i'))\le0.
\end{equation}
To prove this inequality, we construct a nonnegativity certificate from two positive-semidefinite kernels $A$ and $B$. Specifically, for $r,s>0$ and $c\ge 0$ satisfying the triangle inequalities, namely $|r-s|\le c\le r+s$ , we establish
\begin{equation}\label{eq:intro:kernel-ineq}
    K(r,s,c)\le-J(r,s,c),
    \qquad\text{where}\qquad
    J(r,s,c):=\sqrt{rs}\bigl(A(r,s)\cos\theta+B(r,s)\bigr),
\end{equation}
where $\theta:=\frac{r^2+s^2-c^2}{2rs}$. In particular, for $i,i'\in T$, $\theta$ is the angle between the vectors from $j$ to $i$ and $i'$, where $r=d_i,s=d_{i'}$, and $c=d(i,i')$. Pairs with a zero distance to $j$ are handled separately.

The matrix of angular terms $\cos\theta$ is the Gram matrix of the normalized vectors from $j$ to the facilities. Its entrywise product with the matrix defined by $A$ is positive semidefinite by the Schur product theorem, and adding the matrix defined by $B$ preserves the positive semidefiniteness. Thus, the weighted sum of $J(d_i,d_{i'},d(i,i'))$, with weights $y_i y_{i'}$, is nonnegative. By \eqref{eq:intro:kernel-ineq}, the corresponding sum involving $K$ is nonpositive, yielding \eqref{eq:intro:kernel-goal} after accounting for the pairs with $d_id_{i'}=0$.

Our construction begins by establishing a lower bound on the approximation factor $\alpha$ achievable by a certificate of the form in \eqref{eq:intro:kernel-ineq}. This analysis combines the positive semidefiniteness of $A$ and $B$ with four geometric configurations to derive necessary conditions leading to the lower bound. We then construct a certificate attaining this bound using scaled exponential and triangular-wave kernels applied to logarithms of the distances from $j$. Within this chosen setup, the conditions from the four configurations, together with the behavior near the orthogonal configuration, force the coefficients and parameters for the kernels. Finally, in \Cref{lem:median-certificate}, we verify that this certificate is valid for every feasible $(r,s,c)$ at the resulting value $\alpha\le1.71955$. We also provide a Lean 4 formalization to support verification of this lemma.

\subsection{Organization of the paper}
The rest of the paper is organized as follows.
Section~\ref{sec:preliminaries} introduces the LP relaxation, pipage rounding, and concentration tools used throughout the paper.
Section~\ref{sec:preprocessing} develops the preprocessing procedure that converts the LP solution into a $\textrm{poly}(\varepsilon)$-integral solution.
Section~\ref{sec:client-dep} establishes the client-dependent LMP-guarantee for the iterative rounding algorithm for general metrics.
Section~\ref{sec:euclidean-lmp} analyzes the LMP guarantee for Euclidean spaces and proves Theorem~\ref{thm:euclidean-ufl}. 

Section~\ref{sec:single-pivot-conversion} gives a general framework converting the iterative rounding algorithm into one that opens at most $k+O(1)$ facilities, with only a $1  + O(\varepsilon)$ loss in the guarantees established in Sections~\ref{sec:client-dep} and \ref{sec:euclidean-lmp}. This altogether finishes the proof of Theorems~\ref{thm:euclidean-kmedian} and \ref{thm:euclidean-kmeans}.
To finish the proof of the 1.9999-approximation for general metrics in Theorem~\ref{thm:intro-main}, we present in
Section~\ref{sec:robust-algorithm} the complementary robust rounding
algorithm, and wrap up the analysis in Section~\ref{sec:final-analysis} by combining the two rounding algorithms.

% \paragraph{Organization} The rest of the paper is organized as follows.
% Section~\ref{sec:preliminaries} introduces the LP relaxation and the
% rounding and concentration tools used throughout the paper.
% Section~\ref{sec:preprocessing} develops the preprocessing procedure.
% Section~\ref{sec:client-dep} establishes the client-dependent bound for
% single-pivot rounding, and Section~\ref{sec:euclidean-lmp} analyzes its
% LMP guarantees in Euclidean space.
% Section~\ref{sec:single-pivot-conversion} converts these rounding guarantees
% into algorithms that open at most $k+O(1)$ facilities.
% Section~\ref{sec:robust-algorithm} presents the complementary robust rounding
% algorithm, and Section~\ref{sec:final-analysis} combines the two roundings
% to obtain an approximation factor below $2$ for metric $k$-median.

\section{Preliminaries}\label{sec:preliminaries}

\paragraph{LP relaxation.}
Let $C$ be the set of clients, $F$ the set of facilities, and $d$ a metric
on $C\cup F$. Given an integer $k\ge1$ and a constant $p\ge1$, the
$k$-clustering problem asks for at most $k$ open facilities minimizing
the sum of the $p$-th powers of the clients' connection distances.
We use the standard linear programming relaxation
\begin{equation}\label{lp:k-clustering}
\begin{gathered}
\min\quad \sum_{j\in C}\sum_{i\in F}d(i,j)^p x_{i,j}\quad \text{subject to}\\
\begin{aligned}
\sum_{i\in F}x_{i,j}&=1 &&\forall j\in C,
\qquad \qquad &x_{i,j}&\ge0 &&\forall i\in F,\ j\in C,\\
x_{i,j}&\le y_i &&\forall i\in F,\ j\in C,
\qquad\qquad &0&\le y_i\le1 &&\forall i\in F,\\
\sum_{i\in F}y_i&\le k.
\end{aligned}
\end{gathered}
\end{equation}
Here $y_i$ is the opening of facility $i$, and $x_{i,j}$ is the amount
of client $j$ assigned to $i$. For a fixed feasible opening vector
$\vecy$, we may assume that each client is assigned to its nearest one
unit of opening, breaking ties consistently. This minimizes the cost
for that vector. We denote the resulting cost of client $j$ by
$\cost_{\vecy}(j)$. For a vector $\vecy$ indexed by a finite set $U$
and $L\subseteq U$, write $y(L):=\sum_{i\in L}y_i$.

\paragraph{Pipage rounding with a laminar family.}
A family $\mathcal L\subseteq2^U$ is \emph{laminar} if any two of its
sets are disjoint or one contains the other. By applying pipage
rounding to a nonnegative vector $\vecy$ indexed by $U$, with laminar
family $\mathcal L$ and granularity $\delta>0$, we mean the following
dependent-rounding procedure~\cite{srinivasan2001distributions}.
We may add $U$ and all singletons to $\mathcal L$ without destroying
laminarity. Write
\[
\frac{y_i}{\delta}=n_i+f_i,
\qquad n_i:=\left\lfloor\frac{y_i}{\delta}\right\rfloor,
\qquad 0\le f_i<1.
\]
Keep the integer parts fixed. Repeatedly choose an inclusion-minimal
set of $\mathcal L$ containing at least two fractional coordinates,
and move mass between two of them, preserving their sum, until at
least one becomes integral. Choose between the two possible directions
randomly so that each coordinate's conditional expectation is
unchanged. If a single fractional coordinate remains, round it to
$0$ or $1$ with its current value as the probability of $1$.
Let $X_i\in\{0,1\}$ be the final rounded fractional parts, and return
$\widehat y_i:=\delta(n_i+X_i)$.

This procedure preserves expectations and rounds every laminar-set
total to one of its two neighboring multiples of $\delta$:
\begin{align*}
\E[\widehat y_i]&=y_i,
&\widehat y_i&\in
\left\{\delta\left\lfloor\frac{y_i}{\delta}\right\rfloor,
\delta\left\lceil\frac{y_i}{\delta}\right\rceil\right\}
&&\text{for every }i\in U,\\
\widehat y(L)&\in
\left\{\delta\left\lfloor\frac{y(L)}{\delta}\right\rfloor,
\delta\left\lceil\frac{y(L)}{\delta}\right\rceil\right\}
&&&&\text{for every }L\in\mathcal L.
\end{align*}
In particular, if $y(L)$ is a multiple of $\delta$, then
$\widehat y(L)=y(L)$ in every outcome. The rounded fractional parts
also satisfy, for every $S\subseteq U$,
\begin{equation} \label{eq:neg-corr}   
\E\!\left[\prod_{i\in S}X_i\right]\le\prod_{i\in S}f_i,
\qquad
\E\!\left[\prod_{i\in S}(1-X_i)\right]
\le\prod_{i\in S}(1-f_i), \qquad \Var\left(\sum_{i\in S} X_i\right) \le \sum_{i\in S} \Var(X_i).
\end{equation}

\paragraph{Probabilistic Preliminaries.}
\begin{theorem}[Bernstein's Inequality]\label{thm:bernstein}
    Let $X_1,\ldots,X_n$ be independent random variables such that $|X_i-\E[X_i]|\le M$ almost surely, for all $i\in[n]$. For any $t\ge0$ and $Y=\sum_{i=1}^n X_i$, \[
        \Pr\left[Y\le \E\left[Y\right]-t\right] \le \exp\left(-\frac{t^2/2}{\sum_{i=1}^n \Var(X_i) +Mt/3}\right).
    \]
\end{theorem}

Note that the standard Bernstein inequality for independent random variables cannot be directly applied to our pipage (dependent) rounding. The following lemma gives a Bernstein-like inequality for our pipage rounding with granularity $\delta$.

\begin{lemma}[Weighted pipage concentration]\label{lem:pipage-concentration}
Let $\widehat y$ be the output of the pipage rounding above on input
$y$, let $a_i\ge0$ be deterministic, and put
\[
 Z=\sum_i a_i\widehat y_i,\qquad
 v=\delta\sum_i a_i^2y_i,\qquad M=\delta\max_i a_i.
\]
Then $\Var(Z)\le v$. For every $t>0$, each of
$\Pr[Z\le\E Z-t]$ and $\Pr[Z\ge\E Z+t]$ is at most
\[
 \exp\left(-\frac{t^2}{2(v+Mt/3)}\right).
\]
If $v=0$, the variable $Z$ is deterministic and both probabilities vanish.
\end{lemma}
\begin{proof}
Let $\widetilde X_i$ be independent Bernoulli variables with means $f_i$,
and set $\widetilde Z=\delta\sum_i a_i(n_i+\widetilde X_i)$.
For $\lambda\ge0$, expand
$e^{\lambda\delta a_iX_i}=1+(e^{\lambda\delta a_i}-1)X_i$.
For $\lambda<0$, instead expand
$e^{\lambda\delta a_iX_i}=e^{\lambda\delta a_i}
+(1-e^{\lambda\delta a_i})(1-X_i)$.
All coefficients in either product expansion are nonnegative, so
\eqref{eq:neg-corr} implies
$\E[e^{\lambda Z}]\le\E[e^{\lambda\widetilde Z}]$ for every real
$\lambda$. Pairwise negative correlation also gives
\[
 \Var(Z)\le\delta^2\sum_i a_i^2 f_i(1-f_i)
 \le\delta\sum_i a_i^2 y_i=v.
\]
The same variance bound holds for $\widetilde Z$. Each centered
independent summand $W_i=\delta a_i(\widetilde X_i-f_i)$ satisfies
$|W_i|\le M$. For $0<\lambda<3/M$, expansion of its exponential,
using $n!\ge2\cdot3^{n-2}$ for $n\ge2$, gives
\[
 \E[e^{\pm\lambda W_i}]
 \le\exp\left(\frac{\lambda^2\E[W_i^2]}{2(1-\lambda M/3)}\right).
\]
Multiply these bounds and use the moment comparison. Markov's
inequality with $\lambda=t/(v+Mt/3)$ proves both tails when $v>0$.
\end{proof}

\begin{lemma}\label{lem:capped-variance}
For any real random variable $V$ with finite variance and any real number $\kappa$,
\[
\mathrm{Var}(\min(V,\kappa))\le\mathrm{Var}(V).
\]
\end{lemma}

\begin{proof}
Let $V_1:=\min(V,\kappa)$ and $V_2:=\max(0,V-\kappa)$. Then $V=V_1+V_2$, and $V_1V_2=\kappa V_2$. Consequently, we have
\[
\begin{aligned}
\mathrm{Var}(V)
&=\mathrm{Var}(V_1)+\mathrm{Var}(V_2)+2\mathrm{Cov}(V_1,V_2)\\
&=\mathrm{Var}(V_1)+\mathrm{Var}(V_2)+2\left(\mathbb E[V_1V_2]-\mathbb E[V_1]\mathbb E[V_2]\right)\\
&=\mathrm{Var}(V_1)+\mathrm{Var}(V_2)+2(\kappa-\mathbb E[V_1])\mathbb E[V_2]
\ge\mathrm{Var}(V_1)=\mathrm{Var}(\min(V,\kappa)),
\end{aligned}
\]
where the last inequality uses $V_1\le\kappa$ and $V_2\ge 0$.
\end{proof}

\section{Preprocessing}
\label{sec:preprocessing}

Fix a constant $p\ge1$, and let $\ell:=\lceil p\rceil$. Throughout, $0<\varepsilon<1/p$, and all constants hidden in asymptotic notation are absolute and independent of $p$. Let $(\vecx,\vecy)$ be a solution to the LP~\eqref{lp:k-clustering}. As in the preliminaries, assume that each client $j\in C$ is connected to the nearest one unit of open facilities, and let $\cost_{\vecy}(j)$ denote its connection cost with respect to $d(i,j)^p$.
Formally, let $F_j$ be a set of nearest one unit of open facilities, i.e., we have % \pvmargin{Added the definition of $F_j$.} 
\[
    d(j,i)\le d(j,i') \quad \forall i\in F_j, i'\in F\setminus F_j \quad \text{and} \quad y(F_j)=1.
\]
Then $\cost_{\mathbf y}(j)=\sum_{i\in F_j} y_id(j,i)^p$.
Throughout this section, let $\theta:=\varepsilon^{2\ell+1}$ be the scaling parameter used in Step~\ref{step:pre-scale} of \Cref{alg:preprocessing}.

\begin{theorem}\label{thm:preprocessing}
Let $0<\varepsilon<1/p$, where
$\varepsilon\le 1/8$, suppose
$1/(2\varepsilon)\in\mathbb N$, and put $\delta=\varepsilon^{14\ell+2}$.
There is an $n^{O(1)}\varepsilon^{-O(p)}$-time randomized algorithm
that, given a feasible solution $(\vecx,\vecy)$ to
LP~\eqref{lp:k-clustering}, outputs $\vecy'\in[0,1]^F$ such that:
\begin{itemize}
\item every $y'_i$ is an integer multiple of $\delta$;
\item $y'(F)\le k+3$ with probability at least
$1-\exp(-\Omega(\varepsilon^{-10\ell}))$;
\item for every client $j$, the nearest unit of $y'$ exists and
$\E[\cost_{\vecy'}(j)]\le(1+O(p\varepsilon))\cost_{\vecy}(j)$.
\end{itemize}
The cost bound is unconditional; the constants hidden in $O$ and
$\Omega$ are absolute and independent of $p$.
\end{theorem}

% \gptremovedfile{removed-text/preprocessing-01}

\paragraph{Overview.}
For each client $j$, we split its nearest unit of opening into a
\emph{head}, the nearest $1-\varepsilon$ units, and a \emph{tail}, the
next $\varepsilon$ units. Filtering selects representatives
$\mathcal R$ whose heads are pairwise disjoint and well separated, so
the heads can be rounded independently of one another. Tails, however,
may overlap each other and the heads of other representatives. We
therefore replace tail mass by tiny \emph{atoms}. Physical atoms are
pieces of facilities outside all representative heads. Virtual atoms
replace the part of a tail that lies inside another representative's
head; they later serve as bookkeeping and are not part of the output.
All atoms are rounded by a single pipage rounding, and afterwards each
head is rescaled so that it forms one unit together with its rounded
tail.
Opening an atom that lies in $d\ge2$ tails serves all $d$ tails. We
exploit this by scaling up such atoms, and all virtual atoms, by a
factor $1+\theta$ before rounding. Scaling up an atom $e$ increases its
own opening by $\theta y_e$ but decreases the incident heads by
$\theta d y_e$ in total. The net saving pays for the rounding errors
and keeps the total opening at most $k+O(1)$. Atoms that lie in too many
tails are first duplicated, so that each atom has a bounded influence on
this saving, as needed for concentration.

\paragraph{Parameters.}
The parameters $\theta=\varepsilon^{2\ell+1}$ and
$\delta=\varepsilon^{14\ell+2}$ are chosen so that the following hold.
(i) $\theta^2/\delta=\varepsilon^{-10\ell}$, which gives the
concentration in Lemma~\ref{lem:wa}.
(ii) $\sqrt\delta=\varepsilon^{7\ell+1}\le\varepsilon\theta$, so the
fluctuation term in Lemma~\ref{lem:local-deficit} is dominated by its
bias term $O(\varepsilon\theta)$.
(iii) In Lemma~\ref{lem:connection-cost}, the additive loss is
$O(\varepsilon\theta)t_0^p$, where
$t_0^p\le(1+O(\varepsilon))^p\varepsilon^{-2p}R_j^p
\le (1+O(p\varepsilon))\varepsilon^{-2p-1}\cost_{\vecy}(j)$. Since $\ell\ge p$, this
loss is at most $\varepsilon(1+O(p\varepsilon))\cost_{\vecy}(j)$.

\subsection{Algorithm}

% \gptremovedfile{removed-text/preprocessing-02}

Algorithm~\ref{alg:preprocessing} summarizes the preprocessing procedure.
The heads, tails, atoms, and rounding inputs used in the pseudocode are
defined as we walk through the algorithm. 
\begin{algorithm}[htbp]
% \begin{gptadded}
\caption{Preprocessing the fractional opening vector}
\label{alg:preprocessing}
\begin{algorithmic}[1]
	  \Require{A solution $(\vecx, \vecy)$ to LP~\eqref{lp:k-clustering}, and sufficiently small constant $\varepsilon > 0$ with $1/(2\varepsilon)\in\mathbb N$}
	  \Ensure{A $\delta$-integral solution $\vecy' \in [0, 1]^F$, with $\delta := \varepsilon^{14\ell+2}$}
    \State\label{step:pre-heads} Split facilities into colocated pieces to construct $H_j,T_j$ and
    $R_j$ for every $j\in C$.
    \State\label{step:pre-filter} Run Algorithm~\ref{alg:filtering-quantile-radii} to obtain
    $\mathcal R$ and the representatives $r(j)$ for all $j\in C$.
    \State\label{step:pre-physical} For every $i\in F\setminus\bigcup_{r\in\mathcal R}H_r$,
    split $i$ into physical atoms and update its tail memberships.
    \State\label{step:pre-virtual} For every $i\in (\bigcup_{r\in\mathcal R}H_r)\cap(\bigcup_{r\in\mathcal R} T_r)$, split $i$ into virtual atoms and update its head and tail memberships
    \State\label{step:pre-influence} Replace high-influence atoms by low-influence copies,
    assigning each incident tail to one copy.
    \State\label{step:pre-partition} Partition the atoms (denoted as $\mathcal E$) into $\mathcal E_0,\mathcal E_1,
    \mathcal E_{\mathrm{rest}}$, and set
    $\mathcal E_1(r)\gets\mathcal E_1\cap T_r$.
    \State\label{step:pre-scale} Set $\bar y_e\gets(1+\theta)y_e$ for
    $e\in\mathcal E_{\mathrm{rest}}$, and $\bar y_e\gets y_e$ for all other atoms.
    \State\label{step:pre-round-atoms} Apply pipage rounding to $\bar{\vecy}_{|\mathcal E}$ with
    laminar family $\mathcal L:=\{\mathcal E_0,\mathcal E_1,\mathcal E_{\mathrm{rest}}\}
    \cup\{\mathcal E_1(r):r\in\mathcal R\}
    \cup\{\{e\}:e\in\mathcal E\}$ and granularity $\delta$, obtaining
    $\vecy'_{|\mathcal E}$.
    \For{every $r\in\mathcal R$}\label{step:pre-round-heads}
        \State\label{step:pre-head-target} Set $Y_r\gets1-\min\{1/2,y'(T_r)\}$.
        \State Apply pipage rounding to $\frac{Y_r}{1-\varepsilon}\vecy_{|H_r}$
        with laminar family $\{H_r\}\cup\{\{i\}:i\in H_r\}$ and granularity
        $\delta$, obtaining $\vecy'_{|H_r}$.
    \EndFor
    \State\label{step:pre-return} \Return the vector on original facilities obtained by aggregating
    the opening in heads and physical atoms, capping each total at $1$.
    Virtual atoms do not contribute to this vector.
\end{algorithmic}
% \end{gptadded}
\end{algorithm}

\paragraph{Step~\ref{step:pre-heads}: Heads and Tails}
For $j\in C$, let $H_j$ be the set of the closest $1-\varepsilon$ fraction of open facilities to $j$, and $T_j$ be the set of the next closest $\varepsilon$ fraction of open facilities to $j$. $H_j$ and $T_j$ are called the head and tail of $j$. By fixing tie-breaking and splitting facilities into colocated pieces, we assume that all heads and tails consist of whole pieces; we update $F$ and $(F_j)_{j\in C}$ to denote this refined set. 
More formally, after splitting for every $j \in C$, we have $H_j\cup T_j=F_j$, $y(H_j) = 1 - \varepsilon$ and $y(T_j) = \varepsilon$. For every $i \in H_j$ and $i' \in T_j$, we have $d(j, i) \leq d(j, i')$. For every $i \in T_j$ and  $i' \in F \setminus F_j$, we have $d(j, i) \leq d(j, i')$. We write $y^{\mathrm{orig}}$ for the opening of these facility pieces; it stays fixed, while the steps below further split pieces into atoms. Write $T_j^\mathrm{orig}$ for the tail at this point and keep $F$ and $y^{\mathrm{orig}}$ fixed. The changing atom sets below have their own indices; each atom records its original facility and location. The vector returned in Step~\ref{step:pre-return} is again indexed by the original facilities, as in \Cref{thm:preprocessing}.

\paragraph{Step~\ref{step:pre-filter}: Filtering}
Define $R_j:=\max_{i\in H_j}d(i,j)$. We use \Cref{alg:filtering-quantile-radii} to select a representative set $\mathcal R$.
\begin{algorithm}[H]
\caption{Filtering by quantile radii}
\label{alg:filtering-quantile-radii}
\begin{algorithmic}[1]
    \State $C' \gets C,\ \mathcal R \gets \varnothing$;
    \While{$C' \neq \varnothing$}
        \State choose $r \in C'$ with minimum $R_r$;
        \State $\mathcal R \gets \mathcal R \cup \{r\}$;
        \State assign to $r$ and remove from $C'$ every client $j$
        satisfying $d(j,r) \leq \varepsilon^{-1} R_j$;
    \EndWhile
    \State \Return $\mathcal R$ and the representative $r(j)$ assigned to each $j\in C$;
\end{algorithmic}
\end{algorithm}

For any client $j\in C$, let $r(j)$ be the representative that removes it in the algorithm.

\begin{observation}\label{obs:filtering}
    For any client $j\in C$ and $r=r(j)$, we have 
    \[R_j^p\le\varepsilon^{-1}\cost_{\vecy}(j),\qquad R_r\le R_j,\qquad d(j,r)\le\varepsilon^{-1}R_j.\]
\end{observation}
\begin{proof}
    The last $\varepsilon$ units of the nearest one unit of opening to $j$ have distance at least $R_j$. In other words, \( \cost_{\mathbf y}(j)\ge \sum_{i\in T_j} d(i,j)^py_i\ge \varepsilon R_j^p \).
    The inequalities $R_r\le R_j$ and $d(j,r)\le\varepsilon^{-1}R_j$ follow from the filtering condition.
\end{proof}

\begin{observation}[Disjoint Heads]\label{obs:disjoint_heads}
    The heads $\{H_r: r\in \calR\}$ are pairwise disjoint.
\end{observation}
\begin{proof}
    Let $r$ and $r'$ be distinct representatives in $\mathcal R$ such that $r'$ is chosen after $r$ by \Cref{alg:filtering-quantile-radii}. Then $R_{r}\le R_{r'}$, and since $r'$ is not removed when $r$ is selected, $d(r,r')>\varepsilon^{-1} R_{r'}$.
    Assume toward contradiction that there exists $i\in H_r\cap H_{r'}$. By triangle inequality, \[
        d(r,r')\le d(r,i)+d(i,r')\le R_r+R_{r'}\le 2R_{r'}.
    \]
    However, from the statement of \Cref{thm:preprocessing}, we have $1/(2\varepsilon)\ge 1$. Hence, $d(r,r')\le \varepsilon^{-1}R_{r'}$, which is a contradiction. Therefore, $H_r\cap H_{r'}=\varnothing$.
\end{proof}
    
\paragraph{Steps~\ref{step:pre-physical} and~\ref{step:pre-virtual}: Atomization} Then we split some facilities into \textit{atoms}, which are facilities with tiny $y$ values. The purpose of this process is to ensure that \(y_i\) is sufficiently small for all \(i \in F\) for the pipage rounding with granularity $\delta$. We construct them as follows:
(Figure~\ref{fig:atoms} in Appendix~\ref{app:atoms} illustrates the physical and virtual atoms introduced in these steps.)
\begin{itemize}
    \item \textbf{Physical Atoms:} For any facility $i\in F$ that does not belong to any $H_r$ for $r\in \calR$, we split $i$ into $m$ colocated \emph{physical atoms} $\{i_1,\ldots,i_m\}$ for some $m$ such that $y_{i_1}+\cdots+y_{i_m}=y_i$ and $y_{i_o}\le \frac{\delta}{1+\theta}$ for every $o\in [m]$.
    Note that the facility $i$ is replaced by these $m$ atoms, and they inherit its tail membership.
    % \item For any facility $i\in F$ that does not belong to any $H_r$ for $r\in \mathcal R$, we divide $i$ into $m$ colocated atoms $\{i_1,\ldots,i_m\}$ for some $m$, such that $y_{i_1}+\cdots+y_{i_m}=y_i$ and $y_{i_o}\le \frac{\delta}{1+\theta}$ for every $o \in [m]$. We also replace $i$ in $F$ and any $T_{r'}$ with the atoms $\{i_1,\ldots,i_m\}$. The atoms created in this case are called \textit{physical atoms} since they are split from a physical facility $i$.

    \item \textbf{Virtual Atoms:} %\pvmargin{Change in virtual atoms} 
    For $i\in F$ that belongs to some $T_r\cap H_{r'}$ for some representatives $r\ne r'$, we split each facility $i$ into $m$ colocated \emph{virtual atoms} such that $y_{i_1}+\cdots+y_{i_{m}}=y_i$ and $y_{i_o}\le \frac{\delta}{1+\theta}$ for every $o\in [m]$. Note that each facility $i$ is replaced by these $m$ atoms, and they inherit its head and tail membership.

    % \item For any pair of representatives $r,r'\in \mathcal R$ with $T_r\cap H_{r'}\ne\varnothing$, we create $m$ atoms $\{i_1,\ldots,i_m\}$ \del{collocated}\add{colocated} with $r'$ for some $m$, such that $y_{i_1}+\cdots+y_{i_m}=y(T_r\cap H_{r'})$ and $y_{i_o}\le \frac{\delta}{1+\add{\theta}}$ for every $o \in [m]$.  \add{Then, we remove $T_r \cap H_{r'}$ from $T_r$, and add the $m$ atoms to $T_r$ and $F$. This does not change $y(T_r)$. The atoms created in this case are called \textit{virtual atoms}.} We remark that virtual atoms will not contribute to the final opening.
\end{itemize}

The key distinction between physical and virtual atoms is that the opening of a virtual atom can be reassigned to a facility in its associated head, whereas the opening of a physical atom cannot. Specifically, the algorithm performs pipage rounding in two stages: first on the tails and then on the heads. Since virtual atoms belong to both components, even if a virtual atom in $T_r \cap H_{r'}$ is opened during the first stage on the tails, we can instead use any facility in $H_{r'}$ by reallocating to that facility. Note that, when performing pipage rounding on the heads, we may merge the virtual atoms back into their original facilities. Consequently, virtual atoms do not contribute to the final facility openings. Although this reassignment incurs an additional connection cost, the resulting increase can be bounded by a multiplicative factor of $1+O(\varepsilon)$ for any representative $r$ (see \Cref{lem:virtual_distance}).

\begin{lemma}\label{lem:virtual_distance} %\pvmargin{This is moved out from \Cref{obs:filtering}}
    For distinct representatives $r,r'\in \calR$ and any $u,v\in H_{r'}$, \[
        d(r,v)\le (1+\xi)d(r,u) \qquad \text{where } \xi:=\frac{2}{\varepsilon^{-1}-1}=O(\varepsilon).
    \]
\end{lemma}
\begin{proof}
    Assume that representative $r_1$ is chosen before representative $r_2$ by \Cref{alg:filtering-quantile-radii}. Because $r_2$ is not removed by $r_1$, it implies $d(r_1,r_2)>\varepsilon^{-1}R_{r_2}=\varepsilon^{-1}\max(R_{r_1},R_{r_2})$.
    If $R_{r'}=0$, then all locations in $H_{r'}$ coincide, and $d(r,v)=d(r,u)$. Otherwise, by triangle inequality, $d(r,u)\ge d(r,r')-d(r',u)>(\varepsilon^{-1}-1)R_{r'}$, and 
    \begin{align*}
        d(r,v)& \le d(r,u)+d(u,v)\le d(r,u)+2R_{r'} \\
        & \le \left(1+\frac{2}{\varepsilon^{-1}-1}\right)d(r,u). \qedhere
    \end{align*}
\end{proof}

% For the last assertion, \add{the separation above gives $d(r,r')>\varepsilon^{-1}R_{r'}$ in either order.} \del{if}\add{If} $R_{r'}=0$, then all locations in $H_{r'}$ coincide, and there is no distance increase. Otherwise, we have $d(r,u)\ge d(r,r')-d(r',u)>(\varepsilon^{-1}-1)R_{r'},$ and

% Figure~\ref{fig:atom-construction} illustrates the construction before splitting high-influence atoms.

\paragraph{Step~\ref{step:pre-influence}: Splitting High-Influence Atoms}
Let $\mathcal E_\phys$ and $\calE_\virt$ denote the set of all physical and virtual atoms, respectively. Let $\calE$ denote the set of all atoms, i.e, $\calE=\calE_\phys\cup \calE_\virt$. We further partition $\calE$ into
\[
    \calE_0:= \{e\in \calE_\phys: d_e=0\}, \quad
    \calE_1:= \{e\in \calE_\phys: d_e=1\}, \quad
    \calE_\rest:=  \calE\setminus (\calE_0\cup\calE_1).
\]

For each atom $e$, let $d_e$ denote the number of tails containing $e$, i.e., $d_e:=|{r\in\mathcal R: e \in T_r}|$.
Let $F_0$ denote the updated set $F$, which contains all atoms and head facilities. Note that $F_0$ contains each virtual atom only once: it is regarded as a virtual atom from the perspective of its associated tails and as a head facility from the perspective of its associated head. Observe that
\begin{equation*}
    y(F)=y(F_0)=\sum_{r\in\mathcal R} y(F_r) +\sum_{e\in \mathcal E_\phys:d_e=0} y_e -\sum_{e\in \mathcal E_\phys: d_e\ge 1} (d_e-1)y_e -\sum_{e\in \mathcal E_\virt} d_ey_e.
\end{equation*}
The second equality follows by accounting for atoms that are omitted or counted multiple times in $\sum_{r\in\mathcal R} y(F_r)$. Specifically, this sum omits all physical atoms outside $\bigcup_{r\in\mathcal R} F_r$, namely those with $d_e=0$. Moreover, each physical atom $e$ with $d_e\ge 1$ is counted $d_e$ times, once for each tail containing it, while each virtual atom $e$ is counted $d_e+1$ times: $d_e$ times through the tails and once through its unique head, where uniqueness follows from \Cref{obs:disjoint_heads}.

We denote the total redundant amount by \[
    A_0:= \sum_{e\in\mathcal E_\phys: d_e\ge 2} (d_e-1)y_e
    +\sum_{e\in\mathcal E_{\virt}} d_ey_e.
\]
Consequently,
\begin{equation}\label{eq:y0-formula}
    k\ge y(F_0)=|\mathcal R| + y(\calE_0) -A_0.
\end{equation}
% Let
% \[A_0:=\sum_{e\text{ physical atom, }d_e\ge 2}(d_e-1)y_e+\sum_{e\text{ virtual atom}}y_e.\]

% \del{We take a detour to briefly explain the coefficients $d_e - 1$ for physical atoms with $d_e \geq 2$ and $1$ for virtual atoms in the definition of $A_0$. In Step~\ref{step:pre-scale} of Algorithm~\ref{alg:preprocessing}, scale up the $y$-value of such an atom $e$ (i.e, atoms in $\mathcal{E}_\text{rest}$) by $(1+\theta)$ for a tiny $\theta  = \varepsilon^{2\ell + 1}$ to $y_e + \theta y_e$.  To maintain that the total $y$-value of the head and tail for any $r \in \mathcal R$ to be $1$, we shall decrease the total $y$-value of heads whose correspondent tails contain $r$. This will increase the total $y$ value of heads by $\theta d_e y_e$. However, as we do not consider the openings of virtual atoms, the net decrease of $y$-value given by the operation is $\theta (d_e - 1)y_e$ if $e$ is a physical atom with $d_e \geq 2$, and $\theta y_e$ if $e$ is a virtual atom.}
% \add{The coefficients in $A_0$ measure the saving described in the overview. Step~\ref{step:pre-scale} scales up each physical atom $e$ with $d_e\ge2$ and each virtual atom $e$ by a factor $1+\theta$. This decreases the heads of the $d_e$ tails containing $e$ by $\theta d_e y_e$ in total, ignoring the threshold $1/2$, which is handled by $\Delta_1$ below. The opening of a physical atom increases by $\theta y_e$, whereas virtual atoms are not part of the output. The net saving is therefore $\theta(d_e-1)y_e$ for a physical atom with $d_e\ge2$ and $\theta y_e$ for a virtual atom.}

\begin{definition}[High Influence]
    A physical atom $e$ has \emph{high influence} if $d_e-1> A_0$. A virtual atom $e$ has high influence if $d_e> A_0$. Otherwise, it has low influence.
\end{definition}

For our analysis, particularly the application of Lemma~\ref{lem:pipage-concentration} in \Cref{lem:wa}, we need to bound $d_e$ for every atom $e$. To this end, when $A_0\ge2$, we split each high-influence atom $e$ into $m$ copies $e_1,\ldots,e_m$, each retaining the original opening value $y_{e_i}=y_e$.

Specifically, let $q:=\lfloor A_0\rfloor$ and $m:=\left\lceil d_e/q \right\rceil$.
If $A_0\ge 2$, we partition the $d_e$ tails containing high influence atom $e$ into $m$ nonempty groups, each of size at most $q$, and assign a distinct copy $e_i$ to each group. Consequently, every copy satisfies \[d_{e_i} \le A_0.\]
Each copy inherits the type of the original atom: copies of physical atoms remain physical, while copies of virtual atoms remain virtual.

% A physical atom $e$ is called high-influence if $d_e-1> A_0$ (otherwise low-influence). We \del{replace}\add{split} each high-influence physical atom $e$ into several colocated low-influence physical atoms $\{e_1,\ldots,e_m\}$ for some $m$,  such that $y_{e_o}=y_e$ for every $o \in [m]$, and each of the $d_e$ tails that contains $e$ now contains exactly one of the $e_o$'s. Specifically, let $q:=\lfloor A_0\rfloor+1$, and partition the $d_e$ tails into $m=\lceil d_e/q\rceil$ nonempty groups of size at most $q$, and use one copy per group. That is, for any tail in the group, we replace the $e$ in the tail with the copy.  After this step, all physical atoms are low-influence. \add{Indeed, every copy has degree at most $q$, so $d_e-1\le\lfloor A_0\rfloor\le A_0$ for every physical atom $e$. This bounds the contribution of each atom to the sum $A'$ defined below, which is needed for concentration in Lemma~\ref{lem:wa}.}

Let $\calE'_\phys$ and $\calE'_\virt$ denote the sets of physical and virtual atoms, respectively, after this step. For each $r\in \calR$, let $F'_r$ be obtained from $F_r$ by replacing each high-influence atom with the copy assigned to the tail of $r$. Note that $y(F'_r)=y(F_r)=1$. Let $\Gamma$ be the total amount of openings by physical atoms and facilities in head. Hence,
\begin{align}
    \Gamma = & y(\mathcal E'_\phys)+\sum_{r\in \mathcal R} y(H_r) \nonumber \\
    = & \sum_{r\in\mathcal R} y(F'_r) +y(\mathcal E_0) -\sum_{e\in \mathcal E'_\phys:d_e\ge2} (d_e-1)y_e -\sum_{e\in \mathcal E'_\virt} d_ey_e\nonumber \\
    = & |\mathcal R| +y(\mathcal E_0)-A, \label{eq:y1-formula}
\end{align}
where \[
    A:= \sum_{e\in \mathcal E'_\phys:d_e\ge2} (d_e-1)y_e +\sum_{e\in \mathcal E'_\virt} d_ey_e.
\]
For the second equality, $\sum_{r\in \calR} y(F'_r)$ counts each atom $e$ for $d_e$ times. Note that every atom outside $\mathcal E_0$ belongs to at least one tail. To ensure that each physical atom is counted exactly once and each virtual atom is excluded, we subtract $(d_e-1)y_e$ for each physical atom and $d_ey_e$ for each virtual atom.

% \del{Let $A$ be the value of the expression defining $A_0$, evaluated on the atoms and their degrees after Step~\ref{step:pre-influence}; $A_0$ remains its value before that step.}\add{Let $A$ denote the value of the expression defining $A_0$, evaluated on the atoms and their degrees after Step~\ref{step:pre-influence}; thus $A_0$ and $A$ are the values before and after this step, respectively.} The splitting will increase the total $y$ value of physical atoms, and let $\Delta_0$ be the total increase.

Define $\Delta_0$ as the increase of $y$ value of physical atoms.
\begin{lemma}\label{lem:Delta-equality}
	$\Delta_0 = A_0 - A$.
\end{lemma}
\begin{proof}
    Because $\Delta_0=\Gamma-y(F_0)$, the statement is obtained by subtracting \eqref{eq:y0-formula} from \eqref{eq:y1-formula}.
\end{proof}
% \begin{proof}
% 	When we split a high-influence physical atom $\add{e}$ into $m \geq 2$ low-influence ones, we increase total $y$ value by $(m - 1)y_{\add{e}}$, and we decreased $\sum_{\add{f}\text{ physical atom}: d_{\add{f}} \geq 2}(d_{\add{f}} - 1) y_{\add{f}}$ by $(m-1)y_{\add{e}}$, as $\sum_{o=1}^m d_{\add{e_o}}=d_{\add{e}}$ and $d_{\add{e_o}} \geq 1$ for every $o \in [m]$.
% 	Hence the lemma holds.
% \end{proof}

\begin{lemma}\label{lem:delta0}
	We have $\Delta_0<3/2$. Moreover, if $A_0\ge2$, then
$\Delta_0\le A_0/2$.
\end{lemma}

\begin{proof}
	If $A_0<2$, then $\Delta_0=0$. Otherwise, suppose a
high-influence physical atom $e$ is divided into $m_e$
colocated atoms. By the grouping rule, we have 
\[
m_e-1
=\left\lfloor\frac{d_e-1}{q}\right\rfloor.
\]
Let $\mathcal H$ denote the set of all high-influence physical
atoms. Therefore,
\[
\Delta_0
=\sum_{e\in\mathcal H}(m_e-1)y_e
\le \frac1q\sum_{e\in\mathcal H}(d_e-1)y_e
\le \frac{A_0}{q}.
\]
Since $q=\lfloor A_0\rfloor\ge2$ and $A_0<q+1$, we have
$\Delta_0\le A_0/2$ and
$\Delta_0<1+1/q\le3/2$.
\end{proof}

\paragraph{Step~\ref{step:pre-scale}: Scaling the $y$-Values of Some Atoms}
% Let $\mathcal E$ be the set of all atoms we constructed so far. Let $\mathcal E_0,\mathcal E_1$ be the set of physical atoms with degree $0$ and $1$, respectively.
Let $\mathcal E_1(r) = \mathcal E_1 \cap T_r$ for every $r \in \mathcal R$. Therefore, $(\mathcal E_1(r))_{r \in \mathcal R}$ forms a partition of $\mathcal E_1$ as every atom in $\mathcal E_1$ appears in exactly one $T_r$.
Recall that $\mathcal E_{\text{rest}}$ is the set of physical atoms with degree at least $2$ and virtual atoms. % Then $\mathcal E:=\mathcal E_0\cup\mathcal E_1\cup\mathcal E_{\text{rest}}$.

For $e\in\mathcal E_{\text{rest}}$, let $\bar y_e\leftarrow (1+\theta)y_e$. For $e\in\mathcal E_0\cup\mathcal E_1$, let $\bar y_e\leftarrow y_e$. Note that $\bar y_e\le\delta = \varepsilon^{14\ell + 2}$ for every $e \in \mathcal E$, since every atom was created with $y_e\le\delta/(1+\theta)$ and all copies keep this value.

\paragraph{Step~\ref{step:pre-round-atoms}: Rounding Atoms (i.e., Tails)}
Consider the following laminar family
\[\mathcal L=\{\mathcal E_0,\mathcal E_1,\mathcal E_{\text{rest}}\}\cup \left\{\mathcal E_1(r)\mid r\in\mathcal R\right\}\cup\left\{\{e\}\mid e\in\mathcal E\right\}.\]
We apply pipage rounding to $\bar{\vecy}_{|\mathcal E}$ with laminar
family $\mathcal L$ and granularity $\delta$, as defined in
the preliminaries, to obtain $\vecy'$ over $\mathcal E$.
So, for any $L\in\mathcal L$, we have
\[y'(L)\in\{\delta\lfloor \bar y(L)/\delta\rfloor,\delta\lceil \bar y(L)/\delta\rceil\}.\] Particularly, $y'_e\in\{0,\delta\}$ for any atom $e$.

Define \[
    A':=\sum_{e\in \calE'_\phys:d_e\ge2}(d_e-1)y'_e+\sum_{e\in \calE'_\virt}d_ey'_e.
\]
% The definition is the same as that of $A_0$ and $A$, but for the rounded vector $y'$.
For the following lemma, observe that without scaling up on $\calE_\rest$, we would have $\E[A']=A$. Consequently, the desired event $A'\ge\left(1+\frac{\theta}{2}\right)A$ would require $A'\ge \E[A']+\frac{\theta}{2}A$, which does not occur with sufficiently high probability. This motivates scaling up $\vecy$ to $\bar \vecy$, so that we can expect $A'$ to be slightly larger than $A$.
\begin{lemma}\label{lem:wa}
	% Assume $A_0>8$, then $A_0\le 2A$ and
	If $A_0\ge 2$, then $A_0\le 2A$ and
	\[\Pr\left[A'<\left(1+\frac{\theta}{2}\right)A\right]\le\exp\left(-\Omega\left(\varepsilon^{-10\ell}\right)\right).\]
\end{lemma}

\begin{proof}
	By Lemmas \ref{lem:Delta-equality} and \ref{lem:delta0}, $A=A_0-\Delta_0$ and $\Delta_0 \leq A_0/2$, we have $A_0\le 2A$.

    According to the pipage rounding with granularity $\delta$, for any
atom $e$, $\Pr(y'_e=\delta)=\frac{\bar y_e}{\delta}$.
Since $\mathbb E[y'_e]=\bar y_e=(1+\theta)y_e$ for
$e\in\mathcal E_{\text{rest}}$, we have
$\mathbb E[A']=(1+\theta)A$. Moreover,
\begin{align*}
    &\sum_{e\in \calE'_\phys: d_e\ge 2}
        (d_e-1)^2\delta\bar y_e
      +\sum_{e\in \calE'_\virt} d_e^2\delta\bar y_e\\
    &\qquad\le
      \delta A_0(1+\theta)
      \left(
        \sum_{e\in \calE'_\phys: d_e\ge 2}(d_e-1)y_e
        +\sum_{e\in \calE'_\virt}d_ey_e
      \right)
      \tag{$e$ has low influence}\\
    &\qquad=
      \delta(1+\theta)A_0A
      \le O(\delta)A^2.
\end{align*}
	% \[\begin{aligned}
	% 	\mathrm{Var}(A')&\le\sum_{e\in \mathcal E'_\phys d_e\ge 2}(d_e-1)\delta\mathbb E[(d_e-1)y'_e]+\sum_{e\in \mathcal E'_\virt}\delta\mathbb E[y'_e]\\
	% 	&\le \delta A_0\mathbb E[A']\le O(\delta)A^2.
	% \end{aligned}\]

	% \del{By the negative product-moment property (A3) in~\cite{srinivasan2001distributions}, the exponential moments of $A'$ are bounded by those of the independent Bernoulli sum with the same marginals. The first upper bound above also bounds the variance of that independent sum.}\add{We apply Lemma~\ref{lem:pipage-concentration} to $A'$, with coefficient $d_e-1$ on each physical atom with $d_e\ge2$, coefficient $1$ on each virtual atom, and $0$ elsewhere; the first upper bound above is the variance proxy of that lemma. Recall that $d_e-1\le A_0$ for every physical atom after Step~\ref{step:pre-influence}.}
    Since all atoms have low influence, we have, for every
$e\in \calE'_\phys\cap\calE_\mathrm{rest}$,
\[
    (d_e-1)\delta\le\delta A_0\le O(\delta)A.
\]
Similarly, for every $e\in\calE'_\virt$, we have
$d_e\delta\le\delta A_0\le O(\delta)A$.
Thus, applying Lemma~\ref{lem:pipage-concentration} to $A'$
with $t=\theta A/2$, we obtain
\begin{align*}
\Pr\left[A'<\left(1+\frac{\theta}{2}\right)A\right]
    &\le
    \exp\left(
        -\frac{(\theta A/2)^2}
        {O(\delta)A^2+O(\delta)A\cdot(\theta A/2)}
    \right)\\
    &\le
    \exp\left(-\Omega\left(\frac{\theta^2}{\delta}\right)\right)
    =
    \exp\left(-\Omega\left(\varepsilon^{-10\ell}\right)\right).
    \qedhere
\end{align*}
\end{proof}

\paragraph{Loop~\ref{step:pre-round-heads}: Rounding Heads}
Next, we round the $y$ values of facilities in heads. For head rounding, we use the facility pieces in each head before
atomization, with their original opening values.
The rounded atom values are retained separately and remain unchanged
during head rounding. For $r\in\mathcal R$, let $Y_r:=1-\min(1/2,y'(T_r))$. We apply pipage rounding to $\frac{Y_r}{1-\varepsilon}\vecy_{|H_r}$ with laminar family $\{H_r\}\cup\{\{i\}:i\in H_r\}$ and granularity $\delta$, as defined in preliminaries, using independent fresh randomness for each head conditional on the rounded atoms. So, for every $i\in H_r$, we have
\[\mathbb E\left[y'_i \middle|\ (y'_e)_{e\in\mathcal E}\right]=\frac{Y_r}{1-\varepsilon}\cdot y_i,\qquad y'_i\in \left\{\delta \left\lfloor \frac{Y_r}{1-\varepsilon}\frac{y_i}{\delta}\right\rfloor,\delta \left\lceil \frac{Y_r}{1-\varepsilon}\frac{y_i}{\delta}\right\rceil\right\}.\]
Moreover, $y(H_r)=1-\varepsilon$, so the input vector for this call
has total $Y_r$. Since $1/(2\varepsilon)\in\mathbb N$ and
$\delta=\varepsilon^{14\ell+2}$, the values $1$, $1/2$, and
$y'(T_r)$ are multiples of $\delta$, % \add{ (indeed, $1/(2\delta)=\frac{1}{2\varepsilon}\cdot\varepsilon^{-(14\ell+1)}\in\mathbb N$, and $y'_e\in\{0,\delta\}$ for every atom $e$)}
 and hence so is $Y_r$.
The rounding therefore preserves this total exactly: $y'(H_r)=Y_r\in[1/2,1]$.

% \gptremovedfile{removed-text/preprocessing-06}

% \begin{gptadded}
To account for the threshold $1/2$, define the excess for each representative and its total by
\[
\Delta_{1,r}:=Y_r-(1-y'(T_r))=(y'(T_r)-1/2)_+,
\qquad \Delta_1:=\sum_{r\in\mathcal R}\Delta_{1,r}.
\]
We prove the following lemma that upper bounds $\Delta_1$:
\begin{lemma}\label{lem:delta1} 
    For $\varepsilon\le1/8$,
	\[\Pr\left[\Delta_1>\frac{\theta A}{2}\right]\le\exp\left(-\Omega\left(\varepsilon^{-(14\ell+2)}\right)\right).\]
\end{lemma}
% \end{gptadded}

\begin{proof}
	Fix a representative $r\in\mathcal R$. Define $Y_1=\sum_{e\in\mathcal E_1\cap T_r}y_e,\ Y_\rest=\sum_{e\in\mathcal E_{\rest}\cap T_r} y_e$ and define $Y'_1$ and $Y'_\rest$ in the same way. Since $\mathcal E_1(r)\in\mathcal L$ and $\bar y_e=y_e$ on $\mathcal E_1$, pipage rounding gives $Y_1'\le Y_1+\delta$. We also have $Y_1+Y_\rest=\varepsilon$. % For sufficiently small $\varepsilon$, we have $\varepsilon+\delta\le1/4$, and 
    Hence,
	\[\begin{aligned}
		\Delta_{1,r}&=\max(0,Y_1'+Y_\rest'-1/2)\\
        &\le \max(0,Y_1+\delta+Y'_\rest-1/2)\\
		&\le\max(0,\varepsilon-Y_\rest+\delta+Y'_\rest-1/2) \\
		&\le\max(0,Y'_\rest-Y_\rest-1/4).
	\end{aligned}\]
	
	For $\lambda\ge0$, we first estimate the exponential moment of $Y'_\rest$, using negative-correlation property in \eqref{eq:neg-corr}.
	\[\begin{aligned}
		\mathbb E[\exp(\lambda Y'_\rest)]&\le \prod_{e\in \mathcal E_{\text{rest}}\cap T_r}\mathbb E[\exp(\lambda y'_e)]\\
		&\le \prod_{e\in \mathcal E_{\text{rest}}\cap T_r}\left(1+\frac{(1+\theta)y_e}{\delta}(e^{\lambda\delta}-1)\right)\\
		&\le \exp\left(\frac{(1+\theta)Y_\rest}{\delta}(e^{\lambda\delta}-1)\right).
	\end{aligned}\]
	
	Consequently, by the same negative-correlation property in \eqref{eq:neg-corr}, and $y'_e\in\{0,\delta\}$,
	
	\[\begin{aligned}
		&\mathbb E[Y'_\rest\exp(\lambda Y'_\rest)]\\
		&\quad=\sum_{e\in\mathcal E_{\text{rest}}\cap T_r}
		\mathbb E[y'_e\exp(\lambda Y'_\rest)]\\
		&\quad\le e^{\lambda\delta}\sum_{e\in\mathcal E_{\text{rest}}\cap T_r}\bar y_e
		\prod_{\substack{f\in\mathcal E_{\text{rest}}\cap T_r\\f\ne e}}
		\left(1+\frac{\bar y_f}{\delta}(e^{\lambda\delta}-1)\right)\\
		&\quad\le e^{\lambda\delta}\cdot (1+\theta)Y_\rest\cdot
		\exp\left(\frac{(1+\theta)Y_\rest}{\delta}(e^{\lambda\delta}-1)\right).
	\end{aligned}\]
	
	Then we can control $\Delta_{1,r}$ as follows:
	
	\[\begin{aligned}
		\mathbb E[\Delta_{1,r}]&\le \mathbb E[Y'_\rest\mathbf 1_{\{\Delta_{1,r}>0\}}]\\
		&\le \mathbb E[Y'_\rest\exp\left(\lambda\left(Y'_\rest-Y_\rest-1/4\right)\right)]\\
		&\le \exp(-\lambda(Y_\rest+1/4))\cdot \mathbb E[Y'_\rest\exp(\lambda Y'_\rest)]\\
		&\le e^{\lambda\delta}\cdot (1+\theta)Y_\rest\cdot \exp\left(-\lambda(Y_\rest+1/4)+\frac{(1+\theta)Y_\rest}{\delta}(e^{\lambda\delta}-1)\right)\\
		&=e^{\lambda\delta}\cdot (1+\theta)Y_\rest\cdot \exp\left(-\lambda/4+Y_\rest\left(\frac{(1+\theta)}{\delta}(e^{\lambda\delta}-1)-\lambda\right)\right).
	\end{aligned}\]
	
	Let $\lambda=1/(4\delta)$, then $(1+\theta)(e^{\lambda\delta}-1)/\delta-\lambda\le 1/(16\delta)$. Combined with $Y_\rest\le\varepsilon$, we have
	
	\[\begin{aligned}
		\mathbb E[\Delta_{1,r}]&\le e^{1/4}\cdot (1+\theta)Y_\rest\cdot \exp\left(-\frac{1}{16\delta}+\frac{Y_\rest}{16\delta}\right)\\
		&\le O(1)\cdot Y_\rest\cdot\exp\left(-\Omega\left(\frac{1}{\delta}\right)\right)
		\le \exp\left(-\Omega\left(\varepsilon^{-(14\ell+2)}\right)\right)\cdot Y_\rest.
	\end{aligned}\]
	
	Summing over all representatives, we have
	
	\[\begin{aligned}
		\mathbb E[\Delta_1]&\le\exp\left(-\Omega\left(\varepsilon^{-(14\ell+2)}\right)\right)\cdot\sum_{r\in\mathcal R}\sum_{e\in\mathcal E_{\text{rest}}\cap T_r}y_e\\
		&=\exp\left(-\Omega\left(\varepsilon^{-(14\ell+2)}\right)\right)\cdot\sum_{e\in\mathcal E_{\text{rest}}}d_ey_e\\
		&\le2\cdot \exp\left(-\Omega\left(\varepsilon^{-(14\ell+2)}\right)\right)\cdot\left(\sum_{e\in \mathcal E_\phys:d_e\ge 2}(d_e-1)y_e+\sum_{e\in\mathcal E_\virt} d_ey_e\right)\\
		&\le\exp\left(-\Omega\left(\varepsilon^{-(14\ell+2)}\right)\right)\cdot A.
	\end{aligned}\]
	The equality holds because each atom $e$ lies in exactly $d_e$ tails. The next inequality uses $d_e\le2(d_e-1)$ for physical atoms with $d_e\ge2$.

	If $A=0$, then $\Delta_1=0$ almost surely. Otherwise, Markov's inequality gives
\[
\Pr\left[\Delta_1>\frac{\theta A}{2}\right]
\le \frac{2}{\theta}\exp\left(-\Omega\left(\varepsilon^{-(14\ell+2)}\right)\right)
=\exp\left(-\Omega\left(\varepsilon^{-(14\ell+2)}\right)\right). \qedhere
\] 
\end{proof}

\paragraph{Step~\ref{step:pre-return}: Returning the Opening of Heads and Physical Atoms} Finally, we return the opening vector $\vecy'$ restricted to $\union_{r \in \mathcal R}H_r$ and physical atoms. We map the opening to the original facilities in $F$, and we cap the opening of each original facility by $1$.
Notice that virtual atoms are auxiliary variables for the tail rounding and are not
included in the final opening vector.
Merging sums multiples of $\delta$, and $1$ is a multiple of $\delta$, so every returned coordinate is a multiple of $\delta$ in $[0,1]$. Neither merging nor capping increases the total opening. Moreover, for every client $j$ and radius $t\ge0$, the deficit $(1-y'(\operatorname{ball}(j,t)))_+$ is unchanged. Merging preserves the opening at every location. Capping changes the opening only at a location whose opening is at least $1$, and every ball containing that location has deficit $0$ both before and after. Since connection costs are determined by these deficits (see \eqref{eq:con-integral}), costs are unchanged as well. The analysis below is therefore carried out on the pieces before this step.

\subsection{Bounding the Number of Open Facilities}\label{subsec:pre-count}

Let $\Gamma'$ be the total opening in heads and of physical atoms after Loop~\ref{step:pre-round-heads}. It follows that
\begin{align}
    \Gamma' & = \sum_{r\in \calR} y'(H_r) +\sum_{e\in \calE'_\phys} y'_e \nonumber \\
    & = \sum_{r\in \calR}Y_r + \left(y'(\calE_0)+\sum_{r\in \calR} y'(T_r) -\sum_{e\in \calE'_\phys:d_e\ge 2} (d_e-1)y'_e -\sum_{e\in \calE'_\virt} d_ey'_e\right)\nonumber \\
    & = \sum_{r\in \calR} (1+\Delta_{1,r})  +y'(\calE_0) -A'\nonumber\tag{$\Delta_{1,r}=Y_r-(1-y'(T_r))$}\\
    & = |\calR|+\Delta_1+y'(\calE_0)-A'. \label{eq:y2-formula}
\end{align}
For the second equality, recall that $y'(H_r)=Y_r$ because $Y_r$ is a multiple of $\delta$ and pipage rounding preserves this value. Moreover, $\sum_{r\in \calR} y'(T_r)$ counts each atom $e$ exactly $d_e$ times. To ensure that each physical atom is counted exactly once and each virtual atom is excluded, we subtract $(d_e-1)y'_e$ for each physical atom and $d_ey'_e$ for each virtual atom.

\begin{lemma}\label{lem:good-event}
    With probability at least $1-\exp(-\Omega(\varepsilon^{-10\ell}))$, $\Gamma'\le k+3$. %\pvmargin{Improved from $k+10$ to $k+3$.}
\end{lemma}
\begin{proof}
    By combining \eqref{eq:y0-formula}, \eqref{eq:y1-formula}, \eqref{eq:y2-formula}, and since $\mathcal E_0$ is a laminar rounding set and is not scaled,
$y'(\mathcal E_0)\le y(\mathcal E_0)+\delta$, we have
    \begin{align*}
        \Gamma' & = (\Gamma'-\Gamma)+(\Gamma-y(F_0))+y(F_0) \\
            & =(\Delta_1+y'(\calE_0)-y(\calE_0)-A'+A)+\Delta_0+y(F_0) \\
            & \le k+\Delta_0+\Delta_1+\delta-A'+A.
    \end{align*}
    Now consider two cases:
    \begin{itemize}
        \item \textbf{Case I:} $A_0>2$. \\
         It follows from \Cref{lem:wa} that $A'-A\ge \frac{\theta}{2}A$ with probability at least $1-\exp(-\Omega(\varepsilon^{-10\ell}))$. Then by \Cref{lem:delta1}, we have $\Delta_1\le \frac{\theta}{2}A$ with probability at least $1-\exp(-\Omega(\varepsilon^{-14\ell-2}))$. By \cref{lem:delta0}, we have $\Delta_0<3/2.$ Thus,
        \begin{align*}
            \Gamma' & \le k+\Delta_0+\Delta_1+\delta+A-A' \\
            & \le k+\frac{3}{2}+\frac{\theta}{2}A+\delta-\frac{\theta}{2}A \\
            & \le k+2.
        \end{align*}
        By union bound, both events together hold with probability at least $1-\exp(-\Omega(\varepsilon^{-10\ell}))$.
        \item \textbf{Case II:} $A_0\le 2$. \\
        According to \Cref{lem:delta1}, $\Delta_1\le \theta$ with probability at least $1-\exp(-\Omega(\varepsilon^{-10\ell}))$. By \Cref{lem:Delta-equality}, we have $A_0=A+\Delta_0$. Thus,
        \begin{align*}
            \Gamma' & \le k+\Delta_0+\Delta_1+\delta+A-A' \\
            & \le k+A_0+\theta+\delta-A' \\
            & \le k+3. \tag{$\delta<\theta< \varepsilon\le 1/2$}
        \end{align*}
    \end{itemize}
    Therefore, $\Gamma'\le k+3$ with probability at least $1-\exp(-\Omega(\varepsilon^{-10\ell}))$.
\end{proof}

\subsection{Connection Cost Analysis}
Recall that $\mathcal E'_\phys\subseteq\mathcal E$ denotes the physical atoms
just before the rounding, and let
\[\mathcal U:=\mathcal E'_\phys\cup\bigcup_{r\in\mathcal R}H_r.\]
Thus $\mathcal U$ contains exactly the physical atoms and head facilities;
virtual atoms are excluded. Throughout this subsection, $y$ and $y'$ on
subsets of $\mathcal U$ refer to the openings before and after rounding,
prior to the final merging and capping. Throughout the connection-cost analysis, \(\operatorname{cost}_{\mathbf y}(j)\) continues to denote the original LP connection cost defined by $\vecy^\mathrm{orig}$.

For $j\in C$ and $t\ge0$, let $\operatorname{ball}(j,t)$ denote the
closed metric ball centered at $j$ with radius $t$. When an opening or
assignment is evaluated on a ball, we sum its mass over the points of its
own domain in that ball. Recall that $\vecy^{\mathrm{orig}}$ is the input LP
opening on the facility pieces used to define heads and tails, before
any atom construction or duplication, and define
\[q_j(t):=\left(1-y^{\mathrm{orig}}(\operatorname{ball}(j,t))\right)_+.\]
Recall that
\[\xi=\frac{2}{\varepsilon^{-1}-1}=O(\varepsilon).\]
For every $r\in\mathcal R$ and every rounded outcome, let $\mu_r$ be the
assignment to the nearest one unit of rounded fractional opening $y'$
from $r$, using only $\mathcal U=\mathcal E_\phys\cup
\bigcup_{s\in\mathcal R}H_s$ and breaking distance ties in a fixed order.
This unit of opening exists: the comparison assignment $\nu_r$ constructed at the beginning of the proof of Lemma~\ref{lem:radial} uses only the opening $y'$ on $\mathcal U$ and has total mass $1$, and its construction does not use $\mu_r$. In particular, for $t\ge 0$,
\[\mu_r(\mathcal U)=1,\qquad 0\le\mu_r\le\vecy',\qquad
\mu_r\bigl(\mathcal U\setminus\operatorname{ball}(r,t)\bigr)
=\left(1-y'(\operatorname{ball}(r,t))\right)_+.\]

We bound the expected cost of each client $j$ by integrating its ball deficits $(1-y'(\operatorname{ball}(j,t)))_+$ over $t$. For radii $t$ below threshold $t_0=\Theta((d(j,r(j))+R_{r(j)})/\varepsilon)$, \Cref{lem:local-deficit} shows that rounding increases each deficit by at most an additive $O(\varepsilon\theta)$.
\begin{lemma}\label{lem:local-deficit}
For any $S\subseteq \mathcal U$, %\pvmargin{The constant is changed from 2 to 3 }
\[\mathbb E[(1-y'(S))_+]\le(1-y(S))_+ +3\varepsilon\theta.\]
\end{lemma}
For larger radii, we compare with the representative $r(j)$: \Cref{lem:radial} shows that, around $r(j)$, the rounded opening is as good as the original opening up to a factor $1+\xi=1+O(\varepsilon)$ in the radius and a factor $1+\theta$ in the mass.
\begin{lemma}\label{lem:radial}
For every representative $r$ and every $t\ge R_r$,
\[\mathbb E\left[\mu_r\bigl(\mathcal U\setminus
\operatorname{ball}(r,(1+\xi)t)\bigr)\right]
\le(1+\theta)q_r(t).\]
\end{lemma}
\Cref{lem:connection-cost} combines the two regimes. Note that $\operatorname{cost}_{\vecy'}(j)$ here only accounts for the contribution from $\mathcal U$.
\begin{lemma}\label{lem:connection-cost}
    For every client $j\in C$ and $\varepsilon\le1/8$, \[
        \E[\operatorname{cost}_{\vecy'}(j)] \le (1+O(p\varepsilon))\operatorname{cost}_{\vecy}(j).
    \]
\end{lemma}
\begin{proof}
Let $r=r(j)$. By the layer-cake formula $D^p=\int_0^\infty pt^{p-1}\mathbf 1_{\{D>t\}}\,dt$, we have 
    \begin{align}
        \cost_\mathbf{y}(j) & = \int_{0}^\infty pt^{p-1}q_j(t) \, dt, \nonumber \\
        \E[\cost_{\mathbf{y}'}(j)] & =\int_{0}^\infty pt^{p-1}\E[(1-y'(\calU\cap \operatorname{ball}(j,t)))_+] \, dt. \label{eq:con-integral}
    \end{align}
    For \eqref{eq:con-integral}, Tonelli’s theorem justifies exchanging expectation and integration, because the integrand is nonnegative.

    Let $t_0\ge d(j,r)+(1+\xi)R_r$ be the threshold at which we split the integral into two parts. By \Cref{lem:local-deficit}, we derive
    \begin{align*}
        \int_0^{t_0} pt^{p-1} \E[\big( 1-y'(\calU\cap \operatorname{ball}(j,t))\big)_+] \, dt
        \le &
        \int_0^{t_0} pt^{p-1} \big((1-y(\operatorname{ball}(j,t)))_++3\varepsilon\theta\big) \, dt
        \\ \le &
        \int_0^{t_0} pt^{p-1} \big((1-y^{\mathrm{orig}}(\operatorname{ball}(j,t)))_+ +3\varepsilon\theta\big) \, dt
        \\
        = & \int_0^{t_0} pt^{p-1}q_j(t) \, dt + 3\varepsilon\theta t_0^p,
    \end{align*}
    where the second inequality follows from the fact that $y(\operatorname{ball}(j, t)) \geq y^{\mathrm{orig}}(\operatorname{ball}(j, t))$ for any client $j \in C$ and radius $t \geq 0$ due to high-influence atoms split in Step~\ref{step:pre-influence}.

    To leverage the backup facilities associated with $r$, we use the fact that $0\le \mu_r\le \mathbf y'$, which implies $(1-y'(S))_+ \le (1-\mu_r(S))_+=\mu_r(\calU\setminus S)$ for any $S\subseteq \calU$. Consequently,
    \begin{align*}
        \E[(1-y'(\operatorname{ball}(j,t)))_+] \le &  \E[\mu_r(\calU\setminus\operatorname{ball}(j,t))] \\
        \le & \E[\mu_r\big(\calU\setminus \operatorname{ball}(r,t-d(j,r))\big)] \tag{*} \\
        \le & \E[\mu_r\big(\calU\setminus \operatorname{ball}(r,(1+\xi)s)\big)] \\
        \le & (1+\theta)q_r(s) \tag{\Cref{lem:radial}} \\
        \le & (1+\theta)q_j(s-d(j,r)), \tag{**}
    \end{align*}
    where $s$ is a value such that $(1+\xi)s=t-d(j,r)$. Observe that $s\ge R_r$ when $t\ge t_0$. %Note that this choice of $s$ is possible by an appropriate choice of $t_0$.
    For (*) and (**), we know that, for $t\ge R_r$, $\operatorname{ball}(r,t-d(j,r))\subseteq \operatorname{ball}(j,t)$ and $\operatorname{ball}(j,s-d(j,r))\subseteq \operatorname{ball}(r,s)$, respectively.

    Let $s_0$ be a value such that $(1+\xi)s_0=t_0-d(j,r)$. We know that, for $u=s-d(j,r)$, \[
        t = (1+\xi)u+(2+\xi)d(j,r)=t_0+(1+\xi)(u-s_0)+(1+\xi)d(j,r).
    \]
    Suppose that $t_0\le (1+\chi)u_0$ where $u_0:=s_0-d(j,r)\ge 0$ and $\chi$ is a constant such that $\xi\le \chi$. Note that $u_0\ge 0$ by this condition. Then $t\le (1+\chi)u$. In this case, we have
    \begin{align*}
        \int_{t_0}^\infty pt^{p-1} \E[(1-y'(\calU\cap \operatorname{ball}(j,t)))_+] \, dt & \le \int_{t_0}^\infty pt^{p-1} (1+\theta)q_j(s-d(j,r)) \,dt \\
        & = (1+\theta)(1+\xi) \int_{u_0}^\infty pt^{p-1}q_j(u) \, du \\
        & \le (1+\theta)(1+\xi)(1+\chi)^{p-1} \int_{u_0}^\infty pu^{p-1}q_j(u)\,du.
    \end{align*}
    Combine both integrals; we get
    \begin{align*}
        \mathbb E[\cost_{\mathbf{y}'}(j)] & \le \int_0^{t_0} pt^{p-1}q_j(t)\, dt +3\varepsilon\theta t_0^p+(1+\theta)(1+\xi)(1+\chi)^{p-1} \int_{u_0}^\infty pu^{p-1}q_j(u)\,du \\
        & \le (1+\theta)(1+\xi)(1+\chi)^{p-1} \int_0^{\infty} pt^{p-1}q_j(t)\, dt +3\varepsilon\theta t_0^p +\int_{u_0}^{t_0} pt^{p-1}q_j(t)\, dt \\
        & \le (1+\theta)(1+\xi)(1+\chi)^{p-1} \cost_{\mathbf y}(j) +3\varepsilon\theta t_0^p +(t_0^p-u_0^p)q_j(u_0) \\
        & \le (1+\theta)(1+\xi)(1+\chi)^{p-1} \cost_{\mathbf y}(j) +3\varepsilon\theta t_0^p +((1+\chi)^p-1) u_0^p q_j(u_0) \\
        & \le (1+\theta)(1+\xi)(1+\chi)^{p-1} \cost_{\mathbf y}(j) +3\varepsilon\theta t_0^p +((1+\chi)^p-1) \int_0^{u_0} pu^{p-1} q_j(u)\,du \\
        & \le (1+O(p\varepsilon))\cost_{\mathbf y}(j) +3\varepsilon\theta t_0^p.
    \end{align*}
    To achieve the desired upper bound of $(1+O(p\varepsilon))\operatorname{cost}_\mathbf{y}(j)$, we additionally impose the constraint $\chi= O(\varepsilon)$. We now try to minimize $t_0$ subject to all the constraints established above. That is,
    \[
        t_0 \ge d(j,r)+(1+\xi)R_r, \quad
        t_0 \le (1+\chi)\left(\frac{t_0-d(j,r)}{1+\xi}-d(j,r)\right), \quad
        \chi \ge \xi, \quad
        \chi = O(\varepsilon).
    \]
    One feasible solution to these constraints is $(t_0,\chi)=\left(\frac{d(j,r)+(1+\xi)R_r}{\varepsilon}, \frac{4\varepsilon}{1-3\varepsilon}\right)$ when $\varepsilon\le 1/8$. 

    According to \Cref{obs:filtering}, $R_r^p\le R_j^p\le \varepsilon^{-1}\cost_{\mathbf y}(j)$ and $d(j,r)\le \varepsilon^{-1}R_j$. That is, %\pvmargin{$t_0$ no longer has constant 40.}
    \begin{align*}
        t_0 & \le \varepsilon^{-1}(\varepsilon^{-1}R_j+(1+\xi)R_j)\le \frac{1+\varepsilon(1+\xi)}{\varepsilon^2}R_j, \\
        t_0^p &\le \frac{(1+\varepsilon(1+\xi))^p}{\varepsilon^{2p+1}}\operatorname{cost}_\mathbf{y}(j)\le \frac{1+O(p\varepsilon)}{\theta}\operatorname{cost}_\mathbf{y}(j). \tag{$\ell=\lceil p \rceil$}
    \end{align*}
    Therefore,
    \begin{align*}
        \mathbb E[\cost_{\mathbf{y}'}(j)] 
        & \le (1+O(p\varepsilon))\cost_{\mathbf y}(j) +3\varepsilon\left(1+O(p\varepsilon)\right)\cost_{\mathbf y}(j)  \tag{$p\varepsilon<1$} \\
        & = (1+O(p\varepsilon))\cost_{\mathbf y}(j). \qedhere
    \end{align*}
\end{proof}

\paragraph{Proof of Lemma \ref*{lem:local-deficit}.}

% \localdeficit*

% \gptremovedfile{removed-text/preprocessing-09}

% \begin{gptadded}
Our proof works with a sub-mass of $S$ of total exactly $\min(y(S),1)$. Consider a fixed fractional combination of physical atoms and facilities in the heads, specified by deterministic coefficients $\vecc\in [0,1]^\mathcal U$ satisfying \[
    \sum_{e\in \calE_\phys} c_ey_e+\sum_{r\in \calR} \sum_{i\in H_r} c_iy_i=m, \quad \text{and} \quad c_i=0 \quad\forall i\notin S,
\]
where $m:=\min(y(S),1)$. Let $m'$ denote the total opening after rounding, weighted by the coefficients $\vecc$, i.e., \[
    m':=\sum_{e\in\calE_\phys} c_ey'_e+\sum_{r\in\calR}\sum_{i\in H_r} c_iy'_i.
\]
Since $\operatorname{supp}(\vecc)\subseteq S$ and $c_i\in[0,1]$ for all $i\in\mathcal U$, we have $0\le m'\le y'(S)$. Furthermore, by the Cauchy–Schwarz inequality (or Jensen's inequality), the mean absolute deviation of $m'$ is bounded above by its standard deviation: $\E[|m'-\E[m']|]\le \sqrt{\Var(m')}$. Hence,
\begin{align}
    \E[(1-y'(S))_+] & \le \E[(1-m')_+]\nonumber \\
    & \le \E[(1-\E[m'])_+]+\E[|\E[m']-m'|] \nonumber \\ 
    & \le (1-\E[m'])_+ +\sqrt{\Var(m')}. \label{eq:local-deficit}
\end{align}
% \end{gptadded}

It remains to bound $\E[m']$ (\Cref{lem:local-deficit-expectation}) and $\Var[m']$ (\Cref{lem:local-deficit-variance}). To this end, we define the mass contributed by physical atoms and facilities in the heads to the fractional combination as follows: \[
    m_\phys:=\sum_{e\in \calE_\phys} c_ey_e, \quad m_r:=\sum_{i\in H_r} c_iy_i, \quad m'_\phys:=\sum_{e\in \calE_\phys} c_ey'_e,  \quad m'_r:=\sum_{i\in H_r} c_iy'_i.
\] 
for all $r\in \calR$. Thus, \[
    m=m_\phys+\sum_{r\in \calR}m_r \quad \text{and} \quad m'=m'_\phys+\sum_{r\in \calR} m'_r.
\]

\begin{lemma}\label{lem:local-deficit-expectation}
$\E[m']\ge m-\frac{\varepsilon\theta}{1-\varepsilon}$.
\end{lemma}
\begin{proof}
By the properties of the pipage rounding, $\E[y'(T_r)]\le (1+\varepsilon^{2\ell+1})\varepsilon$ and $y'(H_r)=Y_r=1-\min(y'(T_r),1/2)$.
Consequently, \[
    \E[Y_r]\ge 1-\E[y'(T_r)]\ge 1-\varepsilon-\varepsilon\theta.
\]
Moreover, recall that $\mathbb E\left[y'_i\mid(y'_e)_{e\in\mathcal E}\right]
=\frac{Y_r}{1-\varepsilon}y_i$ for any $i\in H_r$. By the law of total expectation,
\begin{align*}
    \E[m'_r] & = \sum_{i\in H_r} c_i\E_{(y'_e)_{e\in \calE}}\big[\E[y_i'\mid (y'_e)_{e\in \calE}]\big] \\
    & = \sum_{i\in H_r} c_i\frac{\E[Y_r]}{1-\varepsilon}\E[y_i] \\
    & \ge \frac{m_r}{1-\varepsilon}\E[Y_r].
\end{align*}
Therefore,
\begin{align*}
    \E[m'] & =\E[m'_\phys]+\sum_{r\in \calR}\E[m'_r] \\
    & \ge\E[m'_\phys]+\sum_{r\in \calR}\frac{m_r}{1-\varepsilon}\E[Y_r] \\
    & \ge m_\phys + \sum_{r\in \calR}m_r-\frac{\varepsilon\theta}{1-\varepsilon}\sum_{r\in \calR} m_r \\
    & \ge m-\frac{\varepsilon\theta}{1-\varepsilon}(m-m_\phys). \qedhere
\end{align*}
\end{proof}

\begin{lemma}\label{lem:local-deficit-variance}
$\Var(m')\le 9\delta$.
\end{lemma}
\begin{proof}
We bound the variance using the law of total variance.
Conditioned on $(y'_e)_{e\in\mathcal E}$,
$m'_\phys$ and all head targets $Y_r$ are fixed, and the head
rounding calls are independent.
Since $\Var(y'_e)\le \E[(y'_e)^2]= \delta\bar y_e$, and
by the negative-correlation property of the
pipage rounding (see \eqref{eq:neg-corr}) within each head, 
\begin{align*}
    \mathrm{Var}\left(m'\mid(y'_e)_{e\in\mathcal E}\right)
    &=\sum_{r\in\mathcal R}\mathrm{Var}\left(m'_r\mid(y'_e)_{e\in\mathcal E}\right)\\
    &\le\delta \sum_{r\in\mathcal R}\sum_{i\in H_r}
    c_i^2 \frac{Y_r}{1-\varepsilon}y_i  \\
    &\le\delta\sum_{r\in\mathcal R}\frac{m_r}{1-\varepsilon}Y_r \tag{$c_i^2\le c_i\le 1$}\\
    &\le\delta\frac{\sum_{r\in\mathcal R} m_r}{1-\varepsilon} \tag{$Y_r\in [1/2,1]$} \\
    \E[\mathrm{Var}\left(m'\mid(y'_e)_{e\in\mathcal E}\right)] &\le \frac{\delta}{1-\varepsilon}(m-m_\phys).
\end{align*}
On the other hand, the conditional head marginals give \[
    \mathbb E\left[m'\mid(y'_e)_{e\in\mathcal E}\right]
 = m'_{\phys}+\sum_{r\in\mathcal R}\frac{m_r}{1-\varepsilon}
\left(1-\min(y'(T_r),1/2)\right).
\]
By the subadditivity of the standard deviation,
\begin{align*}
\sqrt{\Var(\E\left[m'\mid(y'_e)_{e\in\mathcal E}\right])} & \le \sqrt{\Var(m'_\phys)}+ \sum_{r\in \calR} \frac{m_r}{1-\varepsilon}\sqrt{\Var\big(\min(y'(T_r),1/2)\big)}.
\end{align*}
By the negative correlation of the pipage rounding,
\begin{align*}
    \Var(m'_\phys) & \le \sum_{e\in \calE_\phys} c_e^2 \Var(y'_e) \\
    & \le \delta\sum_{e\in \calE_\phys} c_e^2\bar y_e \\
    & \le \delta(1+\theta) m_\phys, \tag{$c_e^2\le c_e\le 1$} \\
    \Var\big(\min(y'(T_r),1/2)\big)&\le \Var\big(y'(T_r)\big) \tag{\Cref{lem:capped-variance}} \\
    & \le \delta\sum_{e\in T_r} \bar y_e \le \delta(1+\theta)\varepsilon
    \tag{$y(T_r)=\varepsilon$ and $\bar y_e\le (1+\theta)y_e$}.
\end{align*}
Recall that $m_\phys\le m\le1$. Consequently,
\begin{align*}
    \Var(\E\left[m'\mid(y'_e)_{e\in\mathcal E}\right]) & \le \left(\sqrt{\delta(1+\theta)} +\frac{\sqrt{\delta(1+\theta)\varepsilon}}{1-\varepsilon}\sum_{r\in \calR}m_r\right)^2  \\
    &\le \delta(1+\theta)\left(1+\frac{\sqrt{\varepsilon}}{1-\varepsilon}\right)^2.
\end{align*}
Moreover, recall from the assumption of \Cref{thm:preprocessing} that $1/(2\varepsilon)\in\mathbb N$, implying $\varepsilon\le 1/2$. Thus, by the law of total variance,
\begin{align*}
    \Var(m')
    &
    =\E[\mathrm{Var}\left(m'\mid(y'_e)_{e\in\mathcal E}\right)]+\Var(\E\left[m'\mid(y'_e)_{e\in\mathcal E}\right])
    \\&
    \le \frac{\delta}{1 - \varepsilon} + \delta(1 + \theta) \left(1 + \frac{\sqrt \varepsilon}{1 - \varepsilon} \right)^2
    \\&
    \le \delta \max_{0<\varepsilon\le 1/2} \left(\frac{1}{1-\varepsilon}+(1+\varepsilon^3)\left(1+\frac{\sqrt\varepsilon}{1-\varepsilon}\right)^2\right) \\
    & < 9\delta. \qedhere
\end{align*}
\end{proof}

\begin{proof}[Proof of \Cref{lem:local-deficit}]
By combining \eqref{eq:local-deficit}, \Cref{lem:local-deficit-expectation,lem:local-deficit-variance}, we get \[
    \E[(1-y'(S))_+]\le 1-m+\frac{\varepsilon\theta}{1-\varepsilon}+3\sqrt{\delta} < (1-y(S))_+ +3\varepsilon\theta.
\]
For the last inequality, we can find that $1/(1-\varepsilon)+3\varepsilon^{5\ell-1}<3$ for $0<\varepsilon\le 1/2$.
\end{proof}

\paragraph{Proof of Lemma \ref*{lem:radial}.}

\begin{proof}[Proof of \Cref{lem:radial}]
% Fix a representative $r\in\mathcal R$. Let $T_r^{\mathrm{orig}}$ denote its original tail on facility pieces, before the construction of physical and virtual atoms. Thus, $H_r\cup T_r^{\mathrm{orig}}$ consists of the nearest one unit of original opening to $r$.
First we construct a feasible comparison assignment $\nu_r$ for every rounded outcome. Retain selected atoms in $T_r$ with total rounded value
$\min(1/2,y'(T_r))$; this is possible because $y'_e\in\{0,\delta\}$ for every atom and $1/2$ is a multiple of $\delta$. We use all rounded opening in $H_r$ and each retained physical tail atom at its original facility. For every
$r'\in\mathcal R\setminus\{r\}$, collect all retained virtual atoms
created from $T_r^\orig\cap H_{r'}$, and allocate their
total rounded value to the opened facilities in $H_{r'}$. Such an allocation is possible because this demand is at most $1/2$ and $y'(H_{r'})=Y_{r'}\ge1/2$.

Therefore, the demand can be allocated without exceeding $y'_i$
at any facility $i\in H_{r'}$. Consequently, such
an assignment satisfies $\nu_r(\mathcal U)=y'(H_r)+\min(1/2,y'(T_r))=1$ and $0\le\nu_r\le\vecy'$.

Next we bound the expectation of the assignment outside the enlarged ball. Fix $t\ge R_r$. The distance from $r$ to each facility in $H_r$ is at most $R_r$, so, according to \Cref{lem:virtual_distance}, $\nu_r\bigl(H_r\setminus
\operatorname{ball}(r,(1+\xi)t)\bigr)=0.$

For any physical tail atom $e$, its location does not change.
If it contributes to the assignment outside radius $(1+\xi)t$,
then $d(r,e)>t$. Its contribution is at most $y'_e$, and
marginal preservation in pipage rounding gives $\mathbb E[y'_e]=\bar y_e
\le(1+\theta)y_e$.
% For the virtual atoms, fix $r'\ne r$, and let $\mathcal V_{r,r'}$ be the set of virtual atoms created from $T_r^{\mathrm{orig}}\cap H_{r'}$. According to construction, we have
% \[\sum_{e\in\mathcal V_{r,r'}}y_e
% =y^{\mathrm{orig}}(T_r^{\mathrm{orig}}\cap H_{r'}).\]

Recall from \Cref{lem:virtual_distance} that for every $u\in T_r\cap H_{r'}$ and $v\in H_{r'}$, \[
    d(r,v)\le(1+\xi)d(r,u).
\]
Define $T_r(r')$ as a set of virtual atoms in $T_r$ whose original facility belongs to $H_{r'}$. Thus, the contribution of atoms in $T_r\cap H_{r'}$ with respect to allocation $\nu$ is as follows:
\begin{align*}
    \E[\nu_r(H_{r'}\setminus\operatorname{ball}(r,(1+\xi)t))] & \le \sum_{e\in T_r(r'):d(r,e)>t} \E[y'_e] \\
    & \le (1+\theta)\sum_{e\in T_r\cap H_{r'}:d(r,e)>t} y_e \\
    & = (1+\theta) y(T_r(r')\setminus \operatorname{ball}(r,t)).
\end{align*}

Moreover, since $\mu_r$ uses the nearest unit of opening, for
every rounded outcome and every radius $u\ge0$, \[
    \mu_r\bigl(\mathcal U\setminus\operatorname{ball}(r,u)\bigr)
    =\left(1-y'(\operatorname{ball}(r,u))\right)_+
    \le\nu_r\bigl(\mathcal U\setminus\operatorname{ball}(r,u)\bigr).
\]

Therefore,
\begin{align*}
    \mathbb E\left[\mu_r\bigl(\mathcal U\setminus
    \operatorname{ball}(r,(1+\xi)t)\bigr)\right] & \le \mathbb E\left[\nu_r\bigl(\mathcal U\setminus
    \operatorname{ball}(r,(1+\xi)t)\bigr)\right] \\
    & \le \sum_{e\in T_r\cap \calE_\phys: d(r,e)>(1+\xi)t} \E[y'_e] +\sum_{r'\in \calR\setminus {r}} \E[\nu_r(H_{r'}\setminus\operatorname{ball}(r,(1+\xi)t))] \\
    & \le (1+\theta)\Bigg(y((T_r\cap \calE_\phys)\setminus \operatorname{ball}(r,(1+\xi)t)) \\
    & \qquad\qquad + \sum_{r'\in\calR\setminus\{r\}} y((T_r(r'))\setminus\operatorname{ball}(r,t)) \Bigg) \\
    & \le (1+\theta) y(T_r\setminus \operatorname{ball}(r,t)) \\
    & = (1+\theta) y^\mathrm{orig}(T^\mathrm{orig}_r\setminus \operatorname{ball}(r,t)) \\
    & = (1+\theta)(1-y^\mathrm{orig}(\operatorname{ball}(r,t)))_+ = (1+\theta)q_r(t). \qedhere
\end{align*}

\end{proof}

\subsection{Putting It All Together}
% \begin{gptadded}
\begin{proof}[Proof of Theorem~\ref{thm:preprocessing}]
If there are no clients, the zero vector suffices. Otherwise filtering
selects at least one representative, and the construction of $\nu_r$
above ensures at least one unit of rounded opening in every outcome.
Merging and capping preserve all ball deficits, as shown in
Step~\ref{step:pre-return}; thus a nearest unit still exists afterward.
They also preserve grid integrality and do not increase total opening.
Lemma~\ref{lem:good-event} proves the facility bound, and
Lemma~\ref{lem:connection-cost} proves the unconditional cost bound.

There are polynomially many original pieces after the head and tail
splits. Atomization uses at most a polynomial in $n$ times
$\delta^{-1}$ atoms, and influence splitting makes at
most $|\mathcal R|$ copies per atom. Each pipage step makes
at least one fractional coordinate integral; the laminar sets and
the filtering can be maintained in polynomial time. Since
$\delta^{-1}=\varepsilon^{-O(p)}$, the claimed running time follows.
\end{proof}
% \end{gptadded}

\section{A Client-Dependent Bound for the Single-Pivot Rounding Algorithm of \texorpdfstring{\cite{BGHL2026}}{[BGH+2026]}}\label{sec:client-dep}\label{sec:client-dependent-analysis}

Throughout this section, we take $p=1$. That is, we focus on the $k$-median problem. So, the connection cost of a client is its distance to the facility serving it.  Most notations in this section are independent of those from Section~\ref{sec:preprocessing}. 

The main goal of the section is to show that, for the $2$-LMP rounding algorithm of~\cite{BGHL2026}, the approximation ratio for a client $j$ is $2$ only when its connection distances in the support of the LP solution are the same. In particular, we define the \emph{dispersion} $\calD_j$ to measure how far away the distribution of distances is from a deterministic distance and prove an approximation ratio for $j$ that depends on $\calD_j$.

\subsection{Recap of Rounding Algorithm of~\texorpdfstring{\cite{BGHL2026}}{[BGH+2026]}}
First, we review the single-pivot rounding algorithm of~\cite{BGHL2026}.  The input to the rounding algorithm is an opening vector $\mathbf{y}\in[0,1]^F$ with $y(F)\ge1$. It runs in iterations and maintains an opening vector $\mathbf{z}^{(t)}$ at the end of iteration $t$.  Initially, we have $\mathbf{z}^{(0)}=\mathbf{y}$. For every $t = 0, 1, 2, \cdots$, the iteration $t + 1$ of the algorithm runs as follows. For every facility $i$, the \emph{nearest unit of opening} of $i$ is described by weights $w_{qi}^{(t)}$, $q\in F$, obtained by
collecting the opening of $\mathbf{z}^{(t)}$ in nondecreasing order of distance from $i$, breaking ties consistently and in favor of $i$ itself, until one unit is collected. Thus
\[
 0\le w_{qi}^{(t)}\le z_q^{(t)},
 \qquad
 \sum_{q\in F}w_{qi}^{(t)}=1,
 \qquad
 w_{ii}^{(t)}=z_i^{(t)}.
\]

Let $\mathcal F_t$ be the history through iteration $t$, including all
random choices made so far; $\mathcal F_0$ is trivial.  In iteration $t+1$, the pivot $q$ is chosen with probability
$z_q^{(t)}/{|\mathbf{z}^{(t)}|_1}$. The pivot is opened (its opening becomes $1$), and each other facility $i$ is deleted (its opening becomes $0$) with probability $w_{qi}^{(t)}/z_q^{(t)}$. Indeed, the analysis uses only these conditional marginals, not the joint
distribution of the deletions. 

The algorithm terminates when the opening vector becomes integral. It is possible that an iteration does not change the opening vector. So it may happen that, for any finite $T$ (which may depend on the instance), with positive probability, the algorithm does not terminate after $T$ iterations.  Using conditional probabilities, one can easily avoid these ``useless'' iterations. However, for the purposes of the analysis, it is convenient for us to keep these iterations. Moreover, we define $\mathbf{z}^{(t)}$ over all integers $t \geq 0$: if the algorithm terminates in less than $t$ iterations, then $\mathbf{z}^{(t)}$ is the $\mathbf{z}$-vector when the algorithm terminates, which is integral.

The following simple lemma is needed later:
\begin{lemma}[Almost-sure termination]\label{lem:client-termination}
	The number of iterations the algorithm runs is finite almost surely.
\end{lemma}
\begin{proof}
	Let $n_F:=|F|$.  Conditional on $\mathcal F_t$, a facility $i$ that is
	fractional at time $t$ becomes integral in the next iteration either by
	being chosen as the pivot or by being deleted by another pivot.  Since
	the nearest-unit weights satisfy $w_{ii}^{(t)}=z_i^{(t)}$ and
	$\sum_q w_{qi}^{(t)}=1$, this has probability
	\[
	\frac{z_i^{(t)}+\sum_{q\ne i}w_{qi}^{(t)}}{{|\mathbf{z}^{(t)}|_1}}
	=\frac1{{|\mathbf{z}^{(t)}|_1}}\ge\frac1{n_F},
	\]
	where the inequality uses ${|\mathbf{z}^{(t)}|_1}\le n_F$.
	No integral facility becomes fractional again.  Thus the probability
	that a given initially fractional facility remains fractional after
	$t$ iterations is at most $(1-1/n_F)^t$, which tends to $0$. Since there
	are finitely many facilities, the number of iterations is finite almost
	surely.
\end{proof}

The expected number of open facilities given by the algorithm is $|\mathbf{y}|_1$, which follows from a simple martingale argument.   The main focus of our analysis is on the expected connection cost of a client $j$. 

\subsection{Main Theorem and Notations}\label{subsec:client-termination}
Till the end of Section~\ref{sec:client-dep}, we fix a client $j$, and our goal is to give the expected connection cost of $j$.  After the standard facility-splitting step, we may assume that
$x_{ij}\in\{0,y_i\}$ for every $i\in F$; let
$F_j:=\{i\in F:x_{ij}>0\}$ be the set of supporting facilities  of $j$. Then
$y(F_j)=1$ and $y_i>0$ for
every $i\in F_j$.  Write $d_i=d(i,j)$, and order the facilities in
$F_j$ so that $i<q$ implies $d_i\le d_q$, breaking ties arbitrarily.
Define
\[
 V_j:=\sum_{i\in F_j}y_i d_i,
 \qquad
 \Disp_j:=\sum_{i,q\in F_j}y_i y_q|d_i-d_q|
 =2\sum_{i<q}y_i y_q(d_q-d_i).
\]

So, $V_j$ is the fractional cost of $j$ in the LP solution $\mathbf{y}$. We call $\calD_j$ the \emph{dispersion} of the client $j$, which measures how different the distances $d_i$'s are.  Let $c_j$ denote the connection cost of $j$ given by the rounding algorithm. The main theorem we shall prove is:
\begin{theorem}[Main Theorem]\label{thm:client-dep}
$ \E[c_j]\le 2V_j-\frac{\Disp_j}{12}$.
\end{theorem}

If $y_i=1$ for some $i\in F_j$, then $F_j=\{i\}$ and this facility
is already open and remains open. Thus $c_j\le d_i=V_j$ and
$\Disp_j=0$, proving the theorem in this case. Henceforth assume
$0<y_i<1$ for every $i\in F_j$.

Consequently, the approximation factor for $j$ is at
most $2-\Disp_j/(12V_j)$ (assuming $V_j>0$).  %This bound is strictly below $2$ unless all facilities of $F_j$ are at the same distance from $j$. Since $\Disp_j\le2V_j$ by~\eqref{eq:disp-bound} below, it always lies in $[11/6,2]$. 
The rest of this section is dedicated to the proof of the theorem. If $F_j$ is a singleton, then $\Disp_j=0$ and the theorem is immediate. Henceforth assume $|F_j|\ge2$; in particular, $0<y_i<1$ for every $i\in F_j$. \medskip

%\paragraph{Proof strategy.}
%We use a potential that is defined at every state of the algorithm.
%Let $T$ be the set of support facilities that are still fractional, and
%let $b$ be the backup distance, that is, the distance from $j$ to its
%nearest integrally open facility. For $U\subseteq T$, the factor-$2$
%potential of~\cite{BGHL2026} is $\Phi(U,b)=2V(U)+(1-y(U))b$. We subtract
%a dispersion credit $\Disp(U)/12$ and minimize over $U\subseteq T$. The
%resulting potential $P_\star(T,b)$ is at most $2V_j-\Disp_j/12$ initially
%and equals $c_j$ when the algorithm terminates, so it suffices to show
%that it is a supermartingale (Lemma~\ref{lem:closed-potential}). For a
%minimizing set $U$, one iteration decreases the expected factor-$2$
%potential by at least $\Gamma/{|\mathbf{z}^{(t)}|_1}$ (Lemma~\ref{lem:one-step}), where
%$\Gamma\ge\Disp(U)/4$ (Lemma~\ref{lem:combined-saving}), while the
%expected dispersion credit decreases by at most $3\Disp(U)/{|\mathbf{z}^{(t)}|_1}$
%(Lemma~\ref{lem:persistence}). The gain $\Disp(U)/(4{|\mathbf{z}^{(t)}|_1})$ exactly pays
%for the loss $\frac1{12}\cdot3\Disp(U)/{|\mathbf{z}^{(t)}|_1}$. The constant $12=4\cdot3$
%comes from these two lemmas; improving either of them improves the
%theorem directly.
%
%Unlike the analysis of~\cite{BGHL2026}, we do not stop when a support
%facility is selected as the pivot. The minimum over $U$ takes care of
%the support facilities that become useless once the backup distance
%drops: they are simply left out of the minimizing set.

For every $U\subseteq F_j$, let
\[
 {V(U)}:=\sum_{i\in U}y_i d_i,
 \qquad
 \Disp(U):=\sum_{i,i'\in U}y_iy_{i'}|d_i-d_{i'}|.
\]
Note that $V(F_j)=V_j$ and $\Disp(F_j)=\Disp_j$.  Since
$|d_i-d_{i'}|\le d_i+d_{i'}$, we have
\begin{equation}\label{eq:disp-bound}
	\Disp(U)\le\sum_{i,i'\in U}y_iy_{i'}(d_i+d_{i'})=2y(U){V(U)}.
\end{equation}

For every $U\subseteq F_j$ and
$b\in\mathbb{R}_{\ge0}\cup\{\infty\}$,
we define
\[\Phi(U,b):=2V(U)+(1-y(U))b,\]
with the convention that $0\cdot\infty=0$.  This is the same as the definition in \cite{BGHL2026}. \medskip
 %For fixed $U$, the potential $\Phi(U,b)$ is nondecreasing in $b$.  It charges each facility $i\in U$ twice its distance, and the backup distance $b$ to the missing mass $1-y(U)$. 

%
%We have $T^{(0)}=F_j$. Every iteration opens its pivot, so $\add{b^{(t)}}\le d_{\max}(j):=\max_{i\in F}d(i,j)$ for every $t\ge1$. At $\sigma_{\mathrm{alg}}$ every facility is integral, so $\add{T^{(\sigma_{\mathrm{alg}})}}=\varnothing$ and $\add{b^{(\sigma_{\mathrm{alg}})}}=c_j$.

For every $t \geq 0$, let $T^{(t)} = \{i \in F_j: z_i^{(t)} = y_i\}$. Let $b^{(t)}$ be the distance from $j$ to its
nearest integrally open facility at the end of iteration $t$; set $b^{(t)}=\infty$ if no facility is
integrally open, which could happen only when $t = 0$.  Let $S^{(t)} = \{i \in T^{(t)}: 2d_i < b^{(t)} \}$, which is the subset $S$ of $T^{(t)}$ that minimizes $\Phi(S, b^{(t)})$. 

%\gptremovedfile{removed-text/client-dep-03}

Assuming the metric is finite, we have that $b^{(t)}<\infty$ for every $t\ge1$.
Let $t^*$ be the first iteration in which the pivot lies in $F_j$
or all facilities of $F_j$ have been deleted. This is a stopping
time for the algorithm's history. Before $t^*$, a facility in $F_j$
can leave $T^{(t)}$ only by deletion. Thus any positive-mass pivot
outside $T^{(t)}$ lies outside $F_j$. Moreover,
$1\le t^*<\infty$ almost surely, by
Lemma~\ref{lem:client-termination}. For $0\le t<t^*$ define

\begin{equation}\label{eq:psi-def}
	\Psi^{(t)}:=\min_{U\subseteq T^{(t)}}\left\{\Phi(U,b^{(t)})-\frac{\Disp(U)}{12}\right\},
\end{equation}
and $U^{(t)}$ to be the set $U\subseteq T^{(t)}$ achieving the value of $\Psi^{(t)}$.   However, when $t \geq t^*$,  we define $\Psi^{(t)} = \Phi(S^{(t)}, b^{(t)})$, in which case $U^{(t)}$ is not defined.

% \begin{gptadded}
% By \eqref{eq:disp-bound}, every candidate in \eqref{eq:psi-def}
% is at least $(2-y(U)/6)V(U)+(1-y(U))b^{(t)}\ge0$.
% Taking $U=F_j$ initially gives
% $\Psi^{(0)}\le2V_j-\Disp_j/12$. At positive times before $t^*$,
% the empty candidate gives $\Psi^{(t)}\le b^{(t)}\le d_{\max}(j)$;
% the same bound holds at $t^*$ by minimizing $\Phi$.
% \end{gptadded}

\subsection{Saving in One Iteration}\label{sec:one-step-saving}
With the notations defined, we can now prove the main theorem.  We fix an integer $t$ and consider the iteration $t+1$ of the algorithm. We assume $t < t^*$; notice that the exact value of $t^*$ is not known yet.  We suppress the time subscript and write $\mathbf{z}=\mathbf{z}^{(t)}$, $T=T^{(t)}$, $U = U^{(t)}$, $b=b^{(t)}$, and $w_{qi}=w_{qi}^{(t)}$.

Recall that we choose a pivot $q\in F$ with probability $\frac{z_q}{|\mathbf{z}|_1}$; for $q\in U$, we have $z_q = y_q$ and thus the
probability is $\frac{y_q}{|\mathbf{z}|_1}$. For every $q\in F$ and $i  \in U$, the probability that $q$ is chosen
and deletes $i$ is $\frac{z_q}{|\mathbf{z}|_1} \cdot \frac{w_{qi}}{z_q}=\frac{w_{qi}}{|\mathbf{z}|_1}$. For $i\in U$, write
\[
 w(U,i):=\sum_{q\in U}w_{qi},
 \qquad
 u_{qi}:=y_q-w_{qi}, \forall q\in U.
\]
$w(U, i)$ is the unnormalized probability that $i$ will be removed by pivots in $U$, and $u_{qi}$ is the extent to which $q$ does not remove $i$.

Since $w_{qi}\le z_q=y_q$ for $q\in U$, we have $u_{qi}\ge0$ and
$\sum_{q\in U}u_{qi}={y(U)}-w(U,i)\ge0$. Define
\[
 r_i:=\min\{d_q:q\in U,\ u_{qi}>0\},
\]
where the minimum of the empty set is $\infty$. Then $r_i=\infty$ exactly
when ${y(U)}=w(U,i)$, and in that case we interpret $({y(U)}-w(U,i))r_i$ as $0$.

\begin{lemma}\label{lemma:backup-bound}
	Suppose some $q \notin U$ and $i \in U$ have $w_{qi} > 0$. Then, $d(j, q) \leq 2d_i + r_i$. 
\end{lemma}
\begin{proof}
We can assume $r_i<\infty$. Let $v\in U$ be such that $u_{vi}>0$ and
$d_v=r_i$. The nearest unit of $i$ is collected in nondecreasing order
of distance from $i$, and it uses $q$ but not all of $v$; hence
$d(i,q)\le d(i,v)$, and
\begin{equation*}
	d(j,q)\le d_i+d(i,q)
	\le d_i+d(i,v)
	\le 2d_i+d_v=2d_i+r_i. \qedhere
\end{equation*}	
\end{proof}

Consequently, whenever a pivot outside $U$ deletes $i\in U$, the new
backup distance is at most $\min\{b,2d_i+r_i\}$; this is trivial if
$r_i=\infty$.

%TODO: there is an issue with t vs t + 1.

Let $U' = \{i \in U: z^{(t + 1)}_i = y_i\}$ be the set of facilities of $U$ that are still
fractional after iteration $t + 1$, and let $b' = b^{(t+1)}$ be the new backup distance.  Notice that $U'$ may be different from $U^{(t+1)}$.

Define the next charge and the next dispersion credit of $U$ by
\begin{align}
\Phi'&:=
 \begin{cases}
 	d_q-\sum_{i\in U'}y_i\pos{d_q-2d_i}, & \text{if the pivot $q$ lies in $U$},\\
  \Phi(U',b'), & \text{otherwise},
 \end{cases}\nonumber\\
 	\Disp'&:=
 	\begin{cases}
 		0, & \text{if the pivot $q$ lies in $T$},\\
 		\Disp(U'), & \text{otherwise}.
 	\end{cases}	 \label{equ:Disp'}
\end{align}
Notice that the case conditions for $\Phi'$ are slightly different from those for $\Disp'$. 

 Finally, we define the saving terms
\[
 I:=\sum_{i,q\in U}y_iu_{qi}\pos{d_q-2d_i},
 \qquad
 \Gamma:={V(U)}-\sum_{i\in U}y_i\bigl({y(U)}-w(U,i)\bigr)r_i+I.
\]

\begin{lemma}[One-step saving]\label{lem:one-step}
${|\mathbf{z}|_1} \cdot \E[\Phi'\mid\mathcal F_t]\le {|\mathbf{z}|_1} \cdot \Phi(U,b)-\Gamma$.
\end{lemma}

\begin{proof}
All expectations and pivot probabilities below are conditional on
$\mathcal F_t$. First suppose that the pivot $q$ lies in $U$.
Conditional on this, a facility $i\in U$ survives with probability
$\frac{u_{qi}}{z_q} = \frac{u_{qi}}{y_q}$; for $i=q$ this is $0$, as $u_{qq}=0$. Therefore, the contribution of pivots
in $U$ to ${|\mathbf{z}|_1}\cdot \E[\Phi']$ is
\[
 \sum_{q\in U}y_q\Bigl(d_q-\sum_{i\in U}\frac{u_{qi}}{y_q}\,y_i\pos{d_q-2d_i}\Bigr)
 ={V(U)}-I.
\]

Now suppose that the pivot $q$ lies outside $U$. Then $U\setminus U'$
consists of the facilities deleted by $q$, and $b'\le b$.
By Lemma~\ref{lemma:backup-bound}, in every outcome,
\[
 \begin{aligned}
\Phi'=\Phi(U',b')
 &=2V(U')+(1-{y(U)})b'+\sum_{i\in U\setminus U'}y_ib'\\
 &\le 2V(U')+(1-{y(U)})b+\sum_{i\in U\setminus U'}y_i\min\{b,2d_i+r_i\}\\
 &= 2V(U) + (1-{y(U)})b+\sum_{i\in U\setminus U'}y_i\min\{b - 2d_i,r_i\}\\
 &= \Phi(U, b) +\sum_{i\in U\setminus U'}y_i\min\{b - 2d_i,r_i\}
 \end{aligned}
\]
The total
pivot mass outside $U$ that deletes $i$ is $\sum_{q\notin U}w_{q,i}=1-w(U,i)$. Hence the contribution of pivots outside $U$ to ${|\mathbf{z}|_1}\cdot \E[\Phi']$ is at most
\[
 {z(F\setminus U)}\,\Phi(U,b)+\sum_{i\in U}y_i\bigl(1-w(U,i)\bigr)\cdot \min\{b - 2d_i,r_i\}.
\]
Since $1-w(U,i)=({y(U)}-w(U,i))+(1-{y(U)})$ with both terms
nonnegative,
\[
 \bigl(1-w(U,i)\bigr)\cdot \min\{b - 2d_i,r_i\}
 \le\bigl({y(U)}-w(U,i)\bigr)r_i+(1-{y(U)})(b-2d_i).
\]
Summing $y_i$ times the above inequality over all  $i\in U$, the contribution of pivots outside $U$ is at most
\[
 {z(F\setminus U)}\cdot\Phi(U,b)+\sum_{i\in U}y_i\bigl({y(U)}-w(U,i)\bigr)r_i+(1-{y(U)})(y(U)b-2V(U)).
\]
Adding the two contributions, we obtain
\begin{align*}
 &\quad {|\mathbf{z}|_1}\cdot \E[\Phi']\\
 &\leq  V(U) - I + {z(F\setminus U)}\cdot\Phi(U,b)+\sum_{i\in U}y_i\bigl({y(U)}-w(U,i)\bigr)r_i+(1-{y(U)})(y(U)b-2V(U))\\
 &=-V(U) - I + {z(F\setminus U)}\cdot\Phi(U,b) +\sum_{i\in U}y_i\bigl({y(U)}-w(U,i)\bigr)r_i + y(U) \big(2V(U) + (1-y(U))b\big)\\
 &= z(F)\cdot \Phi(U,b)-{V(U)}-I+\sum_{i\in U}y_i\bigl({y(U)}-w(U,i)\bigr)r_i
 ={|\mathbf{z}|_1}\cdot\Phi(U,b)-\Gamma. \qedhere
\end{align*}
\end{proof}

\subsection{Bounding the Saving Term \texorpdfstring{$\Gamma$}{Gamma}}
In this section, we lower bound $\Gamma$ using the dispersion $\Disp(U)$.  Define 
\[
 H:=\max_{s\ge0}\left( s\sum_{q\in U:\,d_q\ge s}y_q\right).
\]
This can be viewed as the area of the largest rectangle in the distance histogram for $U$, with width $(y_q)_{q \in U}$ and heights $(d_q)_{q \in U}$.  Notice that the maximum is attained at $s=0$ or at some $d_q$. 

In the sums over pairs $i<q$ below, both facilities range over $U$.
\begin{lemma}[Pairwise rectangle bound]\label{lem:rectangle}
${y(U)}({V(U)}-H)\ge\sum_{i<q}y_iy_q\min\{d_i,d_q-d_i\}$.
\end{lemma}

\begin{proof}
Choose a threshold $s$ attaining $H$, and let $a_i:=d_i$ if $d_i<s$ and
$a_i:=d_i-s$ if $d_i\ge s$. Thus $a_i\ge0$ and $\sum_{i\in U}y_ia_i={V(U)}-H$.
Fix $i<q$. If $s>d_i$, then $a_i=d_i$. If $s\le d_i$, then
$a_q=d_q-s\ge d_q-d_i$. In either case,
$\min\{d_i,d_q-d_i\}\le a_i+a_q$. Multiplying by $y_iy_q$ and summing
over pairs $i<q$ gives
\begin{align*}
 &\quad\sum_{i<q}y_iy_q\min\{d_i,d_q-d_i\}
 \le\sum_{i<q}y_iy_q(a_i+a_q)\\
 &=\sum_{i\in U}y_ia_i({y(U)}-y_i)
 \le {y(U)}\sum_{i\in U}y_ia_i
 ={y(U)}({V(U)}-H).\qedhere
\end{align*}
\end{proof}

\begin{lemma}[Saving versus dispersion]\label{lem:combined-saving}
We have $\Gamma\ge\frac{\Disp(U)}4$. In particular,
$\Gamma\ge0$.
\end{lemma}

\begin{proof}
Fix $i\in U$. Every $q\in U$ with $u_{qi}>0$ satisfies $d_q\ge r_i$, and
$u_{qi}\le y_q$. This gives two bounds:
\[
 \bigl({y(U)}-w(U,i)\bigr)r_i=r_i\sum_{q\in U:\,d_q\ge r_i}u_{qi}\le H,
 \qquad
 \bigl({y(U)}-w(U,i)\bigr)r_i\le\sum_{q\in U}u_{qi}d_q={V(U)}-\sum_{q\in U}w_{qi}d_q.
\]
The first bound, $I\ge0$, and Lemma~\ref{lem:rectangle} give
\[
 \Gamma\ge {V(U)}-y(U)H\ge {y(U)}({V(U)}-H)\ge\sum_{i<q}y_iy_q\min\{d_i,d_q-d_i\}.
\]
The second bound gives
\[
 \Gamma\ge(1-{y(U)}){V(U)}+\sum_{i,q\in U}y_i\bigl(w_{qi}d_q+u_{qi}\pos{d_q-2d_i}\bigr)
 \ge\sum_{i,q\in U}y_iy_q\pos{d_q-2d_i},
\]
because
$w_{qi}d_q+u_{qi}\pos{d_q-2d_i}\ge(w_{qi}+u_{qi})\pos{d_q-2d_i}=y_q\pos{d_q-2d_i}$.
In particular, $\Gamma\ge0$. For every $i<q$, we have
$\min\{d_i,d_q-d_i\}+\pos{d_q-2d_i}=d_q-d_i$, while $\pos{d_q-2d_i}=0$
whenever $d_q\le d_i$. Consequently, the two lower bounds sum to
\[
 \sum_{i<q}y_iy_q(d_q-d_i)=\frac{\Disp(U)}{2},
\]
and their maximum is at least $\Disp(U)/4$.
\end{proof}

\subsection{Dispersion Persistence and Completion of the Proof}
Then, we show that the dispersion term $\Disp(U)$ does not decrease too fast. Recall the definition of $\Disp'$ in \eqref{equ:Disp'}.
\begin{lemma}[Dispersion persistence]\label{lem:persistence}
$\displaystyle \E[\Disp'\mid\mathcal F_t]\ge\Bigl(1-\frac3{|\mathbf{z}|_1}\Bigr)\Disp(U)$.
\end{lemma}

\begin{proof}
Fix an ordered pair $i,i'\in U$ with $i \neq i'$. The pair contributes $y_iy_{i'}|d_i - d_{i'}|$ to $\Disp'$ unless a pivot
in $T \supseteq U$ is selected (in which case $\Disp'$ becomes $0$) or a pivot deletes $i$ or $i'$. These
events have probabilities at most $\frac{y(T)}{|\mathbf{z}|_1} \leq \frac{1}{|\mathbf{z}|_1}$, $\frac1{|\mathbf{z}|_1}$, and $\frac1{|\mathbf{z}|_1}$,
respectively. A union bound therefore gives
\[
 \Pr\bigl[\text{the pair $i,i'$ contributes to $\Disp'$}\bigm|\mathcal F_t\bigr]\ge1-\frac3{|\mathbf{z}|_1}.
\]
This  proves the lemma.
\end{proof}

\begin{corollary}[Combined one-step inequality]\label{cor:combined-one-step}
$\displaystyle 
 \E\Bigl[\Phi'-\frac{\Disp'}{12}\Bigm|\mathcal F_t\Bigr]
 \le\Phi(U,b)-\frac{\Disp(U)}{12}$.
\end{corollary}

\begin{proof}
By Lemmas~\ref{lem:one-step} and~\ref{lem:combined-saving},
$\E[\Phi' \mid\mathcal F_t]\le\Phi(U,b)-\frac{\Disp(U)}{4|\mathbf{z}|_1}$. By
Lemma~\ref{lem:persistence},
$\E[\Disp'\mid\mathcal F_t]\ge\Disp(U)-\frac{3\Disp(U)}{|\mathbf{z}|_1}$. Subtracting $\frac1{12}$ times the second bound from the first proves the corollary.
\end{proof}

%\add{Replacing $\Disp'$ by the larger quantity $\Disp(U')$ when the pivot lies in $T\setminus U$ only decreases the left-hand side, so \Cref{cor:combined-one-step} also holds in that form; this is the form used in \Cref{sec:single-pivot-conversion}.}

% The second label keeps the reference in the proof strategy above valid.
Recall that we assumed that $t < t^*$. We prove the following lemma:
\begin{lemma}[Main lemma]\label{lem:main-potential}\label{lem:closed-potential}
We have
\[
 \E[\Psi^{(t+1)}\mid\mathcal F_t]\le\Psi^{(t)}.
\]
That is, the stopped sequence $(\Psi^{(\min\{t, t^*\})})_{t\ge0}$ is a
supermartingale.
\end{lemma}

\begin{proof}
The probability space in this proof is conditional on $\mathcal F_t$.  We claim that the following inequality always holds:
\[
 \Psi^{(t+1)}\le\Phi'-\frac{\Disp'}{12}.
\]
Recall that $\mathbf{z} = \mathbf{z}^{(t)}, U = U^{(t)}, T = T^{(t)}, b = b^{(t)}$ and $U \subseteq T$. $U' = \{i \in U: z^{(t+1)}_i = y_i\} \subseteq T^{(t+1)}$ and $b ' = b^{(t+1)}$. We distinguish three cases, depending on the pivot we choose on iteration $t + 1$. 
\begin{itemize}
\item The pivot $q$ lies in $U$. Then $t^* = t+1$, $\Disp'=0$, and $b'\le d_q$.
\begin{align*}
	 \Psi^{(t+1)} &= \Psi^{(t^*)} = \Phi(S^{(t^*)}, b') = b' - \sum_{i \in T^{(t^*)}}y_i\pos{b'-2d_i} \\
&\leq d_q-\sum_{i\in T^{(t^*)}}y_i\pos{d_q-2d_i}
\le d_q-\sum_{i\in U'}y_i\pos{d_q-2d_i}=\Phi'.	
\end{align*}
The first inequality follows from a simple derivative technique. 
\item The pivot lies in $T\setminus U$. Then $t^* = t + 1$, $\Disp'=0$,
and $\Psi^{(t+1)} = \Phi(S^{(t^*)}, b') \le\Phi(U',b')=\Phi'$. The inequality used that $S^{(t^*)}$ is the subset $S$ of $T^{(t^*)}$ that minimizes $\Phi(S, b')$ and $U' \subseteq T^{(t^*)}$.
\item The pivot lies outside $T$. Then $\Phi'=\Phi(U',b')$ and
$\Disp'=\Disp(U')$. If $T^{(t+1)}=\varnothing$, then $t^* = t + 1$,
$U'=\varnothing$, $\Disp' = \Disp(U') = 0$ and $\Psi^{(t+1)}=b'=\Phi'-\Disp'/12$. Otherwise,
$t+1<t^*$, and $U'$ is a candidate in~\eqref{eq:psi-def}, so
$\Psi^{(t+1)}\le\Phi(U',b')-\frac{\Disp(U')}{12}=\Phi'-\frac{\Disp'}{12}$.
\end{itemize}
This proves the claim. Taking conditional expectations and applying
\Cref{cor:combined-one-step} gives
$\E[\Psi^{(t+1)}\mid\mathcal F_t]\le\Phi(U,b)-\frac{\Disp(U)}{12}=\Psi^{(t)}$.
\end{proof}

At time $t^*$, we use the factor-2 analysis of \cite{BGHL2026}. For completeness, we recover the proof using the potential argument we already established.

\begin{lemma}[Factor-$2$ continuation]\label{lem:factor2-continuation}
$\E[c_j\mid\mathcal F_{t^*}]
\le\Phi(S^{(t^*)},b^{(t^*)})$.
\end{lemma}
\begin{proof}
For $t\ge t^*$ we define
$Q_t=\min_{U\subseteq T^{(t)}}\Phi(U,b^{(t)})$, and fix a set $U$ minimizing the quantity. The proof of
Lemma~\ref{lem:one-step} only requires that $U$ is fractional and
$y(U)\le1$, so it applies in this state as well. Together with
Lemma~\ref{lem:combined-saving}, it gives
$\E[\Phi'\mid\mathcal F_t]\le\Phi(U,b^{(t)})=Q_t$.
If the pivot lies outside $U$, then
$Q_{t+1}\le\Phi(U',b')=\Phi'$. If the pivot is $q\in U$, use
$b'\le d_q$ and minimize over the surviving $U'$ to obtain
\[
 Q_{t+1}\le b'-\sum_{i\in U'}y_i(b'-2d_i)_+
 \le d_q-\sum_{i\in U'}y_i(d_q-2d_i)_+=\Phi'.
\]
The middle function is nondecreasing because $y(U')\le1$.
Thus $Q_t$ is a supermartingale after $t^*$ with non-negative support. At termination $T=\varnothing$ and $Q=c_j$.
Conditional bounded optional stopping proves the claim.
\end{proof}

\begin{proof}[Proof of Theorem~\ref{thm:client-dep}]
By Lemma~\ref{lem:main-potential}, $(\Psi^{(\min\{t, t^*\})})_{t\ge0}$ is a
supermartingale. It is bounded, since $0\le\Psi^{(0)}\le2V_j$ and
$0\le\Psi^{(t)}< \infty$ for $1\le t\le t^*$. Since $t^*$ is finite
almost surely, the optional stopping theorem gives
\[
 \E\bigl[\Psi^{(t^*)}\bigr]\le\Psi^{(0)}\le2V_j-\frac{\Disp_j}{12}.
\]
By Lemma~\ref{lem:factor2-continuation},
$\E[c_j\mid t^*, \mathcal F_{t^*}]\le\Phi(S^{(t^*)},b^{(t^*)})=\Psi^{(t^*)}$.
Taking expectations proves the theorem.
\end{proof}

\section{Euclidean LMP Approximation}\label{sec:euclidean-lmp}

In this section, we analyze the single-pivot rounding algorithm of~\cite{BGHL2026},
reviewed in Section~\ref{sec:client-dep}, in Euclidean space. We first give a sufficient condition for the LMP guarantee in a
general metric and then verify this condition in Euclidean space.

\subsection{Rounding and the Potential Function}

For the rounding process, we retain the opening vectors
$\mathbf z^{(t)}$, the pivot $q$, the nearest unit of opening $w_{qi}^{(t)}$ and the history $\mathcal F_t$ from
Section~\ref{sec:client-dep}. We continue to update the potential
after a facility in 
$F_j$ has been selected,
extending the connection cost analysis of~\cite{BGHL2026}.

\begin{lemma}\label{lem:nonstop-lmp}
Let $p\ge1$ and $\alpha\ge1$. Suppose that for every client $j$,
every set $T\subseteq F$, and every nonnegative vector of weights
$(y_i)_{i\in T}$, we have
\begin{equation}\label{eq:pairwise}
\sum_{i,i'\in T}y_i y_{i'}
\pos{(d_i+d_{ii'})^p-\alpha d_i^p}
\le
\sum_{i,i'\in T}y_i y_{i'}\pos{\alpha d_i^p-d_{i'}^p},
\end{equation}
where $d_i=d(i,j)$ and $d_{ii'}=d(i,i')$.
Given a feasible fractional solution $(x,y)$, the single-pivot rounding
algorithm reviewed in Section~\ref{sec:client-dep} returns a random set
$O$ such that
\begin{equation}\label{eq:nonstop-rounding-guarantee}
\Pr[i\in O]\le y_i\quad(i\in F),\qquad
\E[d(j,O)^p]\le\alpha\sum_{i\in F}x_{i,j}d(i,j)^p\quad(j\in C).
\end{equation}
Consequently, if $y(F)\le k$, the algorithm gives an $\alpha$-LMP approximation.
\end{lemma}

\begin{proof}
First consider the opening variables. Fix an iteration $t+1$, writing
$\mathbf z=\mathbf z^{(t)}$, $\mathbf z'=\mathbf z^{(t+1)}$, and
$w_{qi}=w_{qi}^{(t)}$. For $z_i>0$, the deletion probability
of $i$ is $\sum_q w_{qi}/|\mathbf z|_1=1/|\mathbf z|_1$. When $i$ is selected, its opening
becomes $1$. Therefore,
$$
\E[z_i'-z_i\mid \mathbf z]
=-\frac{z_i}{|\mathbf z|_1}\sum_{q\in F}w_{qi}+\frac{z_i}{|\mathbf z|_1}=0.
$$
Thus each opening variable
preserves its expectation. This gives $\Pr[i\in O]\le y_i$ for every $i\in F$.

By splitting facilities at the beginning, we can assume $x_{ij} \in \{0, y_i\}$ for every $i \in F$. Fix a client $j$. We describe the analytical state of $j$ by $(S,b)$,
where $S\subseteq F_j$ is the set of alive facilities retained in the
analysis and $b$ is the distance from $j$ to its nearest integrally open
facility. If no facility is open, let $b=\infty$. Define the potential
function
$$
\Phi_p(S,b):=\alpha\sum_{i\in S}y_i d_i^p+(1-y(S))b^p,
$$
which agrees with the potential $\Phi$ from
Section~\ref{sec:client-dep} when $p=1$ and $\alpha=2$.
We follow the convention that $0\cdot\infty^p=0$. The initial state is
$(F_j,\infty)$, with
$$
\Phi_p(F_j,\infty)=\alpha\sum_{i\in F_j}y_i d_i^p.
$$
We maintain the invariants
\begin{equation}\label{eq:nonstop-invariants}
0<z_i=y_i<1,\qquad \alpha d_i^p\le b^p\qquad(i\in S).
\end{equation}

Suppose that $q$ is selected in the current iteration, and let
$b_q:=\min\{b,d_q\}$. The new state is $(S',b')$, where $b'=b_q$ and
$S'$ is obtained from $S$ by removing every facility deleted in this
iteration, as well as every facility with $\alpha d_i^p>b_q^p$. These updates
preserve~\eqref{eq:nonstop-invariants}. The additional removals from
$S$ affect only the analytical state.

We show that the expected potential does not increase.
All expectations below are conditioned on $\mathcal F_t$.
Expanding the potential, we obtain
\begin{align*}
&\E[\Phi_p(S',b')]-\Phi_p(S,b)\\
&\quad=\E\!\left[
 \alpha\sum_{i\in S'}y_i d_i^p+(1-y(S'))b'^p
 \right]-\Phi_p(S,b)\\
&\quad=\E\!\left[
 -\alpha\sum_{i\in S\setminus S'}y_i d_i^p
 +(1-y(S'))b'^p-(1-y(S))b^p
 \right]\\
&\quad=\E\!\left[
 -\alpha\sum_{i\in S\setminus S'}y_i d_i^p
 +y(S\setminus S')b'^p
 +(1-y(S))(b'^p-b^p)
 \right]\\
&\quad=\E\!\left[
 \sum_{i\in S\setminus S'}y_i(b'^p-\alpha d_i^p)
 +(1-y(S))(b'^p-b^p)
 \right].
\end{align*}

Conditional on selecting $q$, a facility $i\in S$ is removed with
probability $1$ if $b_q^p<\alpha d_i^p$, and with probability
$w_{qi}/z_q$ otherwise. Since $q$ is selected with probability $z_q/|\mathbf z|_1$,
we obtain
\begin{align}\label{eq:nonstop-drift}
|\mathbf z|_1\bigl(\E[\Phi_p(S',b')]-\Phi_p(S,b)\bigr)
={}&\sum_{i\in S}y_i\sum_{q\in F}w_{qi}
       \pos{b_q^p-\alpha d_i^p}\nonumber\\
&-\sum_{i\in S}y_i\sum_{q\in F}z_q
       \pos{\alpha d_i^p-b_q^p}\nonumber\\
&-(1-y(S))\sum_{q\in F}z_q(b^p-b_q^p).
\end{align}

It remains to bound the first sum in~\eqref{eq:nonstop-drift}.
Fix $i\in S$, and consider the following two unit masses:
\begin{itemize}
    \item the mass consisting of $w_{qi}$ at each $q \in F$;
    \item the mass consisting of $y_s$ at each $s \in S$ and $1 - y(S)$ at the backup facility $o \in F$ at distance $b$ from $j$.
\end{itemize}
% Consider the unit mass consisting of $y_s$ at each
% $s\in S$ and $1-y(S)$ at the backup $o$ at distance $b$ from $j$.
Observe that, for any $R \geq 0$, the total amount accumulated from the facilities at distance within $R$ from $i$ in the first mass is exactly
\[
    \sum_{q \in F : d(q, i) \leq R} w_{qi} = \min \left\{ \sum_{q \in F : d(q, i) \leq R} z_{q}, 1 \right\}.
\]
Note also that the second unit mass does not exceed $\mathbf{z}$ coordinate-wise since
\[
    y_s = z_s \text{ for all $s \in S$} \quad \text{and} \quad 1-y(S) \leq 1 = z_o,
\]
implying that the total amount accumulated from the facilities at distance within $R$ from $i$ in the second mass is no greater than that in the first mass.
Therefore, by sorting $F$ and $S \cup \{o\}$ in the nondecreasing order of distance from $i$ for the two masses, respectively, and greedily matching the two masses in these orders, we can construct a nonnegative coupling $\pi^{(i)}$ between the two masses satisfying
\begin{itemize}
    \item the marginals are preserved, i.e.,
    \[
    \sum_{s\in S\cup\{o\}}\pi^{(i)}_{qs}=w_{qi} \text{ for all $q\in F$},
    % \\
    \quad 
    \sum_{q\in F}\pi^{(i)}_{qs}=y_s \text{ for all $s \in S$},
    \quad 
    \sum_{q\in F}\pi^{(i)}_{qo}=1-y(S);
    \]
    \item for any pair $(q, s) \in F \times (S \cup \{o\})$ with positive $\pi^{(i)}_{qs} > 0$, $d(i,q) \le d(i,s)$.
\end{itemize}
% This mass is feasible under $\mathbf z$, since $z_s=y_s$ for $s\in S$.
% Since $w_{\cdot i}$ is the nearest one unit of opening to $i$,
% matching the two masses in nondecreasing order of distance from $i$
% therefore gives a nonnegative coupling $\pi^{(i)}$ satisfying
% \begin{gather*}
% \sum_{s\in S\cup\{o\}}\pi^{(i)}_{qs}=w_{qi} \text{ for all $q\in F$},\\
% \sum_{q\in F}\pi^{(i)}_{qs}=y_s
% \qquad(s\in S),\qquad
% \sum_{q\in F}\pi^{(i)}_{qo}=1-y(S),
% \end{gather*}
% Moreover, for every $s \in S$, the coupling satisfies
% $d(i,q) \le d(i,s)$ whenever $\pi^{(i)}_{qs} > 0$.
% For mass matched to a facility $s\in S$, 
For any pair $(q, s)$ with $\pi^{(i)}_{qs} > 0$ and $s \in S$,
the triangle inequality gives
\[
    b_q\le d(j,q)\le d_i+d(i,q)\le d_i+d(i,s),
\]
and hence
$$
\pos{b_q^p-\alpha d_i^p}
\le\pos{(d_i+d(i,s))^p-\alpha d_i^p}.
$$
% For mass matched to $o$,
On the other hand, for any pair $(q, o)$ with $\pi^{(i)}_{qo} > 0$, \eqref{eq:nonstop-invariants} and
$\alpha\ge1$ give
$$
\pos{b_q^p-\alpha d_i^p}
\le b^p-\alpha d_i^p\le b^p-d_i^p.
$$
Using the marginals of $\pi^{(i)}$ and summing over $i\in S$,
we obtain
\begin{align}\label{eq:nonstop-backup}
&\sum_{i\in S}y_i\sum_{q\in F}w_{qi}
       \pos{b_q^p-\alpha d_i^p}\nonumber\\
&\qquad=
 \sum_{i\in S}y_i
 \sum_{s\in S\cup\{o\}}\sum_{q\in F}\pi^{(i)}_{qs}
       \pos{b_q^p-\alpha d_i^p}\nonumber\\
&\qquad\le
 \sum_{i,s\in S}y_i y_s
       \pos{(d_i+d(i,s))^p-\alpha d_i^p}
 +(1-y(S))\sum_{i\in S}y_i(b^p-d_i^p).
\end{align}

Applying~\eqref{eq:pairwise} and using $b_s=d_s$ and $z_s=y_s$ for
$s\in S$, we have
\begin{align}\label{eq:nonstop-alive}
\sum_{i,s\in S}y_i y_s
       \pos{(d_i+d(i,s))^p-\alpha d_i^p}
&\le\sum_{i,s\in S}y_i y_s
       \pos{\alpha d_i^p-d_s^p}\nonumber\\
&\le\sum_{i\in S}y_i\sum_{q\in F}z_q
       \pos{\alpha d_i^p-b_q^p}.
\end{align}
Moreover,
\begin{align*}
(1-y(S))\sum_{i\in S}y_i(b^p-d_i^p)
&=(1-y(S))\sum_{i\in S}z_i(b^p-b_i^p)\\
&\le(1-y(S))\sum_{q\in F}z_q(b^p-b_q^p).
\end{align*}
Combining these bounds with~\eqref{eq:nonstop-backup} and
substituting into~\eqref{eq:nonstop-drift} proves
$$
\E[\Phi_p(S',b')]\le \Phi_p(S,b).
$$

An iteration that leaves $\mathbf z$ unchanged also leaves $(S,b)$ unchanged,
so we may condition on a change. As in Section~\ref{sec:client-dep},
there are at most $|F|$ iterations that change the opening vector.
At termination, $S=\varnothing$ by
\eqref{eq:nonstop-invariants}, and the potential is $b^p=d(j,O)^p$.
Applying the one-step inequality successively from the initial state gives
$$
\E[d(j,O)^p]\le \Phi_p(F_j,\infty)
=\alpha\sum_{i\in F_j}y_i d_i^p
=\alpha\sum_{i\in F}x_{i,j}d(i,j)^p.
$$
Finally, summing the opening bounds gives $\E[|O|]\le y(F)\le k$,
and summing the connection bounds gives the claimed LMP guarantee.
\end{proof}

\subsection{Pairwise Condition~\texorpdfstring{\eqref{eq:pairwise}}{} in
Euclidean Space}

We now verify the pairwise condition~\eqref{eq:pairwise} in
Euclidean space. Fix a client $j$ and a finite set $T$ of facilities
in Euclidean space. Place $j$ at the origin, and let $v_i$ be the
position of facility $i$. Define $d_i:=d(i,j)=\norm{v_i}$ and
$d_{ii'}:=d(i,i')=\norm{v_i-v_{i'}}$, and let $(y_i)_{i\in T}$ be
arbitrary nonnegative weights. Define the symmetric kernel
\begin{align}\label{eq:kernel}
K_{p,\alpha}(r,s,c):=\frac12\big[&
\pos{(r+c)^p-\alpha r^p}+\pos{(s+c)^p-\alpha s^p}\nonumber\\
&-\pos{\alpha r^p-s^p}-\pos{\alpha s^p-r^p}\big].
\end{align}
Note that, for $r=d_i$, $s=d_{i'}$, and $c=d_{ii'}$, 
% it follows that $y_iy_{i'}K_{p,\alpha}(d_i,d_{i'},d_{ii'})+y_{i'}y_iK_{p,\alpha}(d_{i'},d_{i},d_{ii'})$ is equal to \[y_iy_{i'}\left[((d_i+d_{ii'})^p-\alpha d_{i}^p)_+-(\alpha d_i^p-d_{i'}^p)_+\right] +y_{i'}y_i\left[((d_{i'}+d_{ii'})^p-\alpha d_{i'}^p)_+-(\alpha d_{i'}^p-d_{i}^p)_+\right].\] Thus, 
the kernel captures the symmetrized contribution of each pair of facilities. Our target pairwise inequality \eqref{eq:pairwise} can therefore be rewritten as 
\begin{equation}\label{eq:target-inequality}
    \sum_{i,i'\in T} y_iy_{i'} K_{p,\alpha}(d_i, d_{i'}, d_{ii'})\le 0.
\end{equation}

% We first treat $p=1$ and then $p\ge2$.
We show that this condition holds for a particular value of $\alpha< 1.71955$ when $p=1$ in \Cref{sec:euclid-k-median}, and for $\alpha=(3^p+2)/3$ when $p\ge 2$ in \Cref{sec:euclid-k-clustering}.

\subsection{Analysis for k-Median in Euclidean Space}\label{sec:euclid-k-median}

Throughout \Cref{sec:euclid-k-median}, we set $p=1$. Then, \[
    K_{1,\alpha}(r,s,c)=\frac{1}{2}\left[(c-\beta r)_+ +(c-\beta s)_+ -(\alpha r-s)_+ -(\alpha s-r)_+\right],
\]
where $\beta:=\alpha-1$ and, due to the triangle inequality, $|r-s|\le c\le r+s$.

\paragraph{Nonnegativity certificate.}
Define $T^+:=\{i\in T:d_i>0\}$.
To show the target~inequality~\eqref{eq:target-inequality} for $p=1$ and our choice of $\alpha$, we construct a certificate of the form\[
    \mathcal C:=\sum_{i,i'\in T^+} y_iy_{i'} J(d_i,d_{i'},d_{ii'}) \qquad\text{where}\qquad J(r,s,c)=\sqrt{rs}\left(A(r,s)\cos\theta+B(r,s)\right),
\]
$\cos \theta=\frac{r^2+s^2-c^2}{2rs}$, and $A,B:\mathbb R_{>0}^2\to\mathbb R$ are symmetric positive-semidefinite kernels on positive radii.
Our goal is then to show
\begin{equation}\label{eq:general_certificate}
    K_{1,\alpha}(r,s,c)\le -J(r,s,c) \qquad \forall r>0,s>0,c\in[|r-s|,r+s].
\end{equation}

% For the justification of term $\sqrt{rs}$, by AM-GM inequality, we have $(r+s)/2\ge \sqrt{rs}$. {\color{blue} TODO: Explain the justification in more detail.}

Observe that certificate kernel $J(r,s,c)$ is undefined when $rs=0$ because of the $\cos\theta$ term. The following lemma shows that this does not pose an issue for the relevant range of $\alpha$.
\begin{lemma}\label{lem:zero-vector}
If $rs=0$, $|r-s|\le c\le r+s$, and $3/2\le \alpha\le 2$, then $K_{1,\alpha}(r,s,c)\le 0$.
\end{lemma}
\begin{proof}
By symmetry, it suffices to consider the case $s=0$. The triangle inequality $|r-s|\le c\le r+s$ then forces $c=r$. Thus, for $\alpha\le 2$, we have $K_{1,\alpha}(r,0,r)=(3-2\alpha)r/2$, which is nonpositive whenever $\alpha\ge 3/2$.
\end{proof}

\begin{lemma}[Nonnegativity]\label{lem:nonnegativity}
If the kernels $A$ and $B$ defining $\mathcal C$ are positive semidefinite, then $\mathcal C\ge 0$.
\end{lemma}
\begin{proof}
    Let $\mathbf A:=(A(d_i,d_{i'}))_{i,i'\in T^+}$ and $\mathbf B:=(B(d_i,d_{i'}))_{i,i'\in T^+}$. By the positive-semidefiniteness of the kernels $A$ and $B$, both $\mathbf A$ and $\mathbf B$ are positive-semidefinite matrices.

    Recall that $v_i$ denotes the position of facility $i$, with client $j$ placed at the origin. Thus, $\lvert v_i\rvert=d_i$ and $\cos\theta=\frac{\langle v_i,v_{i'}\rangle}{|v_i||v_{i'}|}$, where $\theta$ is the angle between $v_i$ and $v_{i'}$. For each $i\in T^+$, define the unit vector $u_i:=v_i/\lvert v_i\rvert$. Let $\mathbf G$ be the Gram matrix with entries $G_{i,i'}=\langle u_i,u_{i'}\rangle$ for $i,i'\in T^+$.
    
    Now, define $\mathbf r\in \mathbb R^{|T^+|}$ by $r_i=y_i\sqrt{d_i}$ for every $i\in T^+$. We can then express the certificate as
    \begin{align*}
        \mathcal C & = \sum_{i,i'\in T^+} y_iy_{i'}\sqrt{d_id_{i'}}  \left({A(d_i,d_{i'})}\langle u_i,u_{i'}\rangle+B(d_i,d_{i'})\right) \\
        & = \mathbf r^T((\mathbf A\circ \mathbf G)+\mathbf B)\mathbf r \ge 0,
    \end{align*}
    where $\circ$ denotes the Schur product.
    
    Indeed, $\mathbf G$ is positive semidefinite as a Gram matrix. By the Schur product theorem, $\mathbf A\circ \mathbf G$ is also positive semidefinite. Since $\mathbf B$ is positive semidefinite, so is their sum $(\mathbf A\circ \mathbf G)+\mathbf B$, establishing the final inequality.
\end{proof}

\paragraph{The scope of \texorpdfstring{$\alpha$}{alpha} when using two positive-semidefinite kernels.}
We first observe that \eqref{eq:target-inequality} cannot hold for $\alpha\le 1$. Indeed, consider two facilities with $d_1=d_2=1$, $y_1=y_2=1$, and $d_{1,2}=2$. Then
\(
y_1y_2K_{1,\alpha}(d_1,d_2,d_{1,2})=y_2y_1K_{1,\alpha}(d_2,d_1,d_{1,2})=3-\alpha>0,
\)
while
\(
y_1^2K_{1,\alpha}(d_1,d_1,0)=y_2^2K_{1,\alpha}(d_2,d_2,0)=1-\alpha\ge0.
\)
Hence, the sum of all four terms is strictly positive, violating \eqref{eq:target-inequality}.
We therefore restrict our attention to $\alpha>1$.

For positive radii $r$ and $s$, we examine the following configurations under the assumption $1<\alpha<2$:
\begin{itemize}
    \item \emph{Two identical vectors:} Here, $r=s>0$ and $c=0$, so $\cos\theta=1$. Since $K_{1,\alpha}(r,r,0)=-\beta r$, satisfying \eqref{eq:general_certificate} requires
    \begin{equation}\label{eq:identical-vectors}
        A(r,r)+B(r,r)\le \beta.
    \end{equation}

    \item \emph{Two equal-radius antipodal vectors:} Here, $r=s>0$ and $c=2r$, so $\cos\theta=-1$. Since $K_{1,\alpha}(r,r,2r)=2(1-\beta)r$, satisfying \eqref{eq:general_certificate} requires
    \begin{equation*}
        A(r,r)-B(r,r)\ge 2(1-\beta)>0.
    \end{equation*}
    Since $B$ is positive semidefinite, we have $B(r,r)\ge 0$, which forces $A(r,r)>0$ for $1<\alpha<2$.

    \item \emph{Two equal-radius vectors separated by $c=2\beta r$:} Here, $r=s>0$ and $c=2\beta r$. By the law of cosines, $\cos\theta=1-2\beta^2$. This configuration is feasible because $0<\beta<1$, ensuring $|\cos\theta|\le 1$. Since $K_{1,\alpha}(r,r,2\beta r)=0$, satisfying \eqref{eq:general_certificate} requires
    \begin{equation}\label{eq:zero-kernel}
        B(r,r)\le (2\beta^2-1)A(r,r).
    \end{equation}
    Since $A$ and $B$ are positive semidefinite, $A(r,r),B(r,r)\ge 0$. Together with $A(r,r)>0$ from the preceding case, this forces $\beta\ge 1/\sqrt{2}$.

    Consequently, $\alpha\ge 1+1/\sqrt2>3/2$, which is strictly stronger than the lower bound in \Cref{lem:zero-vector}. Thus, we may safely exclude the case $rs=0$ when constructing the certificate.

    \item \emph{Two orthogonal vectors:} Consider $r=\alpha$, $s=1$, and $\cos\theta=0$. By the Pythagorean theorem, $c=\sqrt{\alpha^2+1}$. Substituting into the kernel gives
    \[
        K_{1,\alpha}(\alpha,1,\sqrt{\alpha^2+1})=\frac{1}{2}[(\sqrt{\alpha^2+1}-(\alpha-1)\alpha)_+ +\sqrt{\alpha^2+1}+2 -\alpha^2-\alpha].
    \]
    % Let $q$ be the positive root of $\sqrt{\alpha^2+1}=(\alpha-1)\alpha$, which is approximately $2.1069$.
    For $1<\alpha<2$, we have $\sqrt{\alpha^2+1}>(\alpha-1)\alpha$, and hence
    \[
        K_{1,\alpha}(\alpha,1,\sqrt{\alpha^2+1})=\sqrt{\alpha^2+1}+1-\alpha^2.
    \]
    Therefore, satisfying \eqref{eq:general_certificate} requires
    \begin{equation}\label{eq:gamma-lowerbound}
       \mathcal{L}(\alpha):=\frac{\sqrt{\alpha^2+1}+1-\alpha^2}{\sqrt{\alpha}} \le - B(\alpha,1).
    \end{equation}
\end{itemize}

Combining \eqref{eq:identical-vectors} and \eqref{eq:zero-kernel}, we obtain
\begin{equation}\label{eq:gamma-upperbound}
    B(r,r)\le \mathcal{U}(\alpha):=\beta-\frac{1}{2\beta}\qquad \forall r>0.
\end{equation}
Since every $2\times2$ principal submatrix of a positive-semidefinite matrix is positive semidefinite, its determinant must be nonnegative. This yields the off-diagonal bound $|B(r,s)|\le \sqrt{B(r,r)B(s,s)}$.

For $\alpha\in\left[1+\frac1{\sqrt2},2\right)$, combining this bound with \eqref{eq:gamma-lowerbound} and \eqref{eq:gamma-upperbound} gives
\[
\mathcal{L}(\alpha)\le -B(\alpha,1)
\le |B(\alpha,1)|
\le \sqrt{B(\alpha,\alpha)B(1,1)}
\le \mathcal{U}(\alpha).
\]
Consequently,
\begin{equation}\label{eq:scope-of-alpha}
    \frac{\sqrt{\alpha^2+1}+1-\alpha^2}{\sqrt{\alpha}}\le \alpha-1-\frac{1}{2(\alpha-1)}.
\end{equation}

\paragraph{Exponential and triangular-wave kernels.}
We introduce two kernels that will be used in our construction:
\begin{itemize}
\item \emph{Exponential kernel:} $K_\mathrm{Exp}(x,y;\rho):=e^{-\rho|x-y|}$, where $\rho>0$.
\item \emph{Triangular-wave kernel:} $K_\mathrm{TW}(x,y;L):=\tau_L(x-y)$, where $\tau_L(z):=1-\frac{2}{L}\operatorname{dist}(z,2L\mathbb Z)$ and $L>0$.
\end{itemize}
Here, $\operatorname{dist}(z,2L\mathbb Z)=\inf_{k\in \mathbb Z}|z-2Lk|$. In particular, $\tau_L(0)=1$ and $\tau_L(L)=-1$.

We next show that both kernels are positive semidefinite.
\begin{lemma}\label{lem:exp_is_psd}
For every $d\in \mathbb Z_{>0}$, $\mathbf x,\mathbf y\in\mathbb R^d$ and $\rho>0$, we have $\sum_{i,j\in [d]} y_iy_jK_\mathrm{Exp}(x_i,x_j;\rho)\ge 0$.
\end{lemma}
\begin{proof} We first show that
    $$\begin{aligned}
e^{-\rho |x - x'|} 
&= e^{-\rho(x + x')} \cdot e^{2\rho \min\{x,x'\}}\\
&= e^{-\rho(x + x')} \int_{-\infty}^{\min\{x,x'\}} 2\rho e^{2\rho t} \, \mathrm{d}t \\
&= \int_{-\infty}^{\min\{x,x'\}} 2\rho e^{-\rho(x + x')} e^{2\rho t} \, \mathrm{d}t \\
&= \int_{-\infty}^{\min\{x,x'\}} 2\rho e^{-\rho(x-t)} e^{-\rho(x'-t)} \, \mathrm{d}t.
\end{aligned}$$
Therefore, \[
    \sum_{i,j\in [d]} y_iy_je^{-\rho|x_i-x_j|}=2\rho\int_{-\infty}^\infty \left\lVert \sum_{i\in [d]:x_i\ge t} y_i e^{-\rho(x_i-t)}\right\rVert^2 \, dt \ge 0. \qedhere
\]
\end{proof}

\begin{lemma}\label{lem:tw_is_psd}
For every $d\in \mathbb Z_{>0}$, $\mathbf x,\mathbf y\in\mathbb R^d$ and $L>0$, $\sum_{i,j\in [d]} y_iy_jK_\mathrm{TW}(x_i,x_j;L)\ge 0$.
\end{lemma}
\begin{proof}
The even triangular wave has the absolutely
convergent Fourier expansion
$$
\tau_L(z)=\frac8{\pi^2}\sum_{n=0}^{\infty}
\frac{\cos((2n+1)\pi z/L)}{(2n+1)^2}.
$$
Indeed, its constant coefficient is zero, and its $m$-th cosine
coefficient is $4(1-(-1)^m)/(m^2\pi^2)$ for $m\ge1$.

Substituting this expansion and using
$\cos(x-y)=\cos x\cos y+\sin x\sin y$, we obtain
\begin{align*}   
&\sum_{i,j\in [d]}y_i y_{j}\,
\tau_L(x_i-x_j)\\
&\quad=
\frac8{\pi^2}\sum_{n=0}^{\infty}\frac1{(2n+1)^2}
\sum_{i,j\in [d]}y_i y_{j}\,
\cos\left(
\frac{(2n+1)\pi(x_i-x_{j})}{L}
\right)\\
&\quad=
\frac8{\pi^2}\sum_{n=0}^{\infty}\frac1{(2n+1)^2}
\Bigg[
\left(
\sum_{i\in [d]}y_i
\cos\left(\frac{(2n+1)\pi x_i}{L}\right)
\right)^2\\
&\hspace{48mm}
+
\left(
\sum_{i\in [d]}y_i
\sin\left(\frac{(2n+1)\pi x_i}{L}\right)
\right)^2
\Bigg]\\
&\quad\ge0. \qedhere
\end{align*}
\end{proof}

\subsubsection{Construction of the Certificate}\label{sec:cert-construction}
We determine the parameters by requiring the candidate certificate to satisfy \eqref{eq:general_certificate} for the smallest $\alpha\in \left[1+\frac1{\sqrt2},2\right)$ satisfying \eqref{eq:scope-of-alpha}. This value is attained when equality holds in \eqref{eq:scope-of-alpha}. Specifically, we establish \eqref{eq:target-inequality} for $p=1$ and the unique solution $\alpha^\star\in (1.71954,1.71955)$ to
\[
    \frac{\sqrt{\alpha^2+1}+1-\alpha^2}{\sqrt{\alpha}}=\alpha-1-\frac{1}{2(\alpha-1)}.
\]
Note that $\alpha^\star\in (1+1/\sqrt 2, 2)$. % {\color{blue} TODO: Show the uniqueness of $\alpha$.}

We now justify that $\alpha^\star$ is the smallest value satisfying \eqref{eq:scope-of-alpha}. The function $\mathcal{L}(\alpha)-\mathcal{U}(\alpha)$ is continuous on $\left[1+1/\sqrt2,2\right)$, strictly positive at $\alpha=1+1/\sqrt2$, and strictly negative at $\alpha=\sqrt3$. Hence, there exists a smallest $\alpha$ satisfying \eqref{eq:scope-of-alpha}, and it lies in $(1+1/\sqrt2,\sqrt3)$. At this smallest value, equality must hold; otherwise, continuity would imply that a slightly smaller $\alpha$ also satisfies \eqref{eq:scope-of-alpha}, contradicting minimality.

Throughout the remainder of \Cref{sec:cert-construction}, we fix $\alpha=\alpha^\star$. We choose $A(r,s)=\eta\cdot K_\mathrm{Exp}(\log r,\log s;\rho)$ and $B(r,s)=\gamma\cdot K_\mathrm{TW}(\log r,\log s;L)$, where $L:=\log \alpha>0$, $\eta\ge 0,\gamma\ge0$, and $\rho>0$. With these choices, \eqref{eq:general_certificate} becomes
\begin{equation*}
     K_{1,\alpha}(r,s,c)\le -\frac{\eta e^{-\rho|\log(r/s)|}}{2\sqrt{rs}}(r^2+s^2-c^2)-\gamma\sqrt{rs}\tau_L(\log(r/s)).
\end{equation*}

\paragraph{Deriving \texorpdfstring{$\gamma$}{gamma}.}
Consider the orthogonal configuration with $r=\alpha$, $s=1$, and $c=\sqrt{\alpha^2+1}$. In this case, $B(\alpha,1)=\gamma\tau_L(\log\alpha)=\gamma\tau_L(L)$. For the equal-radius configurations, we similarly have $B(r,r)=\gamma\tau_L(0)$.

By \eqref{eq:gamma-lowerbound} and \eqref{eq:gamma-upperbound}, we obtain
\[
    \mathcal{L}(\alpha)\le -\gamma\tau_L(L)=\gamma \qquad\text{and}\qquad \mathcal{U}(\alpha)\ge \gamma\tau_L(0)=\gamma.
\]
Since our choice of $\alpha$ satisfies $\mathcal{L}(\alpha)=\mathcal{U}(\alpha)$, both inequalities must hold with equality. Therefore,
\[
    \gamma=\mathcal{U}(\alpha)=\beta-\frac{1}{2\beta}.
\]
Moreover, $\gamma>0$ because $\beta>1/\sqrt{2}$ for our choice of $\alpha$.

\paragraph{Deriving \texorpdfstring{$\eta$}{eta}.}
For the equal-radius configurations, we have $A(r,r)=\eta$ and $B(r,r)=\gamma$. By \eqref{eq:identical-vectors} and \eqref{eq:zero-kernel}, we obtain
\[
    \eta+\gamma\le \beta \qquad\text{and}\qquad \gamma\le(2\beta^2-1)\eta.
\]
Substituting our choice of $\gamma=\beta-\frac{1}{2\beta}$ into these inequalities forces
\[
    \eta=\frac{1}{2\beta}.
\]
In particular, $\eta>0$ since $\beta>0$ for our choice of $\alpha$.

\paragraph{Deriving \texorpdfstring{$\rho$}{rho}.}
We now consider configurations with $r=\alpha$ and $s=1$, allowing the vectors to be non-orthogonal. Let $c_0=\sqrt{\alpha^2+1}$ denote the value of $c$ in the orthogonal configuration. For these configurations, we have
\[
    K_{1,\alpha}(\alpha,1,c)=\frac{1}{2}\left[(c-(\alpha-1)\alpha)_+ +2+c-\alpha-\alpha^2\right].
\]
Since $c_0>\alpha(\alpha-1)$, for $c$ sufficiently close to $c_0$, \eqref{eq:general_certificate} reduces to
\[
    c+1-\alpha^2\le \frac{\eta e^{-\rho L}}{2\sqrt\alpha}(c^2-c_0^2)+\gamma\sqrt\alpha.
\]
By our choices of $\eta$ and $\gamma$, this inequality holds with equality at $c=c_0$. Moreover, $c_0$ lies strictly inside the feasible interval $(\alpha-1,\alpha+1)$, so the inequality must hold for $c$ on both sides of $c_0$. Consequently, the derivatives of both sides must agree at $c=c_0$, yielding
\[
    1=\frac{\eta e^{-\rho L}}{\sqrt\alpha}\cdot c_0 \qquad\implies\qquad e^{\rho L}=\frac{\eta\sqrt{\alpha^2+1}}{\sqrt\alpha} \qquad\implies\qquad \rho=-\frac{1}{2L}\log \left(\frac{\eta^{-2}\alpha}{\alpha^2+1}\right).
\]
Finally, $\rho>0$ because $\eta^{-2}\alpha<\alpha^2+1$ for our choice of $\alpha$ and $\eta$.

\paragraph{The final certificate.}
Combining the preceding choices of parameters, we obtain the certificate kernel \[
    J(r,s,c):=\frac{e^{-\rho|\log(r/s)|}}{4\beta\sqrt{rs}}(r^2+s^2-c^2)+\gamma\sqrt{rs}\tau_L(\log(r/s))
\] where \[
    \beta=\alpha-1, \qquad \gamma=\beta-\frac{1}{2\beta}, \qquad L=\log\alpha, \qquad\text{and}\qquad \rho=-\frac{1}{2L}\log \left(\frac{4\beta^2\alpha}{\alpha^2+1}\right).
\]
Numerically, these parameters are given by \[
    \beta\approx0.71955, \qquad \gamma\approx0.02466, \qquad L\approx0.54206, \qquad \rho\approx0.09718.
\]

The validity of this certificate~\eqref{eq:general_certificate} is established in \Cref{lem:median-certificate}. Given this validity, the following lemma proves our claim.

\begin{lemma}\label{lem:euclidean-median}
For $\alpha=\alpha^\star$, we have $\sum_{i,i'\in T} y_iy_{i'}K_{1,\alpha}(d_i,d_{i'},d_{ii'})\le 0$.
\end{lemma}
\begin{proof}
    Recall that $T^+:=\{i\in T:d_i>0\}$. Since $3/2<\alpha<2$, \Cref{lem:zero-vector} implies
    \begin{align*}
        \sum_{i,i'\in T} y_iy_{i'}K_{1,\alpha}(d_i,d_{i'},d_{ii'})& \le \sum_{i,i'\in T^+} y_iy_{i'}K_{1,\alpha}(d_i,d_{i'},d_{ii'}) \\
        & \le -\sum_{i,i'\in T^+} y_iy_{i'}J(d_i,d_{i'},d_{ii'}) \tag{\Cref{lem:median-certificate}} \\
        & \le 0. \tag{\Cref{lem:nonnegativity}}
    \end{align*}
    To justify the last inequality, observe that both the exponential and triangular-wave kernels are positive semidefinite by \Cref{lem:exp_is_psd,lem:tw_is_psd}. Since composition with the logarithm and multiplication by nonnegative constants preserve positive semidefiniteness, the kernels $A$ and $B$ are also positive semidefinite. Thus, \Cref{lem:nonnegativity} applies, completing the proof.
\end{proof}

\subsubsection{Validity of the Certificate} The following lemma shows that \eqref{eq:general_certificate} holds. We also provide a formal proof of this lemma in Lean, available at {\url{https://github.com/phanuv/Euclidean-k-Median-Certificate}}.

\begin{lemma}\label{lem:median-certificate}
Let $\alpha$ be the unique solution in $(1.71954,1.71955)$ of
\begin{equation}\label{eq:median-alpha}
\sqrt{\alpha^2+1}+1-\alpha^2
=\sqrt\alpha\left(\alpha-1-\frac{1}{2(\alpha-1)}\right).
\end{equation}
Set
$$
\beta=\alpha-1,\qquad
\gamma=\beta-\frac1{2\beta},\qquad L=\log\alpha,
$$
$$
\rho=-\frac1{2L}\log\left(\frac{4\beta^2\alpha}{\alpha^2+1}\right),
\qquad
\tau_L(z)=1-\frac2L\operatorname{dist}(z,2L\mathbb Z).
$$
For every $r,s>0$ and $|r-s|\le c\le r+s$, we have
\begin{equation}\label{eq:certificate}
\begin{split}
K_{1,\alpha}(r,s,c)
\le{}&-\frac{e^{-\rho|\log(r/s)|}}{4\beta\sqrt{rs}}
(r^2+s^2-c^2)\\
&-\gamma\sqrt{rs}\,\tau_L(\log(r/s)).
\end{split}
\end{equation}

\end{lemma}

\begin{proof}
We first justify the definition of $\alpha$. On $(1.71954,1.71955)$,
the LHS of \eqref{eq:median-alpha} is strictly decreasing, since
its derivative is $\alpha/\sqrt{\alpha^2+1}-2\alpha<0$.
The RHS is strictly increasing, since its two factors are
positive and strictly increasing. At $\alpha=1.71954$, the LHS
minus the RHS is greater than
$$
1.98917515+1-1.71954^2-1.31131233\cdot0.02465160
>0.000031.
$$
At $\alpha=1.71955$, the same difference is less than
$$
1.98918381+1-1.71955^2-1.31131613\cdot0.02467125
<-0.000020.
$$
Continuity and monotonicity therefore give
existence and uniqueness of $\alpha$ in the stated interval.

The definitions give
\begin{equation}\label{eq:parameters}
\begin{gathered}
0.7195<\beta<0.7196,\qquad 0.0246<\gamma<0.0247,\\
0.542<L<0.543,\qquad 0.097<\rho<0.098.
\end{gathered}
\end{equation}

In particular, $0<\rho<1/10$ and $0<L<2L<3/2<2$.
The function $\tau_L$ is even and $2L$-periodic, with
$\tau_L(0)=1$, $\tau_L(L)=-1$, and $|\tau_L|\le1$.

By symmetry and homogeneity, it suffices to consider $r=e^z$ and
$s=1$, where $z\ge0$ and $e^z-1\le c\le e^z+1$.
After dividing by $e^{z/2}$ and moving the first term on the
RHS of~\eqref{eq:certificate} to the left, define
$$
\varphi(z,c):=e^{-z/2}K_{1,\alpha}(e^z,1,c)
+\frac{e^{-(1+\rho)z}}{4\beta}(e^{2z}+1-c^2).
$$
It remains to prove $\varphi(z,c)\le-\gamma\tau_L(z)$.
Since $\alpha e^z-1>0$, the definition of $K_{1,\alpha}$ gives

\begin{equation}\label{eq:phi-expanded}
\begin{split}
\varphi(z,c)
={}&\frac{e^{-z/2}}2\bigl[(c-\beta e^z)_++(c-\beta)_+
-\alpha e^z+1-(\alpha-e^z)_+\bigr]\\
&+\frac{e^{-(1+\rho)z}}{4\beta}(e^{2z}+1-c^2).
\end{split}
\end{equation}

We first obtain an upper bound for the case $c\ge\beta e^z$.
Both positive parts involving $c$ are then active. Expanding them and
collecting the terms containing $c$, we obtain

\begin{align}
\varphi(z,c)
={}&e^{-z/2}c-\frac{e^{-(1+\rho)z}}{4\beta}c^2
+\frac{e^{-(1+\rho)z}}{4\beta}(e^{2z}+1)\nonumber\\
&-\frac{e^{-z/2}}2
\bigl[\beta(e^z+1)+\alpha e^z-1+(\alpha-e^z)_+\bigr]\nonumber\\
={}&e^{-z/2}c-\frac{e^{-(1+\rho)z}}{4\beta}c^2
+\frac{e^{-\rho z}}{2\beta}\cosh z\nonumber\\
&-\frac{e^{-z/2}}2
\bigl[\beta(e^z+1)+\alpha e^z-1+(\alpha-e^z)_+\bigr]\label{parameters-ineq1}\\
={}&-\frac{e^{-(1+\rho)z}}{4\beta}
\left[c^2-4\beta e^{(1/2+\rho)z}c\right] + \frac{e^{-\rho z}}{2\beta}\cosh z\nonumber\\
&-\frac{e^{-z/2}}2
\bigl[\beta(e^z+1)+\alpha e^z-1+(\alpha-e^z)_+\bigr]\nonumber\\
={}&\beta e^{\rho z}
-\frac{e^{-(1+\rho)z}}{4\beta}
\left(c-2\beta e^{(1/2+\rho)z}\right)^2+ \frac{e^{-\rho z}}{2\beta}\cosh z\nonumber\\
&-\frac{e^{-z/2}}2
\bigl[\beta(e^z+1)+\alpha e^z-1+(\alpha-e^z)_+\bigr],\nonumber
\end{align}
where \eqref{parameters-ineq1} uses $e^{2z}+1=2e^z\cosh z$.

Denote the terms independent of $c$ after completing the square by
$$
\begin{aligned}
h(z):={}&\frac{e^{-\rho z}}{2\beta}\cosh z+\beta e^{\rho z}\\
&-\frac{e^{-z/2}}2
\bigl[\beta(e^z+1)+\alpha e^z-1+(\alpha-e^z)_+\bigr].
\end{aligned}
$$

Thus, whenever $c\ge\beta e^z$,
\begin{equation}\label{eq:median-square}
\varphi(z,c)=h(z)-\frac{e^{-(1+\rho)z}}{4\beta}
\left(c-2\beta e^{(1/2+\rho)z}\right)^2
\le h(z).
\end{equation}

Since $e^L=\alpha=1+\beta$, the bracket in the definition of $h$ is
$$
\beta(e^z+1)+\alpha e^z-1+(\alpha-e^z)_+
=\begin{cases}
2\beta(e^z+1),&0\le z\le L,\\
(2\beta+1)e^z-(1-\beta),&z\ge L.
\end{cases}
$$

Consequently,
\begin{equation}\label{eq:median-h}
\begin{split}
h(z)={}&\frac{e^{(1-\rho)z}+e^{-(1+\rho)z}}{4\beta}
+\beta e^{\rho z}\\
&-\begin{cases}
\beta(e^{z/2}+e^{-z/2}),&0\le z\le L,\\[2pt]
\displaystyle\frac{2\beta+1}{2}e^{z/2}
-\frac{1-\beta}{2}e^{-z/2},&z\ge L.
\end{cases}
\end{split}
\end{equation}
The two expressions agree at $z=L$.

We next show that $h$ is convex separately on $[0,L]$ and
$[L,\infty)$. For $0<z<L$, differentiating the exponential terms in
\eqref{eq:median-h} twice gives
\begin{align*}
h''(z)
={}&-\frac\beta2\cosh(z/2)
+\frac{(1-\rho)^2e^{(1-\rho)z}+(1+\rho)^2e^{-(1+\rho)z}}{4\beta}
+\beta\rho^2e^{\rho z}.
\end{align*}

On this interval, we use $\cosh(z/2)\le\cosh(L/2)$,
$e^{(1-\rho)z}\ge1$, $e^{-(1+\rho)z}\ge e^{-(1+\rho)L}$, and
$e^{\rho z}\ge1$. Therefore,
$$\begin{aligned}
h''(z)
\ge{}&-\frac\beta2\cosh(L/2)
+\frac{(1-\rho)^2+(1+\rho)^2e^{-(1+\rho)L}}{4\beta}
+\beta\rho^2.
\end{aligned}$$

By~\eqref{eq:parameters},
$L/2<0.272$ and $(1+\rho)L<(1+0.098)\cdot0.543<0.597$.
We therefore have

$$
e^{-(1+\rho)L}>e^{-0.597}\ge
1-0.597+\frac{0.597^2}{2}-\frac{0.597^3}{6}
+\frac{0.597^4}{24}-\frac{0.597^5}{120}>0.550,
$$
\begin{align}
  \cosh\left(L/2\right) &< \cosh(0.272) \nonumber \\
  &= 1 + \frac{0.272^2}{2} + \sum_{k=2}^{\infty} \frac{0.272^{2k}}{(2k)!} \nonumber \\
  &\le 1 + \frac{0.272^2}{2} + \frac{0.272^4}{24} \sum_{m=0}^{\infty} \left(\frac{0.272^2}{30}\right)^m \label{parameters-ineq2} \\
  &= 1 + \frac{0.272^2}{2} + \frac{0.272^4}{24\left(1 - 0.272^2/{30}\right)} \nonumber \\
  &< 1.038. \nonumber
\end{align}
In~\eqref{parameters-ineq2}, we bound the tail of the Taylor series by
a geometric series, using $(2k+1)(2k+2) \ge 30$ for $k \ge 2$.

Substituting these bounds, $0.7195<\beta<0.7196$, and
$0.097<\rho<0.098$, we obtain

$$
\begin{aligned}
h''(z)
&>-\frac{0.7196}{2}\cdot1.038
+\frac{(1-0.098)^2+(1+0.097)^2\cdot0.550}{4\cdot0.7196}
+0.7195\cdot0.097^2\\
&>0.145>0.14.
\end{aligned}
$$
Hence $h$ is convex on $[0,L]$.

For $z>L$, differentiating the second expression
in~\eqref{eq:median-h} twice gives

$$\begin{aligned}
h''(z)
={}&e^{z/2}\left[
\frac{(1-\rho)^2}{4\beta}e^{(1/2-\rho)z}
-\frac{2\beta+1}{8}\right]\\
&+\frac{1-\beta}{8}e^{-z/2}+\frac{(1+\rho)^2}{4\beta}e^{-(1+\rho)z}
+\beta\rho^2e^{\rho z}.
\end{aligned}$$

The last three terms are positive. We focus on the bracketed terms. Since $1/2-\rho>0$ and $z\ge L$,
using $e^x\ge1+x$ and \eqref{eq:parameters} gives
$$\begin{aligned}
\frac{(1-\rho)^2}{4\beta}e^{(1/2-\rho)z}
-\frac{2\beta+1}{8}
&\ge\frac{(1-\rho)^2}{4\beta}e^{(1/2-\rho)L}
-\frac{2\beta+1}{8}\\
&>\frac{(1-0.098)^2}{4\cdot0.7196}
[1+(0.5-0.098)\cdot0.542]
-\frac{2\cdot0.7196+1}{8}\\
&>0.039.
\end{aligned}$$
Thus $h$ is also convex on $[L,\infty)$.

We now evaluate $h$ at the endpoints needed below. By the definition
of $\rho$,
$$
e^{\rho L}=\frac{\sqrt{\alpha^2+1}}{2\beta\sqrt\alpha},\qquad
\frac{e^{-\rho L}}{2\beta}\cosh L
=\beta e^{\rho L}
=\frac{\sqrt{\alpha^2+1}}{2\sqrt\alpha}.
$$

Substituting $z=0$ and $z=L$ into \eqref{eq:median-h}, and using
\eqref{eq:median-alpha} and the definition of $\gamma$, we obtain
\begin{align*}
h(0)&=\frac1{2\beta}-\beta=-\gamma,\\
h(L)&=\frac{\sqrt{\alpha^2+1}-\beta(\alpha+1)}{\sqrt\alpha}\\
&=\frac{\sqrt{\alpha^2+1}+1-\alpha^2}{\sqrt\alpha}
=\gamma.
\end{align*}

At $z=2L$, using $e^L=\alpha$ and the displayed formula for
$e^{\rho L}$ gives
\begin{align*}
h(2L)
&=\frac{\beta(\alpha^3+\alpha^{-1})}{\alpha^2+1}
+\frac{\alpha^2+1}{4\beta\alpha}
-\frac{(2\beta+1)\alpha}{2}
+\frac{1-\beta}{2\alpha}\\
&<\frac{0.7196(1.7196^3+1/1.7195)}{1.7195^2+1}
+\frac{1.7196^2+1}{4\cdot0.7195\cdot1.7195}\\
&\qquad-\frac{(2\cdot0.7195+1)\cdot1.7195}{2}
+\frac{1-0.7195}{2\cdot1.7195}\\
&<-0.185<-\gamma.
\end{align*}

Finally, $0.097<\rho<0.098$ implies
\begin{align*}
h(2)
&=\frac{e^{-2\rho}}{2\beta}\cosh2+\beta e^{2\rho}
-\frac{2\beta+1}{2}e+\frac{1-\beta}{2}e^{-1}\\
&<\frac{e^{-0.194}\cosh2}{2\beta}+\beta e^{0.196}
-\frac{2\beta+1}{2}e+\frac{1-\beta}{2}e^{-1}\\
&<\frac{0.824\cdot3.763}{2\cdot0.7195}
+0.7196\cdot1.217
-\frac{2\cdot0.7195+1}{2}\cdot2.718
+\frac{1-0.7195}{2}\cdot0.368\\
&<-0.23<-\gamma.
\end{align*}

For $0\le z\le L$, convexity and the endpoint values at $0,L$ give
$$
h(z)\le\left(1-\frac zL\right)h(0)+\frac zLh(L)
\le-\gamma+\frac{2\gamma z}{L}
=-\gamma\tau_L(z).
$$

For $L\le z\le2L$, the same argument gives
$$
h(z)\le\frac{2L-z}{L}h(L)+\frac{z-L}{L}h(2L)
\le3\gamma-\frac{2\gamma z}{L}
=-\gamma\tau_L(z).
$$

For $2L\le z\le2$, convexity, the bounds at $2L,2$, and
$\tau_L(z)\le1$ give $h(z)<-\gamma\le-\gamma\tau_L(z)$.
Therefore,
\begin{equation}\label{eq:median-radial-bound}
h(z)\le-\gamma\tau_L(z)\qquad(0\le z\le2).
\end{equation}
Combining this with \eqref{eq:median-square} proves the desired bound
when $c\ge\beta e^z$ and $z\le2$.

It remains to consider the other values of $z$ and $c$.
If $c\le\beta$, both positive parts involving $c$ vanish.
Since $\alpha e^z-1+(\alpha-e^z)_+
\ge\alpha e^z-1+\alpha-e^z=\beta(e^z+1)$
and $c\ge e^z-1$,~\eqref{eq:phi-expanded} gives
$$\begin{aligned}
\varphi(z,c)
&\le-\frac\beta2(e^{z/2}+e^{-z/2})
+\frac{e^{-(1+\rho)z}}{4\beta}
\bigl(e^{2z}+1-(e^z-1)^2\bigr)\\
&=-\beta\cosh(z/2)+\frac{e^{-\rho z}}{2\beta}\\
&\le-\beta+\frac1{2\beta}=-\gamma
\le-\gamma\tau_L(z).
\end{aligned}$$

Next suppose $\beta\le c\le\beta e^z$.
Together with $c\ge e^z-1$, this implies $(1-\beta)e^z\le e^z-c\le 1$.
But $(1-\beta)e^{3/2}>0.2804\cdot4.48>1$, so $z<3/2$.
In this case, only the positive part $c-\beta$ is active. Substituting it into
\eqref{eq:phi-expanded}, and then using the definition of $h$, gives
\begin{align}
\varphi(z,c)
={}&\frac{e^{-z/2}}2 \nonumber
\bigl[c-\beta-\alpha e^z+1-(\alpha-e^z)_+\bigr]
+\frac{e^{-(1+\rho)z}}{4\beta}(e^{2z}+1-c^2)\nonumber\\
={}&h(z)-\beta e^{\rho z}+\frac\beta2e^{z/2}
+\frac{e^{-z/2}}2c-\frac{e^{-(1+\rho)z}}{4\beta}c^2\nonumber\\
={}&h(z)-\frac{3\beta}{4}e^{\rho z}+\frac\beta2e^{z/2}
-\frac{e^{-(1+\rho)z}}{4\beta}
\left(c-\beta e^{(1/2+\rho)z}\right)^2\nonumber\\
\le{}&h(z)-\frac{3\beta}{4}e^{\rho z}+\frac\beta2e^{z/2}.\label{parameters-ineq3}
\end{align}
The upper bound on the RHS of~\eqref{parameters-ineq3}
is convex separately on $[0,L]$ and $[L,3/2]$. Indeed, $h$ is convex
on both intervals, and
$$\begin{aligned}
\left(-\frac{3\beta}{4}e^{\rho z}+\frac\beta2e^{z/2}\right)''
&=\frac\beta8e^{z/2}-\frac{3\beta\rho^2}{4}e^{\rho z}\\
&\ge\beta\left(\frac18-\frac{3\rho^2}{4}\right)e^{\rho z}>0,
\end{aligned}$$
where we use $0<\rho<1/10<1/2$ and $z\ge0$.

We evaluate this upper bound at $z=0$, $z=L$, and $z=3/2$.
At $z=0$, the identity $h(0)=-\gamma$ gives the value
$-\gamma-\beta/4<-0.20$. At $z=L$, using
$h(L)=\gamma<0.0247$, $e^{\rho L}\ge1+\rho L$, and
$e^{L/2}=\sqrt\alpha<1.312$, we obtain
$$\begin{aligned}
h(L)-\frac{3\beta}{4}e^{\rho L}+\frac\beta2e^{L/2}
&<0.0247-\frac{3\cdot0.7195}{4}(1+0.097\cdot0.542)
+\frac{0.7196}{2}\cdot1.312\\
&<-0.071.
\end{aligned}$$
For $z\ge L$, the second expression in \eqref{eq:median-h} gives
$$\begin{aligned}
h(z)-\frac{3\beta}{4}e^{\rho z}+\frac\beta2e^{z/2}
={}&\frac{e^{(1-\rho)z}+e^{-(1+\rho)z}}{4\beta}
+\frac\beta4e^{\rho z}\\
&-\frac{\beta+1}{2}e^{z/2}+\frac{1-\beta}{2}e^{-z/2}.
\end{aligned}$$
At $z=3/2$, we use
$$
e^{1.3545}<3.875,\quad e^{-1.6455}<0.193,\quad
 e^{0.147}<1.159,\quad e^{0.75}>2.117,\quad e^{-0.75}<0.473
$$
together with \eqref{eq:parameters}. This gives
$$\begin{aligned}
&h(3/2)-\frac{3\beta}{4}e^{3\rho/2}+\frac\beta2e^{3/4}\\
&\qquad<\frac{3.875+0.193}{4\cdot0.7195}
+\frac{0.7196}{4}\cdot1.159
-\frac{1+0.7195}{2}\cdot2.117
+\frac{1-0.7195}{2}\cdot0.473\\
&\qquad<-0.13.
\end{aligned}$$
On each of $[0,L]$ and $[L,3/2]$, convexity bounds the function
by the chord joining its endpoint values. Thus
$$
\varphi(z,c)\le h(z)-\frac{3\beta}{4}e^{\rho z}+\frac\beta2e^{z/2}
<-0.07<-\gamma\le-\gamma\tau_L(z).
$$

Finally, suppose $z\ge2$. We have
$(1-\beta)e^z\ge0.2804e^2>1$, so every feasible $c$ satisfies
$c\ge e^z-1\ge\beta e^z$.
Moreover,
\begin{align*}
2\beta e^{-(1/2-\rho)z}
&\le2\cdot0.7196e^{-0.804}\\
&<0.645<0.864<1-e^{-2}\le1-e^{-z}.
\end{align*}
Multiplying by $e^z$ gives
$2\beta e^{(1/2+\rho)z}<e^z-1\le c$.
Consequently, $\varphi(z,c)$ is
nonincreasing in $c$ on the feasible interval. Its maximum is attained
at $c=e^z-1$.
Substituting this value into \eqref{eq:phi-expanded}, and using
$e^z>\alpha$, gives
\begin{align}
\varphi(z,c)
&\le\varphi(z,e^z-1)\nonumber\\
&=\frac{e^{-z/2}}2\bigl[(1-2\beta)e^z-\alpha\bigr]
+\frac{e^{-\rho z}}{2\beta}\nonumber\\
&=-\left(\frac{2\beta-1}{2}e^{z/2}
+\frac\alpha2e^{-z/2}\right)+\frac{e^{-\rho z}}{2\beta}\nonumber\\
&\le-\sqrt{(2\beta-1)\alpha}+\frac1{2\beta}\label{parameters-ineq4}\\
&<-\sqrt{(2\cdot0.7195-1)\cdot1.7195}+\frac1{2\cdot0.7195}\nonumber\\
&<-0.17<-\gamma\le-\gamma\tau_L(z).\nonumber
\end{align}
\eqref{parameters-ineq4} uses the arithmetic--geometric mean inequality on
the bracketed terms, together with
$e^{-\rho z}\le1$.
All feasible values of $z$ and $c$ have now been covered. This proves
\eqref{eq:certificate}.
\end{proof}

\subsection{Analysis for k-Clustering with \texorpdfstring{$p\ge2$}{p>=2} in Euclidean Space}\label{sec:euclid-k-clustering}

\begin{lemma}\label{lem:euclidean-power}
For $p\ge2$ and $\alpha=(3^p+2)/3$,
inequality~\eqref{eq:pairwise} holds. In particular, $\alpha=11/3$
suffices when $p=2$.
\end{lemma}

\begin{proof}
For each facility $i$, define
$$
{\tilde v}_i:=\begin{cases}
d_i^{p/2-1}v_i,&d_i>0,\\
0,&d_i=0.
\end{cases}
$$

We first show that, for every $i,i'\in T$,
\begin{equation}\label{eq:powered-chord}
d_{ii'}^p\le2^{p-2}\norm{{\tilde v}_i-{\tilde v}_{i'}}^2.
\end{equation}

The inequality is immediate if either radius is zero. Otherwise, let
$r=d_i$, $s=d_{i'}$, and
$\theta=\ip{v_i}{v_{i'}}/(rs)\in[-1,1]$. Since
$\norm{{\tilde v}_i}^2=r^p$ and $\norm{{\tilde v}_{i'}}^2=s^p$,
\eqref{eq:powered-chord} becomes

\begin{equation}\nonumber
(r^2+s^2-2rs\theta)^{p/2}
\le2^{p-2}(r^p+s^p-2r^{p/2}s^{p/2}\theta).
\end{equation}

Since $p/2\ge 1$, the LHS is convex in $\theta$, whereas
the RHS is affine. It therefore suffices to verify the
inequality at the two endpoints.

At $\theta=-1$, convexity gives
$$
(r+s)^{p/2}\le2^{p/2-1}(r^{p/2}+s^{p/2}).
$$

Squaring both sides proves the desired bound.

At $\theta=1$, assume $r\ge s$. Since $p/2\ge1$, we have

$$\begin{aligned}
r^{p/2}-s^{p/2}
&=\frac{p}{2}\int_0^{r-s}(s+t)^{p/2-1}\,dt\\
&\ge \frac{p}{2}\int_0^{r-s}t^{p/2-1}\,dt\\
&=(r-s)^{p/2}.
\end{aligned}$$

Consequently,
$(r-s)^p \le (r^{p/2}-s^{p/2})^2\le 2^{p-2}(r^{p/2}-s^{p/2})^2.$
This proves~\eqref{eq:powered-chord}.

For $i,i'\in T$, let $r=d_i$, $s=d_{i'}$, $c=d_{ii'}$, and
$J=\ip{{\tilde v}_i}{{\tilde v}_{i'}}$. Applying convexity with weights $1/3,2/3$
and then~\eqref{eq:powered-chord}, we obtain

\begin{align}\label{eq:path-bound}
(r+c)^p
&=\left(\frac13\cdot3r+\frac23\cdot\frac{3c}{2}\right)^p\nonumber\\
&\le3^{p-1}r^p+(3/2)^{p-1}c^p\nonumber\\
&\le 3^{p-1}r^p+\frac{3^{p-1}}{2}(r^p+s^p-2J)
=\frac{3^p}{2}r^p+\frac{3^{p-1}}{2}s^p-3^{p-1}J.
\end{align}

By symmetry, we also have
$(s+c)^p\le \frac{3^p}{2}s^p + \frac{3^{p-1}}{2}r^p-3^{p-1}J$.

We next prove
\begin{equation}\label{eq:max-pointwise}
\begin{split}
\max\{\alpha r^p,(r+c)^p\}
+\max\{\alpha s^p,(s+c)^p\}\\
\le(2\alpha-1)(r^p+s^p)-2(\alpha-1)J.
\end{split}
\end{equation}

To prove this inequality, we consider three cases according to the
terms attaining the two maxima.

If both maxima are attained by their $\alpha$ terms, the RHS of~\eqref{eq:max-pointwise} minus the LHS is

$$
(\alpha-1)(r^p+s^p-2J)
=(\alpha-1)\norm{{\tilde v}_i-{\tilde v}_{i'}}^2\ge0.
$$

If exactly one maximum is attained by its $\alpha$ term, assume that it is the second maximum. By~\eqref{eq:path-bound},
the RHS minus the LHS is at least

$$\begin{aligned}
&(2\alpha-1)(r^p+s^p)-2(\alpha-1)J-\frac{3^p}{2}r^p-\frac{3^{p-1}}{2}s^p+3^{p-1}J-\alpha s^p\\
&\qquad=\frac{\alpha}{2}r^p+\left(\frac{\alpha}{2}-\frac{2}{3}\right)s^p-2\left(\frac{\alpha}{2}-\frac{2}{3}\right)J\\
&\qquad=\left(\frac{\alpha}{2}-\frac{2}{3}\right)(r^p+s^p-2J) + \frac{2}{3}r^p\\
&\qquad=\left(\frac{\alpha}{2}-\frac{2}{3}\right)\norm{{\tilde v}_i-{\tilde v}_{i'}}^2+ \frac{2}{3}r^p\ge 0.
\end{aligned}$$

If neither maximum is attained by its $\alpha$ term, applying
\eqref{eq:path-bound} and its symmetric form gives a difference of at
least

$$\begin{aligned}
&(2\alpha-1)(r^p+s^p)-2(\alpha-1)J-2\cdot 3^{p-1}(r^p+s^p)+2\cdot 3^{p-1}J\\
&\qquad=\frac{1}{3}(r^p+s^p+2J)=\frac{1}{3}\norm{{\tilde v}_i+{\tilde v}_{i'}}^2\ge 0.
\end{aligned}$$
These cases prove~\eqref{eq:max-pointwise}.

Multiplying~\eqref{eq:max-pointwise} by $y_i y_{i'}$, summing over
ordered pairs, and dividing by two, we obtain

$$\begin{aligned}
\sum_{i,i'\in T}y_i y_{i'}\max\{\alpha d_i^p,(d_i+d_{ii'})^p\}
&\le(2\alpha-1)y(T)\sum_{i\in T}y_i d_i^p
-(\alpha-1)\sum_{i,i'\in T}y_iy_{i'}\ip{{\tilde v}_i}{{\tilde v}_{i'}}\\
&=(2\alpha-1)y(T)\sum_{i\in T}y_i d_i^p
-(\alpha-1)\left\|\sum_{i\in T} y_i{\tilde v}_i\right\|^2\\
&\le(2\alpha-1)y(T)\sum_{i\in T}y_i d_i^p.
\end{aligned}$$

Finally, $\pos{a-b}=\max\{a,b\}-b$ gives

$$\begin{aligned}
&\sum_{i,i'\in T} y_{i}y_{i'}((d_i+d_{ii'})^p-\alpha d_i^p)_+\\
&=\sum_{i,i'\in T}y_iy_{i'}\max\{\alpha d_i^p,(d_i+d_{ii'})^p\}-\alpha y(T)\sum_{i\in T}y_i d_i^p\\
&\le (\alpha-1)y(T)\sum_{i\in T}y_i d_i^p\\
&=\sum_{i,i'\in T}y_i y_{i'}(\alpha d_i^p - d_{i'}^p)\\
&\le \sum_{i,i'\in T}y_i y_{i'}\pos{\alpha d_i^p - d_{i'}^p}.
\end{aligned}
$$
This proves~\eqref{eq:pairwise}.

\end{proof}

\subsection{LMP Approximation and Euclidean UFL}

\begin{theorem}\label{thm:euclidean-lmp-consequence}
Let $\alpha$ be defined by~\eqref{eq:median-alpha}.
There is a randomized $\alpha$-LMP approximation for Euclidean
$k$-median. For every fixed $p\ge2$, there is a randomized
$\frac{3^p+2}{3}$-LMP approximation for Euclidean $k$-clustering with
connection costs $d(i,j)^p$. In particular, the factor for Euclidean
$k$-means is $11/3$.
\end{theorem}

\begin{proof}
The claims follow by combining Lemma~\ref{lem:nonstop-lmp} with
Lemma~\ref{lem:euclidean-median} for $p=1$ and with
Lemma~\ref{lem:euclidean-power} for $p\ge2$.
\end{proof}

We next consider Euclidean uncapacitated facility location, where
the candidate facilities $F$ are given explicitly and have nonnegative
opening costs $(f_i)_{i\in F}$. Here $p=1$, and there is no cardinality
constraint. For a feasible fractional solution, write
$$
F_{\mathrm{LP}}(y):=\sum_{i\in F}f_i y_i,\qquad
C_{\mathrm{LP}}(x):=\sum_{j\in C}\sum_{i\in F}d(i,j)x_{i,j}.
$$
For a nonempty integral set $O$, write
$$
F(O):=\sum_{i\in O}f_i,\qquad C(O):=\sum_{j\in C}d(j,O).
$$
When we apply the single-pivot rounding algorithm reviewed in
Section~\ref{sec:client-dep}, Lemma~\ref{lem:nonstop-lmp} gives
\begin{equation}\label{eq:ufl-fixed-lp}
\E[F(O)]\le F_{\mathrm{LP}}(y),\qquad
\E[C(O)]\le\alpha C_{\mathrm{LP}}(x).
\end{equation}

To obtain a $(1,\alpha)$-approximation, solve the weighted LP
\begin{equation}\label{eq:ufl-weighted-lp}
\min\bigl\{F_{\mathrm{LP}}(y)+\alpha C_{\mathrm{LP}}(x):
 (x,y)\text{ is feasible for the UFL relaxation}\bigr\},
\end{equation}
and round its solution using the same algorithm. For any fixed feasible
integral solution with opening and connection costs $F^*,C^*$,
optimality of the weighted LP and~\eqref{eq:ufl-fixed-lp} give
\begin{equation}\label{eq:ufl-base-bifactor}
\E[F(O)+C(O)]
\le F_{\mathrm{LP}}(y)+\alpha C_{\mathrm{LP}}(x)
\le F^*+\alpha C^*.
\end{equation}
Thus the algorithm is a $(1,\alpha)$-approximation.

We then apply the following scaling lemma
from~\cite[Lemma~2.1]{byrka2010optimal}.

\begin{lemma}\label{lem:ufl-scaling}
Let $A$ be a $(\lambda_f,\lambda_c)$-approximation algorithm for metric
UFL, and let $\delta\ge1$. Let $S_\delta(A)$ scale all opening costs
by $\delta$, run $A$, restore the original opening costs, and then
apply greedy augmentation. Then $S_\delta(A)$ is a
$\left(\lambda_f+\ln\delta,\ 1+\frac{\lambda_c-1}{\delta}\right)$-approximation
algorithm.
\end{lemma}

\begin{proof}[Proof of Theorem~\ref{thm:euclidean-ufl}]
Let $\alpha$ be defined by~\eqref{eq:median-alpha}.
Applying Lemma~\ref{lem:ufl-scaling} to the $(1,\alpha)$ guarantee
in~\eqref{eq:ufl-base-bifactor} gives a
$\left(1+\ln\delta,\ 1+\frac{\alpha-1}{\delta}\right)$ bifactor approximation
in expectation for every $\delta\ge1$. Choose $\delta_\star>1$ satisfying
$\delta_\star\ln\delta_\star=\alpha-1$ to balance the two coefficients.
Numerically, $\delta_\star\approx1.577813480091$.

At this value, the two coefficients are equal, and the approximation
factor is
$$
1+\ln\delta_\star=1+\frac{\alpha-1}{\delta_\star}\approx1.456040015242<1.4561.
$$
This proves the theorem.
\end{proof}
\section{Converting the LMP Algorithm into One with \texorpdfstring{$k+O(1)$}{k+O(1)} Open Facilities}
\label{sec:single-pivot-conversion}

In this section, we convert the LMP algorithm from Section~\ref{sec:client-dependent-analysis} into an algorithm that always opens at most $k+O(1)$ facilities. A similar conversion was given in \cite{BGHL2026} for their $2$-LMP algorithm. Our contribution is twofold. First, their conversion does not apply as a black box to our refined, dispersion-dependent bound, so we establish conditions under which this bound survives the conversion. Second, our conversion is simpler: instead of selecting a batch of pivots, we select a single pivot in each iteration. Relative to the algorithm in Section~\ref{sec:client-dependent-analysis}, we favor operations with negative imbalance and split large deletion sets into smaller ones. This bounds the change in total opening in each iteration and allows us to apply concentration bounds.

\subsection{Algorithm}\label{sec:conversion-algorithm}
In our conversion, we need to first apply the preprocessing procedure in Section~\ref{sec:preprocessing}. Let $\varepsilon\in(0,1)$ be sufficiently small and satisfy the arithmetic assumptions of Theorem~\ref{thm:preprocessing}; it will also be the multiplicative loss parameter for the conversion. We apply the preprocessing to obtain a  $\delta = \textrm{poly}(\varepsilon)$-integral opening vector $\mathbf y \in [0, 1]^F$, where we assume $\delta^{-1}\in\mathbb N$ and $0<\delta<1$.  We assume that the input has at least one unit of opening and write $Y_0:=y(F)\ge1$. If preprocessing returns a vector with $Y_0>k+3$, we repeat preprocessing until $Y_0\le k+3$. Its failure probability is $\exp(-\Omega(\varepsilon^{-10}))$ for $p=1$, so this conditioning preserves its $1+O(\varepsilon)$ expected-cost guarantee. We also define
\begin{align*}
	\eta&:=\varepsilon\delta/4, \quad \zeta:=\varepsilon/4, \quad D:=  \lceil 4(1+\eta)/\eta \rceil + 1,  \quad \lambda:=\frac{\log(1+\eta/2)}{D+1},\\
	 s&:=\left\lceil\lambda^{-1}\log(1/\zeta)\right\rceil
	=O\!\left(1+\eta^{-2}\log(1/\zeta)\right).
\end{align*}
%\quad \text{and}\quad

We split each facility $i$ into $y_i/\delta$ colocated copies of opening $\delta$, so that every facility $i$ now has $y_i = \delta$. For simplicity, we still use  $F$ and $\mathbf y$ for the new set of facilities and new fractional opening vector. 

During the algorithm, let $\mathbf z^{(t)}$ denote the current opening vector, with $\mathbf z^{(0)}=\mathbf y$. Initially, $O_0=\varnothing$ and $R_0=F$.  After each iteration $t$, each facility $i$ has one of the following states:

\begin{itemize}
\item \emph{Closed}: $z_i^{(t)}=0$.
\item \emph{Open}: $z_i^{(t)}=1$. Denote by $O_t$ the set of open facilities at this time.
\item \emph{Fractional}: $z_i^{(t)}=\delta$. Denote by $R_t$ the set of fractional facilities at this time.
\end{itemize}

In each iteration, we choose a pivot that is either fractional or already open, open it if necessary, and close a set of fractional facilities. An iteration may leave the state unchanged. Once a facility is opened or closed, it remains in that state for the rest of the algorithm. The algorithm terminates when all facilities are either open or closed. 

Now we describe iteration $t$ of the algorithm, for $t \geq 1$:
\begin{itemize}
\item Put $1$ \emph{token} on each fractional facility $i \in R_{t-1}$ and $\delta^{-1}$ tokens on each open facility $i \in O_{t-1}$. For every token $x$, let $\operatorname{own}(x)$ denote its facility. So, one token represents a $\delta$-fractional opening. 
\item For each fractional facility $i \in R_{t-1}$, denote by $N(i)$ the set of the nearest $\delta^{-1}$ tokens to $i$ (breaking ties in favor of its own token, and otherwise in a fixed order). Such a set exists initially because $Y_0\ge1$, and after the first iteration because at least one facility is open and remains open.
\item For each token $x$, denote by $N^{-1}(x)$ the set of fractional facilities $i$ such that $x \in N(i)$, except for the facility that owns $x$. Then create several \emph{operations} associated with $x$ as follows. Each operation is represented by $(x,Q)$ where $Q\subseteq R_{t-1}$ is a set of fractional facilities.
\begin{itemize}
\item If $|N^{-1}(x)|\le \delta^{-1}D$, then set $r_x:=1$ and create one operation $(x,N^{-1}(x))$.
\item Otherwise, split $N^{-1}(x)$ into $r_x:=\lceil |N^{-1}(x)|/(\delta^{-1}D)\rceil$ disjoint sets $Q_1,\ldots,Q_{r_x}$ of size between $\delta^{-1}D/2$ and $\delta^{-1}D$, and create $r_x$ operations $(x,Q_1),\ldots,(x,Q_{r_x})$.
\end{itemize}
Let $\mathcal A$ be the set of all operations $(x,Q)$ created. We retain the nearest-unit
weights from the previous LMP analysis. In the $\delta$-integral
setting, their token representation is
\begin{align}
 w_{q,i}&:=\delta\,|\{x\in N(i):\operatorname{own}(x)=q\}|, &&\forall q \in R_{t-1} \cup O_{t-1}, i \in R_{t-1}, \quad \text{and}\nonumber\\
 w(U,i)&:=\sum_{q\in U}w_{q,i}, &&\forall U \subseteq R_{t-1} \cup O_{t-1}, i \in R_{t-1}. \label{eq:conversion-token-weights}
\end{align}
Since the sets $\{Q: (x, Q) \in \mathcal A\}$ partition $N^{-1}(x)$, for $q\ne i$, 
the number of operations $(x,Q)$ with $\operatorname{own}(x)=q$
and $i\in Q$ is $w_{q,i}/\delta$.
\item The imbalance $\beta(x,Q)$ of an operation $(x,Q)$ is defined as follows. If $x$ is on a fractional facility $i$, then $\beta(x,Q)=(1-\delta)-\delta|Q|$. If $x$ is on an open facility, then $\beta(x,Q)=-\delta|Q|$. The weight of this operation is defined as $\omega(x,Q)=(1+\eta)^{\mathbf 1_{\beta(x,Q)<0}}$.
\item Sample an operation $(x,Q)$ with probability proportional to its weight $\omega(x,Q)$. \emph{Perform} the operation: open the facility that owns $x$ and close all facilities in $Q$.
\end{itemize}

This finishes the description of the iterative algorithm.  Note that the imbalance of an operation is the change in the total opening value $|O_t|+\delta|R_t|$ of facilities after performing the operation. By boosting the weight of operations with negative imbalance, we can control the exponential moment of $|O_t|+\delta|R_t|$ after each iteration. This gives a high probability guarantee about the final number of open facilities.

\subsection{Analysis of the Number of Open Facilities}
In this section, we bound the number of open facilities.  Fix an iteration $t$ and the history till the end of iteration $t-1$. For a token $x$, define its imbalance $\beta(x)$ similarly to the imbalance of an operation: If $x$ is on a fractional facility $i$, then $\beta(x)=(1-\delta)-\delta|N^{-1}(x)|$. If $x$ is on an open facility, then $\beta(x)=-\delta|N^{-1}(x)|$.

Let
$$\operatorname{imb}^+=\sum_{(x,Q)\in\mathcal A:\,\beta(x,Q)>0}\beta(x,Q),\qquad \operatorname{imb}^-=\sum_{(x,Q)\in\mathcal A:\,\beta(x,Q)<0}-\beta(x,Q).$$

\begin{lemma}\label{lem:imbalance}
We have $\sum_x \beta(x)=0$, and $$0\le \operatorname{imb}^+-\operatorname{imb}^-\le\frac{\operatorname{imb}^+}{D-1}.$$
\end{lemma}

\begin{proof}
Summing the imbalance of all tokens gives
$$\begin{aligned}
\sum_x \beta(x)&=(1-\delta)|R_{t-1}|-\delta\sum_x|N^{-1}(x)| =(1-\delta)|R_{t-1}|-\delta\sum_{i\in R_{t-1}}(|N(i)|-1)\\
&=(1-\delta)|R_{t-1}|-(1-\delta)|R_{t-1}| =0.
\end{aligned}$$

We split a token $x$ only when $|N^{-1}(x)|>\delta^{-1}D$, which implies $\beta(x)<0$. The change of imbalance after this splitting is (recall that $r_x$ is the number of operations created by splitting $N^{-1}(x)$)
$$\beta(x)-\sum_{Q:(x,Q)\in\mathcal A}\beta(x,Q)=\begin{cases}0 & (\operatorname{own}(x)\in O_{t-1})\\ -(1-\delta)\cdot(r_x-1) & (\operatorname{own}(x)\in R_{t-1})\end{cases}.$$

Summing over all tokens $x$ gives $\operatorname{imb}^+-\operatorname{imb}^-=(1-\delta)\cdot \sum_{x:\operatorname{own}(x)\in R_{t-1}}(r_x-1)$. This value is non-negative and can be upper bounded as follows.
\begin{align*}
(1-\delta)\cdot \sum_{x:\operatorname{own}(x)\in R_{t-1}}(r_x-1)&\le \sum_{\substack{x:\operatorname{own}(x)\in R_{t-1}\\r_x>1}}\frac{|N^{-1}(x)|}{\delta^{-1}D}\\
&\le \sum_{\substack{x:\operatorname{own}(x)\in R_{t-1}\\r_x>1}}\frac{\delta|N^{-1}(x)|-(1-\delta)}{D-1} \le \frac{\operatorname{imb}^+}{D-1}.
\end{align*}
For the last inequality, every fragment of a split token has negative imbalance, since $\delta|Q|\ge D/2>1-\delta$. Thus the positive imbalance is unchanged by splitting, and the total negative imbalance of the original tokens equals $\operatorname{imb}^+$ because $\sum_x\beta(x)=0$.
\end{proof}

Denote by $Y_t$ the total opening value of facilities after iteration $t$, i.e., $Y_t=|O_t|+\delta|R_t|$. When $t$ is larger than the number of iterations, we let $Y_t$ be the number of open facilities returned by the algorithm. 

\begin{lemma}\label{lem:conversion-exponential}
$\left(\exp(\lambda(Y_t-Y_0))\right)_{t\ge 0}$
is a supermartingale. Consequently, for every $h\ge0$,
$$\Pr\left[\sup_{t\ge0} Y_t\ge Y_0+h\right]\le e^{-\lambda h}.$$
\end{lemma}

\begin{proof}
Fix an iteration $t$ and condition on the history through iteration $t-1$. All expectations below are conditional on this history. Let $\Delta Y=Y_t-Y_{t-1}$; it suffices to show that $\mathbb E[\exp(\lambda\Delta Y)-1]\le 0$. After termination, $\Delta Y=0$, so we consider only iterations before termination.

By definition we have $-D\le\beta(x,Q)\le 1$. 
Write $W=\sum_{(x,Q)\in\mathcal A}\omega(x,Q)$ for the total operation weight. For $-D\le\beta(x,Q)\le 0$, we use $e^{\lambda\beta(x,Q)}-1\le \lambda\beta(x,Q)e^{-\lambda D}$. For $0\le\beta(x,Q)\le 1$, we use $e^{\lambda\beta(x,Q)}-1\le \lambda\beta(x,Q)e^{\lambda}$. Therefore, we have
\begin{align*}
\mathbb E[\exp(\lambda\Delta Y)-1]&=\frac{1}{W}\sum_{(x,Q)\in \mathcal A}\omega(x,Q)\left(e^{\lambda\beta(x,Q)}-1\right)\\
&=\frac{1}{W}\left(\sum_{\beta(x,Q)>0}\left(e^{\lambda\beta(x,Q)}-1\right)+\sum_{\beta(x,Q)<0}(1+\eta)\left(e^{\lambda\beta(x,Q)}-1\right)\right)\\
&\le \frac{\lambda}{W}\left(e^{\lambda}\operatorname{imb}^+-(1+\eta)e^{-\lambda D}\operatorname{imb}^-\right)\\
&\le \frac{\lambda e^{-\lambda D}}{W}\left(e^{\lambda(D+1)}\operatorname{imb}^+-(1+\eta)\operatorname{imb}^-\right).
\end{align*}

From Lemma~\ref{lem:imbalance} and $D-1\ge 4(1+\eta)/\eta$, we have
$$(1+\eta)(\operatorname{imb}^+-\operatorname{imb}^-)\le \frac{\eta\operatorname{imb}^+}{4}.$$

Thus $(1+\eta)\operatorname{imb}^-\ge (1+3\eta/4)\operatorname{imb}^+$. Substituting this into the previous inequality gives
\begin{align}
\mathbb E[\exp(\lambda\Delta Y)-1]&\le \frac{\lambda e^{-\lambda D}\operatorname{imb}^+}{{\sum_{(x,Q)\in \mathcal A}\omega(x,Q)}}\left(e^{\lambda(D+1)}-(1+3\eta/4)\right) \le 0.
\end{align}
The stated tail bound follows from Ville's inequality applied to the nonnegative supermartingale $Z_t:=\exp(\lambda(Y_t-Y_0))$, since $Z_0=1$.
\end{proof}

\begin{lemma}\label{lem:conversion-termination}
The algorithm terminates almost surely, with a polynomial expected
number of iterations in $|F|$ and $\delta^{-1}$.
\end{lemma}
\begin{proof}
For any iteration, there are at most $|F|\delta^{-1}$ tokens, every token is split into at most $1+|F|\delta^{-1}/D$ operations, and $\omega(x,Q)\le 2$. Thus
$$\sum_{(x,Q)\in \mathcal A}\omega(x,Q)\le 2|F|\delta^{-1}(1+|F|\delta^{-1}/D)$$

is polynomial. Any fractional facility is opened with probability at least $1/\bigl(2|F|\delta^{-1}(1+|F|\delta^{-1}/D)\bigr)$ in each iteration. By the union bound, the probability that the algorithm does not terminate after $t$ iterations is at most
$$|F|\bigl(1-1/\bigl(2|F|\delta^{-1}(1+|F|\delta^{-1}/D)\bigr)\bigr)^t,$$

which is exponentially small in $t$. Thus the algorithm terminates almost surely and the expected number of iterations is polynomial.
\end{proof}

To bound the number of opened facilities, we run independent trials of the rounding algorithm until a trial satisfies $|O|\le Y_0+s$.
\begin{lemma}\label{lem:conversion-budget}
Each trial is accepted with probability at least $1-\zeta$.
%, and the opened facility set $O$ of the accepted trial satisfies
\end{lemma}
\begin{proof}
By Lemma~\ref{lem:conversion-termination}, the termination time
$\sigma$ is finite almost surely, and $Y_\sigma=|O_\sigma|$.
The rejection event is contained in
$\{\sup_tY_t\ge Y_0+s\}$. By
Lemma~\ref{lem:conversion-exponential}, its probability is at most
$e^{-\lambda s}\le\zeta$.
\end{proof}
Therefore, the expected number of trials is at most $1/(1-\zeta)$. 

\subsection{Analysis of Connection Cost of \texorpdfstring{$j$}{j}}\label{sec:conversion-certificate}
In this section, we fix a client $j \in C$, and we analyze the expected connection cost of $j$. For a fixed distance exponent $p\ge1$, let $F_j$ be the nearest $1/\delta$ facilities of $j$ in $F$, after the splitting operation.  Thus $\delta|F_j|=1$. Define
$$
d_i=d(i,j),\kappa_i=d_i^p, \forall i \in F, \qquad \text{and} \qquad
V(U)=\sum_{i\in U}\delta\kappa_i, \forall U\subseteq F_j.
$$
Then $V(F_j)$ is the cost of $j$ in the LP solution obtained from preprocessing.

In the analysis, we track the alive fractional set $T\subseteq F_j$ and a backup cost $b$, initially
$(T,b)=(F_j,\infty)$. Given a current state $(T,b)$, a token $x$, and a partial deletion
set $Q\subseteq N^{-1}(x)$, we define the state transition under operation $(x,Q)$ by
\begin{equation}\label{eq:conversion-analytic-transition}
\begin{aligned}
 H_x(Q)&=\mathbf1_{\{\operatorname{own}(x)\in T\ \text{or}\ Q\cap T\ne\varnothing\}},\\
 T^{x,Q}&=T\setminus\bigl(\{\operatorname{own}(x)\}\cup Q\bigr),\\
 b^{x,Q}&=\begin{cases}
 \min\{b,\kappa_{\operatorname{own}(x)}\},&H_x(Q)=1,\\
 b,&H_x(Q)=0.
 \end{cases}
\end{aligned}
\end{equation}
We say that an operation $(x,Q)$ is \emph{relevant} to client $j$ if $H_x(Q)=1$.
A relevant operation removes at least one copy from $T$, whereas
an irrelevant operation leaves both $T$ and $b$ unchanged. There
are consequently at most $1/\delta$ relevant operations, and the
analytical state stays fixed once $T=\varnothing$.

\begin{definition}[The local-charge certificate]\label{def:conversion-certificate}
Let $P_j(T,b)$ be a nonnegative potential, uniformly bounded over
the reachable states of the fixed finite instance. We say that $P_j$ has a
local-charge certificate with initial bound $\mathcal B_j$ if
\begin{itemize}
	\item $ P_j(F_j,\infty)\le\mathcal B_j, P_j(\varnothing,b)\ge b, \forall \text{ finite } b \geq 0$, and
	\item at every reachable algorithm state with analytical state $(T,b)$, there exist opening charges
	$a_{j,x}\in\mathbb R$ (for each token $x$) and deletion charges $g_{j,x,i}\in\mathbb R$ (for each token $x$ and each
	$i\in N^{-1}(x)\cap T$) satisfying
	\begin{align}
		P_j(T^{x,Q},b^{x,Q})-P_j(T,b)
		&\le H_x(Q)a_{j,x}+\sum_{i\in Q\cap T}g_{j,x,i},
		\forall Q\subseteq N^{-1}(x),\label{eq:conversion-majorant}\\
		\sum_x\left(a_{j,x}\sum_{Q:(x,Q)\in\mathcal A}H_x(Q)
		+\sum_{i\in N^{-1}(x)\cap T}g_{j,x,i}\right)
		&\le0.\label{eq:conversion-balance}
	\end{align}
\end{itemize}
\end{definition}
We remark that in the second condition, the charges may depend on the entire current algorithm state, including $(T,b)$ and the constructed operations,
but the same $a_{j,x}$ and $g_{j,x,i}$ must satisfy
\eqref{eq:conversion-majorant} for every $Q\subseteq N^{-1}(x)$.

Intuitively, the opening charge $a_{j,x}$ accounts for the effect of
opening $\operatorname{own}(x)$ and updating the backup cost, and is incurred
once for each relevant operation associated with $x$. The deletion charge
$g_{j,x,i}$ accounts for the removal of a tracked copy
$i\in Q\cap T$. Crucially, because the sets $\{Q:(x,Q)\in\mathcal A\}$ partition $N^{-1}(x)$, each deletion charge $g_{j,x,i}$
is counted at most once in the balance condition~\eqref{eq:conversion-balance}, even when token $x$ is split into multiple operations.

\begin{theorem}\label{thm:conversion-local-charges}
If every client has a local-charge certificate with $\mathcal B_j$ as in
Definition~\ref{def:conversion-certificate}, the opened facility
set $O$ of
the accepted trial satisfies
\begin{equation}\label{eq:conversion-transfer}
 |O|\le Y_0+s,\qquad
 \E[d(j,O)^p]\le\frac{(1+\eta)^{1/\delta}}{1-\zeta}\mathcal B_j.
\end{equation}
In particular, our choices $\eta=\varepsilon \delta/4$ and
$\zeta=\varepsilon/4$ yield $\E[d(j,O)^p]\le(1+\varepsilon)\mathcal B_j$ with an
additive facility allowance of
$s=O\!\left(\delta^{-2}\varepsilon^{-2}\log(4/\varepsilon)\right)$.
\end{theorem}

\begin{proof}
We first analyze a single unconditioned trial of the algorithm.
Fix the current state and let $\nu_0(x,Q):=1/|\mathcal A|$ denote the uniform distribution over the constructed operations
$(x,Q)\in\mathcal A$.
Let $(T',b')$ be the random next state. Summing
\eqref{eq:conversion-majorant} over these operations gives
\begin{align}
 &|\mathcal A|\bigl(\E_0[P_j(T',b')]-P_j(T,b)\bigr)
 \nonumber\\
 &\qquad=\sum_x\sum_{Q:(x,Q)\in\mathcal A}
 \bigl[P_j(T^{x,Q},b^{x,Q})-P_j(T,b)\bigr]
 \nonumber\\
 &\qquad\le\sum_x\left(
 a_{j,x}\sum_{Q:(x,Q)\in\mathcal A}H_x(Q)
 +\sum_{i\in N^{-1}(x)\cap T}g_{j,x,i}\right)
 \label{eq:conversion-fragment-majorant}\\
 &\qquad\le0.
 \label{eq:conversion-fragment-balance}
\end{align}
For \eqref{eq:conversion-fragment-majorant}, use
\eqref{eq:conversion-majorant} and the partition of $N^{-1}(x)$,
which counts every deletion charge exactly once.
\eqref{eq:conversion-fragment-balance} is precisely
\eqref{eq:conversion-balance} for the constructed operations.
Thus the unbiased potential drift is nonpositive.

Let $\mathcal I_j:=\{(x,Q)\in\mathcal A: H_x(Q)=1\}$ be the set of operations relevant to client $j$. Irrelevant operations leave the potential
unchanged. Therefore, if $\mathcal I_j\ne\varnothing$,
\begin{equation}\label{eq:conversion-unbiased-relevant}
 \E_0[P_j(T',b')\mid\mathcal I_j]\le P_j(T,b).
\end{equation}
All probabilities and expectations in this part of the proof are
conditional on the current state; the subscript specifies the
sampling distribution.

In fact, the algorithm uses a biased sampling.
It samples operation $(x,Q)$ with probability 
$$
 \nu_\eta(x,Q)=\frac{\omega(x,Q)}
 {\sum_{(\tilde x,\tilde Q)\in\mathcal A}\omega(\tilde x,\tilde Q)}.
$$
Since $1\le \omega(x,Q)\le1+\eta$, for every
$(x,Q)\in\mathcal I_j$ we have
\begin{equation}\label{eq:conversion-relevant-ratio}
 \frac{\nu_\eta((x,Q)\mid\mathcal I_j)}
      {\nu_0((x,Q)\mid\mathcal I_j)}
 =\frac{\omega(x,Q)|\mathcal I_j|}
 {\sum_{(\tilde x,\tilde Q)\in\mathcal I_j}\omega(\tilde x,\tilde Q)}
 \le1+\eta.
\end{equation}
Define $\widehat P_j(T,b):=(1+\eta)^{|T|}P_j(T,b)$. On a relevant
operation, $|T'|\le|T|-1$, and thus
\begin{align}
 \E_\eta[\widehat P_j(T',b')\mid\mathcal I_j]
 &\le(1+\eta)^{|T|-1}\E_\eta[P_j(T',b')\mid\mathcal I_j]
 \label{eq:conversion-bias-support}\\
 &\le(1+\eta)^{|T|}\E_0[P_j(T',b')\mid\mathcal I_j]
 \label{eq:conversion-bias-ratio}\\
 &\le(1+\eta)^{|T|}P_j(T,b)=\widehat P_j(T,b).
 \label{eq:conversion-bias-drift}
\end{align}
For \eqref{eq:conversion-bias-support}, use $|T'|\le|T|-1$
and nonnegativity of $P_j(T',b')$.
\eqref{eq:conversion-bias-ratio} follows from
\eqref{eq:conversion-relevant-ratio} and the same nonnegativity;
\eqref{eq:conversion-bias-drift} uses
\eqref{eq:conversion-unbiased-relevant}. Irrelevant operations
preserve $\widehat P_j(T,b)$. Averaging over relevant and
irrelevant operations proves that $\widehat P_j(T,b)$ is a
supermartingale. If $\mathcal I_j$ is empty, this holds because
the state does not change.

Let $\sigma$ be the termination time of a single trial. By
Lemma~\ref{lem:conversion-termination}, it is finite almost surely.
At termination, $T_\sigma=\varnothing$ and the finite backup
$b_\sigma$ is realized by a copy in $O_\sigma$.
The potential is bounded and $|T|\le 1/\delta$, so applying the
supermartingale inequality gives
\begin{align}
 \E[d(j,O_\sigma)^p]
 &\le\E[b_\sigma]
 \nonumber\\
 &\le\E[P_j(\varnothing,b_\sigma)]
 \label{eq:conversion-terminal-cost}\\
 &\le(1+\eta)^{1/\delta} P_j(F_j,\infty)
 \le(1+\eta)^{1/\delta}\mathcal B_j.
 \label{eq:conversion-trial-cost}
\end{align}
\eqref{eq:conversion-terminal-cost} uses
$P_j(\varnothing,b_\sigma)\ge b_\sigma$ from
Definition~\ref{def:conversion-certificate}.
For \eqref{eq:conversion-trial-cost}, use the supermartingale bound,
$T_\sigma=\varnothing$, and the initial bound.

Let $\mathcal G$ be the event that the trial meets the
facility budget. Lemma~\ref{lem:conversion-budget} gives
$\Pr(\mathcal G)\ge1-\zeta$. Consequently,
$$\begin{aligned}
 \E[d(j,O_\sigma)^p\mid\mathcal G]
 =\frac{\E[d(j,O_\sigma)^p\mathbf1_{\mathcal G}]}
         {\Pr(\mathcal G)}\nonumber
 \le\frac{(1+\eta)^{1/\delta}}{1-\zeta}\mathcal B_j.
\end{aligned}$$
This proves \eqref{eq:conversion-transfer}.

We have
$$   \frac{(1+\varepsilon\delta/4)^{1/\delta}}{1-\varepsilon/4}
 \le\frac{e^{\varepsilon/4}}{1-\varepsilon/4}
 \le\frac1{(1-\varepsilon/4)^2}
 \le1+\varepsilon.
$$

The first two inequalities use $1+\varepsilon\delta/4\le e^{\varepsilon\delta/4}$ and
$\varepsilon/4\le-\log(1-\varepsilon/4)$. For the last inequality,
$(1+\varepsilon)(1-\varepsilon/4)^2-1
=(\varepsilon/4)(2-7\varepsilon/4+\varepsilon^2/4)\ge0$ for $0<\varepsilon\le1$.
Substituting $\eta^{-2}=16\delta^{-2}\varepsilon^{-2}$ into the definition of $s$
yields the stated bound on $s$.
\end{proof}

\subsection{The Dispersion Potential and the Client-Dependent Certificate}
\label{sec:conversion-dispersion}

The client $j$ remains fixed in this section. Throughout we focus on the case $p = 1$, i.e., the metric $k$-median problem.  In this case, the backup cost $b$ is also a distance. Recall that $V(U)=\sum_{i\in U}\delta d_i$ in this case. Define
$$
\Disp(U)=\sum_{i,q\in U}\delta^2|d_i-d_q|,\qquad
V_j=V(F_j),\quad \Disp_j=\Disp(F_j).
$$
 For $U\subseteq F_j$, let
$\Phi(U,b)=2V(U)+(1-\delta|U|)b$ be the potential from Section~\ref{sec:client-dependent-analysis}, and define
\begin{equation}\label{eq:conversion-psi}
 \Psi_j(U,b)=\Phi(U,b)-\frac{\Disp(U)}{12},\qquad
 P_j^{\rm M}(T,b)=\min_{U\subseteq T}\Psi_j(U,b).
\end{equation}
Notice that $P_j^{\rm M}(T,b)$ uses the same minimization over $U\subseteq T$ as in~\eqref{eq:psi-def}, except that here we maintain this definition across all reachable states rather than switching at $t^*$. Moreover, in accordance with~\eqref{eq:conversion-analytic-transition}, the backup distance $b$ is updated only on operations relevant to $j$ and therefore serves as an upper bound on the distance to the nearest open facility. We now verify that $P_j^{\rm M}(T,b)$ satisfies the conditions of Definition~\ref{def:conversion-certificate}.

We first record the basic properties needed for the local-charge
certificate. For a fixed $U$, write $m=\delta|U|$ and $V=V(U)$.
The bound $\Disp(U)\le2mV$ from~\eqref{eq:disp-bound} gives
\begin{equation}\label{eq:conversion-psi-nonnegative}
 \Psi_j(U,b)\ge\left(2-\frac m6\right)V+(1-m)b\ge0.
\end{equation}
Here we follow the convention $0\cdot\infty=0$, so $(1-m)b=0$ whenever $m=1$, including at $b=\infty$. For finite $b$, the empty candidate gives
$P_j^{\rm M}(T,b)\le b$. If $b=\infty$,
then $T=F_j$, and only the full candidate has finite value. Consequently,
$$
 P_j^{\rm M}(F_j,\infty)=2V_j-\frac{\Disp_j}{12},\qquad
 P_j^{\rm M}(\varnothing,b)=b,
$$ 
and $P_j^{\rm M}$ is nonnegative and uniformly bounded over all reachable states.

To control the loss of dispersion under partial deletions,
for each $i\in U$, define the endpoint charge
$$
 e_i(U)=\frac\delta6\sum_{v\in U}\delta|d_i-d_v|.
$$
In particular, $\sum_{i\in U}e_i(U)=\Disp(U)/6$. For any
$Q\subseteq U$,
\begin{align}
 \frac{\Disp(U)-\Disp(U\setminus Q)}{12}
 &=\frac16\sum_{i\in Q,\,v\in U\setminus Q}\delta^2|d_i-d_v|
   +\frac1{12}\sum_{i,v\in Q}\delta^2|d_i-d_v|
 \nonumber\\
 &\le\sum_{i\in Q}e_i(U).
 \label{eq:conversion-delete-credit}
\end{align}

We will combine this deletion bound with the saving from Section~\ref{sec:client-dependent-analysis}. For $U\subseteq T$ and $q,i\in U$, let $\mu_{q,i}:=\delta-w_{q,i}$
and retain the quantities
$$
 r_i=\min\{d_q:q\in U,\ \mu_{q,i}>0\},\qquad
 I=\sum_{i,q\in U}\delta\mu_{q,i}\pos{d_q-2d_i},
$$
$$
 \Gamma(U,w)=V-\sum_{i\in U}\delta(m-w(U,i))r_i+I.
$$
The minimum of the empty set is infinity. If $r_i=\infty$,
then $m=w(U,i)$ and $(m-w(U,i))r_i=0$.
These are the definitions in Section~\ref{sec:one-step-saving}
with $y_q=\delta$ and $u_{qi}=\mu_{q,i}$.
Since every facility in $U\subseteq T$ is still fractional, we have $0\le w_{q,i}\le\delta$.
The proof of Lemma~\ref{lem:combined-saving} therefore applies
to these weights and gives
\begin{equation}\label{eq:conversion-saving}
 \Gamma(U,w)\ge\frac{\Disp(U)}4.
\end{equation}

We now use these bounds to verify the local-charge certificate.
\begin{lemma}\label{lem:conversion-disp-cert}
The potential $P_j^{\rm M}(T,b)$ has a local-charge certificate
with initial bound $\mathcal B_j=2V_j-\Disp_j/12$ as in
Definition~\ref{def:conversion-certificate}.
\end{lemma}

\begin{proof}
The initial and terminal values
were established above. It remains to verify
\eqref{eq:conversion-majorant} and \eqref{eq:conversion-balance}.
Fix a current state $(T,b)$ and a minimizing $U\subseteq T$ in
\eqref{eq:conversion-psi}, so $P_j^{\rm M}(T,b)=\Psi_j(U,b)$. Throughout the proof, $x$ denotes
a current token, $q=\operatorname{own}(x)$ is its owner, and
$Q\subseteq N^{-1}(x)$ is an arbitrary set of facilities to be
deleted.
If $U=\varnothing$, then
$$
 P_j^{\rm M}(T^{x,Q},b^{x,Q})\le b^{x,Q}
 \le b=P_j^{\rm M}(T,b),
$$
and all charges can be zero. Therefore we assume $U\ne\varnothing$ and use $m,V,\mu_{q,i},r_i,I$ as defined above.

For each $i\in U$, let
$$
 b_i=\min\{b,2d_i+r_i\},
$$
where $b_i$ bounds the new backup distance when an owner outside $U$
deletes $i$. Thus, if $q\notin U$ and $i\in Q\cap U$, then
\begin{equation}\label{eq:conversion-outside-backup}
 b^{x,Q}=\min\{b,d_q\}\le b_i.
\end{equation}

For every token $x$, set
$g_{j,x,i}=0$ for $i\in N^{-1}(x)\cap(T\setminus U)$.
If $q\notin U$, set
$$
 a_{j,x}=0,\qquad
 g_{j,x,i}=\delta(b_i-2d_i)+e_i(U)
 \quad(i\in N^{-1}(x)\cap U).
$$
If $q\in U$, put $\bar b_q=\min\{b,d_q\}$ and choose
$W_q\subseteq U\setminus\{q\}$ minimizing
$\Psi_j(W_q,\bar b_q)$. Set
$$
 a_{j,x}=\Psi_j(W_q,\bar b_q)-\Psi_j(U,b),\qquad
 g_{j,x,i}=\delta\pos{d_q-2d_i}+e_i(U)
 \quad(i\in N^{-1}(x)\cap U).
$$

We first prove \eqref{eq:conversion-majorant} for every token $x$
and every $Q\subseteq N^{-1}(x)$.
If $H_x(Q)=0$, both sides are zero, so assume $H_x(Q)=1$.
For $q\notin U$, the set $U\setminus Q$ is a candidate for
the next potential. Hence
\begin{align}
 &P_j^{\rm M}(T^{x,Q},b^{x,Q})-P_j^{\rm M}(T,b)
 \nonumber\\
 &\qquad\le\Psi_j(U\setminus Q,b^{x,Q})-\Psi_j(U,b)
 \label{eq:conversion-outside-candidate}\\
 &\qquad=(1-m)b^{x,Q}-(1-m)b
   +\sum_{i\in Q\cap U}\delta(b^{x,Q}-2d_i)
   +\frac{\Disp(U)-\Disp(U\setminus Q)}{12}
 \nonumber\\
 &\qquad\le\sum_{i\in Q\cap U}
       \bigl[\delta(b_i-2d_i)+e_i(U)\bigr].
 \label{eq:conversion-outside-partial}
\end{align}
For \eqref{eq:conversion-outside-candidate}, use
$U\setminus Q\subseteq T^{x,Q}$ and the choice of $U$.
\eqref{eq:conversion-outside-partial} uses
$b^{x,Q}\le b$, \eqref{eq:conversion-delete-credit},
and \eqref{eq:conversion-outside-backup}.

For $q\in U$, we have $b^{x,Q}=\bar b_q$ and may instead use
$W_q\setminus Q$ as a candidate:
\begin{align}
 &P_j^{\rm M}(T^{x,Q},\bar b_q)-P_j^{\rm M}(T,b)
 \nonumber\\
 &\qquad\le a_{j,x}
       +\Psi_j(W_q\setminus Q,\bar b_q)-\Psi_j(W_q,\bar b_q)
 \label{eq:conversion-inside-candidate}\\
 &\qquad= a_{j,x}
       +\sum_{i\in Q\cap W_q}\delta(\bar b_q-2d_i)
       +\frac{\Disp(W_q)-\Disp(W_q\setminus Q)}{12}
 \nonumber\\
 &\qquad\le a_{j,x}+\sum_{i\in Q\cap W_q}
       \bigl[\delta(\bar b_q-2d_i)+e_i(W_q)\bigr]
 \label{eq:conversion-inside-endpoints}\\
 &\qquad\le a_{j,x}+\sum_{i\in Q\cap U}
       \bigl[\delta\pos{d_q-2d_i}+e_i(U)\bigr].
 \label{eq:conversion-inside-partial}
\end{align}
\eqref{eq:conversion-inside-candidate} uses
$W_q\setminus Q\subseteq T^{x,Q}$ and the definition of $a_{j,x}$.
\eqref{eq:conversion-inside-endpoints} applies
\eqref{eq:conversion-delete-credit} to $W_q$.
For \eqref{eq:conversion-inside-partial}, use
$\bar b_q\le d_q$ and $e_i(W_q)\le e_i(U)$.
This proves \eqref{eq:conversion-majorant}.

We next prove \eqref{eq:conversion-balance}.
Recall that $r_x$ is the number of operations associated with $x$.
If $q\in U$, its unique token $x$ has
$H_x(Q)=1$ for every $Q$ with $(x,Q)\in\mathcal A$.
To control all $r_x$ occurrences of its opening charge,
use $\varnothing$ and $U\setminus\{q\}$ as candidates in
the definition of $W_q$. They give
\begin{align}
 a_{j,x}
 &\le \min\{\Psi_j(\varnothing,\bar b_q),
             \Psi_j(U\setminus\{q\},\bar b_q)\}
       -\Psi_j(U,b)
 \nonumber\\
 &\le \min\{\bar b_q-\Psi_j(U,b),
             -2\delta d_q+\delta\bar b_q+e_q(U)\}
 \label{eq:conversion-repeat-deletion}\\
 &\le \min\{d_q-V,-\delta d_q+e_q(U)\}
 \label{eq:conversion-repeat-candidates}\\
 &\le \min\left\{d_q-V,
             -\delta d_q+\frac{\delta}{6}(md_q+V)\right\}
 \label{eq:conversion-repeat-endpoint}\\
 &\le \min\left\{d_q-V,\frac{\delta}{6}(V-5d_q)\right\} \nonumber\\
 &\le 0.
 \label{eq:conversion-repeat-balance}
\end{align}
For \eqref{eq:conversion-repeat-deletion}, use
$\Psi_j(\varnothing,\bar b_q)=\bar b_q$ for the first candidate and $(1-m)\bar b_q-(1-m)b\le0$ for the second candidate. \eqref{eq:conversion-repeat-candidates} uses
$\bar b_q\le d_q$ and $\Psi_j(U,b)\ge V$ from
\eqref{eq:conversion-psi-nonnegative}. For \eqref{eq:conversion-repeat-endpoint}, use
$|d_q-d_v|\le d_q+d_v$ in the definition of $e_q(U)$.
For \eqref{eq:conversion-repeat-balance}, the first term of the
minimum is nonpositive when $d_q\le V$, whereas the second
is negative when $d_q>V$.
Consequently, $r_xa_{j,x}\le a_{j,x}$.

We bound one occurrence of the opening charge as follows:
\begin{align}
 a_{j,x}
 &\le\min_{W\subseteq U\setminus\{q\}}\Phi(W,d_q)-\Psi_j(U,b)
 \nonumber\\
 &=d_q-\sum_{i\in U}\delta\pos{d_q-2d_i}
       -\Phi(U,b)+\frac{\Disp(U)}{12}.
 \label{eq:conversion-opening-bound}
\end{align}

Multiply the LHS of \eqref{eq:conversion-balance}
by $\delta$.
For $q\ne i$, \eqref{eq:conversion-token-weights} shows that
the number of tokens owned by $q$ whose deletion sets contain
$i$ is $w_{q,i}/\delta$. Each such token deletes $i$ in exactly
one of its operations.
Therefore, the opening charges and the deletion charges from
owners in $U$, excluding endpoint charges, satisfy
\begin{align}
 &\delta\sum_{x:\operatorname{own}(x)\in U}r_xa_{j,x}
       +\sum_{i,q\in U}\delta w_{q,i}\pos{d_q-2d_i}
 \nonumber\\
 &\qquad\le\sum_{q\in U}\delta\left[
 d_q-\sum_{i\in U}\delta\pos{d_q-2d_i}
       -\Phi(U,b)+\frac{\Disp(U)}{12}\right]
       +\sum_{i,q\in U}\delta w_{q,i}\pos{d_q-2d_i}
 \label{eq:conversion-inside-type-charges}\\
 &\qquad=V-I-m\Phi(U,b)+\frac{m\Disp(U)}{12}.
 \label{eq:conversion-inside-balance}
\end{align}
\eqref{eq:conversion-inside-type-charges} uses
$r_xa_{j,x}\le a_{j,x}$ and \eqref{eq:conversion-opening-bound}.
\eqref{eq:conversion-inside-balance} uses
$\mu_{q,i}=\delta-w_{q,i}$ and the definition of $I$.

The deletion charges from owners outside $U$, excluding
endpoint charges and multiplied by $\delta$, sum to
$\sum_{i\in U}\delta(1-w(U,i))(b_i-2d_i)$.
Combining this expression with the terms in
\eqref{eq:conversion-inside-balance} gives
\begin{align}
 &V-I-m\Phi(U,b)
       +\sum_{i\in U}\delta(1-w(U,i))(b_i-2d_i)
 \nonumber\\
 &\qquad\le V-I-2mV-m(1-m)b
 \nonumber\\
 &\qquad\quad+\sum_{i\in U}\delta\Bigl[
       (m-w(U,i))(2d_i+r_i)
       +(1-m)b-2(1-w(U,i))d_i\Bigr]
 \label{eq:conversion-outside-mixture}\\
 &\qquad=V-I-2mV+2(m-1)V
       +\sum_{i\in U}\delta(m-w(U,i))r_i
 \nonumber\\
 &\qquad=-V-I+\sum_{i\in U}\delta(m-w(U,i))r_i
 \nonumber\\
 &\qquad=-\Gamma(U,w).
 \label{eq:conversion-outside-saving}
\end{align}
For \eqref{eq:conversion-outside-mixture}, substitute
$\Phi(U,b)=2V+(1-m)b$ and split
$1-w(U,i)=(m-w(U,i))+(1-m)$.
Since $b_i=\min\{b,2d_i+r_i\}$, we can bound the two corresponding
terms by $(m-w(U,i))(2d_i+r_i)$ and $(1-m)b$.

Finally, for every $i\in U$, the set $N(i)$ contains
$\delta^{-1}-1$ tokens other than its own token.
Each of them deletes $i$ in exactly one operation.
After multiplication by $\delta$, the endpoint charges
therefore sum to
$$
 \sum_{i\in U}(1-\delta)e_i(U)
 =\frac{(1-\delta)\Disp(U)}6
 \le\frac{\Disp(U)}6.
$$
Combining this with \eqref{eq:conversion-inside-balance}
and \eqref{eq:conversion-outside-saving}, we obtain
\begin{align}
 &\delta\sum_x\left(a_{j,x}\sum_{Q:(x,Q)\in\mathcal A}H_x(Q)
       +\sum_{i\in N^{-1}(x)\cap T}g_{j,x,i}\right)
 \nonumber\\
 &\qquad\le -\Gamma(U,w)
       +\frac{(m+2)\Disp(U)}{12}
 \nonumber\\
 &\qquad\le \frac{(m-1)\Disp(U)}{12}
 \le0.
 \label{eq:conversion-disp-saving}
\end{align}
\eqref{eq:conversion-disp-saving} uses
\eqref{eq:conversion-saving} and $m\le1$.
This proves \eqref{eq:conversion-balance}
and completes the certificate.
\end{proof}

Combining Theorem~\ref{thm:conversion-local-charges} with
Lemma~\ref{lem:conversion-disp-cert}, we obtain a rounding algorithm that,
for every fixed preprocessed solution, returns a set $O$ satisfying
\begin{align*}
 |O|&\le Y_0+s,\\
 \E[d(j,O)]&\le(1+\varepsilon)\left(2V_j-\frac{\Disp_j}{12}\right)
 \qquad\text{for every }j\in C.
\end{align*}
Here $s=O(\delta^{-2}\varepsilon^{-2}\log(4/\varepsilon))$, and the
expectation is over the rounding trials, including rejection of trials
that exceed the facility budget. Since $Y_0\le k+3$ and
$\delta=\varepsilon^{O(1)}$, the algorithm always opens at most
$k+O_\varepsilon(1)$ facilities and runs in expected polynomial time for
fixed $\varepsilon$. Thus the conversion preserves the dispersion saving
from Section~\ref{sec:client-dependent-analysis}, up to a multiplicative
$1+\varepsilon$ loss. In the final combination, both rounding algorithms
will be applied to this same preprocessed solution.

\subsection{The Euclidean Certificate for Power Costs}\label{sec:conversion-euclidean}
We now verify the local-charge certificate under the pairwise condition
used in Section~\ref{sec:euclidean-lmp}.
Let $p\ge1$ and $\alpha\ge1$ satisfy the pairwise
condition~\eqref{eq:pairwise} of Lemma~\ref{lem:nonstop-lmp}.
For a fixed client $j$, define the potential
$$
 P_j^{\rm E}(T,b)=\sum_{i\in T}\delta\min\{\alpha\kappa_i,b\}+(1-\delta|T|)b=\min_{U\subseteq T}\bigl\{\alpha V(U)+(1-\delta|U|)b\bigr\}.
$$
Different from $\Phi_p(S,b)$ as defined in Section~\ref{sec:euclidean-lmp}, throughout this subsection, $b$ denotes the cost of the backup facility rather than the distance to the client. For finite $b$, $P_j^{\rm E}(T,b)$ is nondecreasing in
$b$, and at most $b$.
Under the convention $0\cdot\infty=0$, we have
$$
 P_j^{\rm E}(F_j,\infty)=\alpha\cost_{\vecy}(j),\qquad
 P_j^{\rm E}(\varnothing,b)=b,
$$
and $P_j^{\rm E}$ is nonnegative and uniformly bounded over all reachable states.

\begin{lemma}\label{lem:conversion-euclidean-cert}
The potential $P_j^{\rm E}(T,b)$ has a local-charge certificate
with initial bound $\mathcal B_j=\alpha\cost_{\vecy}(j)$ as in
Definition~\ref{def:conversion-certificate}.
\end{lemma}

\begin{proof}
It remains
to verify \eqref{eq:conversion-majorant} and
\eqref{eq:conversion-balance}.
Fix a current state $(T,b)$ before iteration $t$, set $U=\{i\in T:\alpha\kappa_i<b\}$ if $b<\infty$, and $U=T=F_j$ if $b=\infty$.
Writing $m=\delta|U|$, we have
$$
 P_j^{\rm E}(T,b)=P_j^{\rm E}(U,b)
 =\alpha V(U)+(1-m)b.
$$
Throughout the proof, $x$ denotes a current token,
$q=\operatorname{own}(x)$ is its owner, and
$Q\subseteq N^{-1}(x)$ is an arbitrary set of facilities to be
deleted. If $U=\varnothing$, then
$$
 P_j^{\rm E}(T^{x,Q},b^{x,Q})\le b^{x,Q}
 \le b=P_j^{\rm E}(T,b),
$$
and all charges can be zero. Therefore we assume $U\ne\varnothing$.

For every token $x$, put $\bar b_q=\min\{b,\kappa_q\}$ and set
$$
 a_{j,x}=\begin{cases}
 P_j^{\rm E}(U,\bar b_q)-P_j^{\rm E}(U,b),&q\in U,\\
 0,&q\notin U,
 \end{cases}
 \qquad
 g_{j,x,i}=\delta\pos{\bar b_q-\alpha\kappa_i}
 \quad(i\in N^{-1}(x)\cap T).
$$
By the definition of $U$, the deletion charges $g_{j,x,i}=0$ for $i\notin U$.

We first prove \eqref{eq:conversion-majorant} for every token $x$
and every $Q\subseteq N^{-1}(x)$.
If $H_x(Q)=0$, both sides are zero, so assume $H_x(Q)=1$.
Then $b^{x,Q}=\bar b_q$, and
$U'=U\setminus(\{q\}\cup Q)\subseteq T^{x,Q}$. Hence
\begin{align*}
 \quad& P_j^{\rm E}(T^{x,Q},b^{x,Q})-P_j^{\rm E}(T,b)\\
 &\le P_j^{\rm E}(U',\bar b_q)-P_j^{\rm E}(U,b)\\
 &=P_j^{\rm E}(U,\bar b_q)-P_j^{\rm E}(U,b)+\sum_{i\in U\cap(\{q\}\cup Q)}
       \delta\pos{\bar b_q-\alpha\kappa_i}
       \\
 &= P_j^{\rm E}(U,\bar b_q)-P_j^{\rm E}(U,b)+\sum_{i\in Q\cap U}g_{j,x,i}+\textbf 1_{q\in U} \cdot \delta\pos{\bar b_q-\alpha \kappa_q}\\
 &\le a_{j,x}+\sum_{i\in Q\cap T}g_{j,x,i}.
\end{align*}
The last inequality uses $\bar b_q\le\kappa_q\le\alpha\kappa_q$
and, when $q\notin U$, the monotonicity of $P_j^{\rm E}(U,b)$ in $b$.
This proves \eqref{eq:conversion-majorant}.

We next prove \eqref{eq:conversion-balance}.
By \eqref{eq:conversion-token-weights}, the deletion charges,
multiplied by $\delta$, sum to
$$
 \delta\sum_x\sum_{i\in N^{-1}(x)\cap T}g_{j,x,i}
 =\sum_{i\in U}\delta\sum_{q\in R_{t-1}\cup O_{t-1}}w_{q,i}
       \pos{\bar b_q-\alpha\kappa_i}.
$$
Here the self-token terms vanish because
$\bar b_i\le\alpha\kappa_i$.
Applying \eqref{eq:nonstop-backup}, with
$S=U$, $y_i=\delta$, and backup distance $b^{1/p}$, gives
\begin{align}
 \delta\sum_x\sum_{i\in N^{-1}(x)\cap T}g_{j,x,i}
 &\le\sum_{i,q\in U}\delta^2
       \pos{(d_i+d(i,q))^p-\alpha\kappa_i}
       +(1-m)\sum_{i\in U}\delta(b-\kappa_i)
 \nonumber\\
 &\le\sum_{i,q\in U}\delta^2\pos{\alpha\kappa_i-\kappa_q}
       +(1-m)\sum_{i\in U}\delta(b-\kappa_i).
 \label{eq:conversion-euclidean-pairwise}
\end{align}
The last inequality is \eqref{eq:pairwise} with weights
$y_i=\delta$ on $U$.

For $q\in U$, we have $\bar b_q=\kappa_q$, and hence
\begin{align}
 a_{j,x}
 &=\sum_{i\in U}\delta
       \bigl(\min\{\alpha\kappa_i,\kappa_q\}-\alpha\kappa_i\bigr)
       -(1-m)(b-\kappa_q)
 \nonumber\\
 &=-\sum_{i\in U}\delta\pos{\alpha\kappa_i-\kappa_q}
       -(1-m)(b-\kappa_q)\le0.
 \label{eq:conversion-euclidean-opening}
\end{align}
Every $q\in U$ is fractional and owns a unique token $x$.
All $r_x\ge1$ operations associated with this token are relevant,
so $r_xa_{j,x}\le a_{j,x}$.
Multiplying the LHS of \eqref{eq:conversion-balance} by $\delta$
and using \eqref{eq:conversion-euclidean-opening}, we obtain
\begin{align*}
 &\delta\sum_x\left(a_{j,x}\sum_{Q:(x,Q)\in\mathcal A}H_x(Q)
       +\sum_{i\in N^{-1}(x)\cap T}g_{j,x,i}\right)\\
 &\qquad=\delta\sum_x\sum_{i\in N^{-1}(x)\cap T}g_{j,x,i}
       +\delta\sum_{x:\operatorname{own}(x)\in U}r_xa_{j,x}\\
 &\qquad\le\delta\sum_x\sum_{i\in N^{-1}(x)\cap T}g_{j,x,i}
       -\sum_{q\in U}\delta\left[
          \sum_{i\in U}\delta\pos{\alpha\kappa_i-\kappa_q}
          +(1-m)(b-\kappa_q)\right]\\
 &\qquad\le0,
\end{align*}
where the last inequality follows from
\eqref{eq:conversion-euclidean-pairwise}.
This proves \eqref{eq:conversion-balance} and completes the proof.
\end{proof}

Finally, we put everything together to prove Theorems~\ref{thm:euclidean-kmedian}
and~\ref{thm:euclidean-kmeans}.

\begin{proof}[Proof of Theorems~\ref{thm:euclidean-kmedian}
and~\ref{thm:euclidean-kmeans}]
Fix a constant $p\in\{1\}\cup[2,\infty)$, and let
$\ell=\lceil p\rceil$.
Put $\alpha_1=\alpha$, where $\alpha$ is defined
by~\eqref{eq:median-alpha}, and
$\alpha_p=(3^p+2)/3$ for $p\ge2$.
Lemmas~\ref{lem:euclidean-median} and~\ref{lem:euclidean-power}
show that the pairwise condition~\eqref{eq:pairwise} holds with
factor $\alpha_p$.

We first solve the LP relaxation \eqref{lp:k-clustering}, and apply the preprocessing step with a sufficiently small constant
$\varepsilon_0>0$ in Section~\ref{sec:preprocessing} to obtain a $\delta$-integral solution, where
$\delta=\varepsilon_0^{14\ell+2}$.
Let $(\vecx',\vecy')$ denote the new LP
 solution and
$\mathcal H=\{y'(F)\le k+3\}$.
We repeat preprocessing until $\mathcal H$ holds. By Theorem~\ref{thm:preprocessing}, for every client $j$, we obtain
$$
 \E[\cost_{\vecy'}(j)\mid\mathcal H]
 \le
 \frac{1+O(p\varepsilon_0)}
      {1-\exp(-\Omega(\varepsilon_0^{-10\ell}))}
 \cost_{\vecy}(j)
 = (1+O(\varepsilon_0))\cost_{\vecy}(j).
$$

Conditioned on $(\vecx',\vecy')$, apply the conversion algorithm
in Section~\ref{sec:conversion-algorithm} with accuracy parameter
$\varepsilon_0$.
By Theorem~\ref{thm:conversion-local-charges} and
Lemma~\ref{lem:conversion-euclidean-cert}, the accepted output
$O$ satisfies
$$
 \E[d(j,O)^p\mid(\vecx',\vecy')]
 \le (1+\varepsilon_0)\alpha_p\cost_{\vecy'}(j)
 \qquad\text{for every }j\in C.
$$
Averaging over the accepted preprocessing output, we obtain
\begin{align*}
 |O|&\le k+3+
 O\!\left(\delta^{-2}\varepsilon_0^{-2}
          \log(4/\varepsilon_0)\right),\\
 \E[d(j,O)^p]
 &\le(1+O(p\varepsilon_0))\alpha_p\cost_{\vecy}(j)
 \qquad\text{for every }j\in C.
\end{align*}
Thus, for fixed $\varepsilon_0$, we obtain a pseudo-approximation with a
constant additive facility allowance.

We now apply the reduction of Li and Svensson~\cite[Theorem~4]{LS16},
using its extension to power costs in
\cite[Theorem~A.1]{BGHL2026} when $p\ge2$.
Thus, for any fixed $\tau>0$, we obtain a randomized
polynomial-time algorithm opening at most $k$ facilities with approximation ratio $(1+O(\varepsilon_0))\alpha_p+\tau$.

For $p=1$, choose $\varepsilon_0$ and $\tau$ sufficiently small and use
$\alpha_1<1.71955$ to obtain the claimed factor.
For $p\ge 2$, the total additional loss can be made smaller than
any prescribed $\varepsilon>0$, proving the claimed
$((3^p+2)/3+\varepsilon)$-approximation.
\end{proof}
\section{Robust Algorithm}\label{sec:robust-algorithm}

% \ifpaperchanges
% \gptmargin{Restored the previous text in gray for comparison. The blue passages give the revision; old commented drafts are also preserved in the source.}
% \fi
% \gptremovedfile{removed-text/robust1736-01}

% \begin{gptadded}
In \Cref{sec:client-dep}, the algorithm achieves the approximation ratio of 2 only when the connection distances in the support are the same.
In this section, we present an algorithm that complements the previous algorithm, yielding an approximation ratio below $2$. %\ysmargin{Include an intro for the section.}

% Let $F$ and $C$ denote a set of facilities and clients, respectively. 
For a fractional LP solution $(x,y)$, split facilities into colocated copies so that each client has a complete nearest-unit support. 
For client $j \in C$, we again use $F_j \subseteq F$ and $V_j = \sum_{i \in F_j} y_i \, d(i, j)$ for the set of supporting facilities of $j$ and its fractional cost in the LP solution.
% For client $j\in C$, let $F_j$ denote its nearest facilities within one unit of opening. In other words, \[
%     F_j:= \arg\min_{S\subseteq F: y(S)=1} \sum_{i\in S} y_i d(i,j).
% \]
% Let $V_j$ denote the average distance from $j$ to $F_j$, i.e.,
% \(
%     V_j := \sum_{i \in F_j} y_i \, d(i, j).
% \)
We further define the mean absolute deviation 
\[
    M_j:=\sum_{i\in F_j} y_i|d(i,j)-V_j|.
\]
Observe that $M_j$ is always bounded above by its dispersion $\Disp_j$. %\ysmargin{Delete the derivation for $M_j \leq \Disp_j$.}

The main theorem we prove in this section is as follows:
% \gptmargin{Made splitting explicit and stated the total-opening budget. If $y(F)\le k$ for integer $k$, this gives at most $k$ facilities; after preprocessing it gives at most $k+3$.}
\begin{theorem} \label{thm:robust}
    There exists a randomized polynomial-time algorithm that, given a complete fractional solution $(x, y)$ with $y(F)\ge1$, returns a facility set $F_\open \subseteq F$ with $|F_\open| \leq \lceil y(F)\rceil$, satisfying for every $j \in C$, \[
        \E[d(j,F_\open)] \le \min\left\{\left(16+\frac{5}{e}\right)V_j, \left(1+\frac{2}{e}\right)V_j+\sqrt{\frac{90 V_jM_j}{e}}+\frac{15}{2}M_j\right\},
    \]
    where the expectation is over the randomness of $F_\open$.
\end{theorem}
% \yscomment{Do we essentially use $(16+5/e)V_j$?}
% \end{gptadded}

% \begin{gptadded}
% When we define the dispersion of client $j$ as
% \[
%     \mathsf{disp}_j := \sum_{i \in F_j} \sum_{i' \in F_j} y_i y_{i'} |d(i, j) - d(i', j)|,
% \]
% and the normalized dispersion $s_j := \mathsf{disp}_j / V_j$ for $V_j>0$ (set $s_j=0$ if $V_j=0$), the above theorem implies
% \[
%     \E[d(j, F_\open)] \leq \min \Big\{ 16 + \frac{5}{e}, 1 + \frac{2}{e} + \sqrt{\frac{90}{e} s_j} + \frac{15}{2}s_j \Big\} V_j
% \]
% since
% \begin{align*}
%     M_j 
%     & 
%     = \sum_{i\in F_j} y_i|d(i,j)-V_j|
%     % \\&
%     = \sum_{i\in F_j} y_i \Big| d(i,j) - \sum_{i' \in F_j} y_{i'} d(i', j) \Big|\\&
%     \overset{(a)}{=} \sum_{i\in F_j} y_i \Big| \sum_{i' \in F_j} y_{i'} (d(i,j) - d(i', j)) \Big|
%     \\&
%     \overset{(b)}{\leq} \sum_{i \in F_j} y_i \Big( \sum_{i' \in F_j} y_{i'} |d(i,j) - d(i', j)| \Big)
%     % \\&
%     = \mathsf{disp}_j,
% \end{align*}
% where $(a)$ follows from the fact that $\sum_{i' \in F_j} y_{i'} = 1$, and $(b)$ comes from the fact that the absolute value of the sum is always no greater than the sum of absolute values.
% \end{gptadded}

\subsection{Clients and Their Radii}

Due to $y(F_j)=1$, we may view $\{(d (i, j), y_i) \}_{i \in F_j}$ as a probability distribution, denoted by $\calP_j$.
Consequently, we have
\[
    \E_{X \sim \calP_j}[X] = \sum_{i \in F_j} y_i \, d(i, j) = V_j.
    % V_j:= \sum_{i\in F_j} y_i d(i,j)=\E[X].
\]

For each client $j$, we say a radius $R$ is \emph{admissible} for $j$ if the total $y$-mass of the closed ball with radius $R$ centered at $j$ on $\calP_j$ is at least $1/2$, i.e., 
\begin{equation} \label{eq:majority-mass}
    y(\{i\in F_j: d(i,j)\le R\})\ge \frac{1}{2}.
\end{equation}
Given a radius $R$, we denote by $T_j(R)$ the contribution to the expectation of $\calP_j$ from outside the radius-\(R\) ball centered at $j$, i.e.,
% ``partial expectation'' of $\calP_j$ outside the radius-$R$ ball centered at $j$, i.e.,
\[
    T_j(R) := \E_{X \sim \calP_j}[X \, \mathbf{1}_{X > R}] = \sum_{i \in F_j : d(i, j) > R} y_i \, d(i, j).
\]
By definition, we always have $T_j(R) \leq V_j$ for any $R$.

For each client $j$, we find a \emph{threshold radius} $R_j$ as the radius $R$ minimizing $R + 5e \cdot T_j (R) $ subject to the constraint that $R$ is admissible for $j$.
Since $|\mathsf{supp}(\calP_j)| = O(|F|)$, one can easily compute $R_j$ in polynomial time.
For notational simplicity, let $T_j := T_j(R_j)$ and define the \emph{assurance radius} $ A_j:= R_j + 5e \cdot T_j $ as the value of the objective function in calculating the threshold radius $R_j$.

We have three upper bounds for the assurance radius $A_j$ of a client $j$.
In the regular case where every facility in $F_j$ is at the same distance from $j$, we have a clear characterization of $A_j$ as follows.

% \gptremovedfile{removed-text/robust1736-02}

% \begin{gptadded}
\begin{lemma} \label{lem:assurance-bound-regular}
    For a client $j \in C$, if every facility $i \in F_j$ is at the same distance $d(i, j) = V_j$ from $j$, we then have $A_j = R_j = V_j$.
\end{lemma}
\begin{proof}
    Note that radius $V_j$ is admissible for $j$ in this uniform distance case, achieving $V_j$ as the objective value in computing $R_j$.
    No radius $R<V_j$ is admissible, since its ball contains no support mass.
    % \gptmargin{Corrected the uniform-distance proof: radii below $V_j$ are inadmissible. Also added cardinality bars to the support size above.}
    For any radius $R > V_j$, the objective value is $R$.
    These together imply that $V_j$ is the optimal choice for $R_j$, whose objective value $A_j$ is $V_j$.
\end{proof}
% \end{gptadded}

In the general case, the next lemma provides an upper bound independent of the distance profile of client $j$.
% \yscomment{Do we essentially need this lemma?}

% \gptremovedfile{removed-text/robust1736-03}

% \begin{gptadded}
\begin{lemma} \label{lem:assurance-bound1}
    For any client $j\in C$, $A_j\le (2+5e)V_j.$
\end{lemma}
\begin{proof}
    If $V_j=0$, then $X \sim \calP_j$ is 0 almost surely, and hence $A_j=0$ by \Cref{lem:assurance-bound-regular}. 
    Otherwise, for $X \sim \calP_j$, we have $\Pr(X > 2V_j) \le 1/2$ due to Markov's inequality, implying that $2V_j$ is admissible for $j$ since $\Pr(X \leq 2 V_j) = y(\{i \in F_j : d(i, j) \leq 2 V_j \}) \geq 1/2$.
    We thus have
    \begin{align*}
        A_j & \le 2V_j+5e\cdot T_j(2V_j) 
        && \text{(by optimality of $R_j$)}\\
        & \le (2+5e)V_j
        && \text{($T_j(R)\le V_j$ for all $R \ge 0$)}.\qedhere
    \end{align*}
\end{proof}
% \end{gptadded}

Finally, to obtain a profile-dependent bound on $A_j$, we have the following  properties of the mean absolute deviation of a non-negative distribution. 

% \gptremovedfile{removed-text/robust1736-04}

% \begin{gptadded}
\begin{lemma} \label{obs:mad-basic}
    For any client $j \in C$ and $X \sim \calP_j$, we have
    \begin{enumerate}
        \item $M_j \leq 2 \, V_j$;
        \item $\E[(X - V_j)_+] = \frac{M_j}{2}$; and
        \item for any $t > 0$, $\Pr(X > V_j + t) = y(\{i \in F_j : d(i, j) > V_j + t \}) \leq \frac{M_j}{2t}$.
    \end{enumerate}
\end{lemma}
\begin{proof}
    Recall that, by definition, we have
    \[
        \E[|X - V_j|] = \sum_{i \in F_j} y_i \, |d(i, j) - V_j| = M_j.
    \]
    Since $X$ is always non-negative, $V_j\ge 0$. The first statement then follows by
    \[
        M_j = \sum_{i \in F_j} y_i \, |d(i, j) - V_j| \leq \sum_{i \in F_j} y_i \, (d(i, j) + V_j)
        = 2 \, V_j.
    \]
    For the second statement, we split $\E[|X - V_j|]$ into two values, depending on whether the distance is at least or less than $V_j$:
    \[
        \E[|X - V_j|] = \E[(X - V_j) \mathbf{1}_{X \geq V_j}] + \E[(V_j - X) \mathbf{1}_{X < V_j}].
    \]
    Note that these two values are equal since
    \[
        \E[(X - V_j) \mathbf{1}_{X \geq V_j}] - \E[(V_j - X) \mathbf{1}_{X < V_j}]
        = \E[X - V_j]
        = 0,
    \]
    showing the second statement.
    The last statement is then immediately derived from Markov's inequality since $\Pr(X > V_j + t) = \Pr(X - V_j > t)$ and $\E[(X - V_j)_+] = M_j/2$.
\end{proof}
% \end{gptadded}

We can now prove our profile-dependent upper bound on the assurance distance of client $j$.

% \gptremovedfile{removed-text/robust1736-05}

% \begin{gptadded}
\begin{lemma} \label{lem:assurance-bound2}
    For any client $j\in C$, $A_j\le V_j+\sqrt{10e\cdot V_jM_j}+\frac{5e}{2}\cdot M_j$.
\end{lemma}
\begin{proof}
    If $V_j = 0$ or $M_j = 0$, then $\calP_j$ is supported on the equal distance $V_j$.
    We thus have $A_j = V_j$ due to \Cref{lem:assurance-bound-regular}, implying the lemma. 
    % If $M_j=0$, all support distances equal $V_j$, so $A_j=V_j$ and the bound holds.
    We now assume $V_j > 0$ and $M_j>0$. 
    Let $t := \sqrt{\frac{5e\cdot V_jM_j}{2}} > 0$.
    % \gptmargin{Handled zero deviation before dividing by $t$. Defined zero normalized dispersion and handled $V_j=0$ in the Markov argument above.}
    By the first statement of \Cref{obs:mad-basic}, we derive that $t \geq M_j$, further implying that $R = V_j + t$ is admissible for $j$ due to the third statement of \Cref{obs:mad-basic}.
    Observe that
    \begin{align*}
        T_j(V_j+t) & =\E [X \, \mathbf{1}_{X>V_j+t}] 
        = \E[V_j \, \mathbf{1}_{X>V_j+t}]
        + \E[(X-V_j) \mathbf{1}_{X>V_j+t}] 
        \\&
        \textstyle \leq V_j\cdot \Pr(X>V_j+t)+\E[(X-V_j)_+] 
        \\& 
        \le \frac{V_jM_j}{2t} + \frac{M_j}{2},
    \end{align*}
    where the last inequality follows from the second and third statements of \Cref{obs:mad-basic}.
    By the optimality of $R_j$, we have
    \begin{align*}
        A_j 
        & 
        \le (V_j + t) + 5e\cdot T_j(V_j+t) 
        \le V_j +t+\frac{5e\cdot V_jM_j}{2t}+\frac{5e}2M_j\\&
        = V_j + \sqrt{10e\cdot V_jM_j} +\frac{5e}{2}M_j.\qedhere
    \end{align*}
\end{proof}
% \end{gptadded}

\subsection{Algorithm}
The algorithm consists of two steps: The first step is to construct disjoint bundles of representatives by filtering, and the second step is to apply a dependent rounding on a maximum matching of the representatives.

\paragraph{Filtering.}
The filtering procedure starts with all clients and proceeds in rounds until no client remains.
For each round,
\begin{itemize}
    \item we select a client $r$ with the minimum $A_r$ among the remaining clients, and make it a representative;
    \item we then remove every remaining client $j$ with $d(j,r)\le 2\max\{R_j,R_r\}$;
\end{itemize}
% 
% \gptremovedfile{removed-text/robust1736-06}
% 
% \begin{gptadded}
Let $\calR\subseteq C$ be the set of representatives from the above filtering procedure, and for every non-representative client $j \in C\setminus \calR $, let $r(j) \in \calR$ denote the representative that removes client $j$; we further define $r(j) = j$ for $j \in \calR$.
Following are simple observations from the execution of the procedure.
% \end{gptadded}

% \begin{gptadded}
\begin{observation} \label{obs:filter}
    For every $j\in C\setminus \calR$, we have 
    $
    % \begin{equation} \label{eq:client-rep-assurance}
        A_{r(j)} \le A_{j}
    % \end{equation}
    $
    and
    $
    % \begin{equation} \label{eq:client-and-rep}
        d(j,r(j)) \le 2 \max\{R_j,R_{r(j)}\}.
    % \end{equation}
    $
    For every pair of distinct $r,s\in \calR$,
    $
    % \begin{equation} \label{eq:rep-and-rep}
        d(r,s) > 2 \max\{R_r, R_s\}.
    % \end{equation}
    $
\end{observation}
\begin{proof}
    The first two inequalities follow from the fact that $j$ remains unselected when $r(j)$ is selected.
    For the last inequality, if distinct $r, s \in \calR$ satisfy $d(r, s) \leq 2 \max \{R_r, R_s\}$, one of $r$ and $s$ must have been deleted when the other was selected, contradicting that both are representatives.
\end{proof}
% \end{gptadded}

\paragraph{Disjoint bundles.}
For a representative $r$, we define its \emph{bundle radius} 
\(
    B_r := \frac{1}{2}d(r,\calR\setminus \{r\})
\)
and its bundle 
\(
    U_r := \{i \in F_r: d(i,r)<B_r\}.
\)
For the degenerate case $\calR = \{r\}$, we simply define $B_r = \infty$.

% \gptremovedfile{removed-text/robust1736-07}

% \begin{gptadded}
% \gptmargin{Corrected $R$ to $\mathcal R$ and restricted both pairwise statements to distinct representatives; also removed the stray period after $j$.}
\begin{observation}[Disjointness]
    For every pair of distinct $r, s\in \calR$, $U_r\cap U_s=\emptyset$.
\end{observation}
\begin{proof}
    Assume toward contradiction that there exists a facility $i \in U_r \cap U_s$.
    By definition, we have $d(i,r)<\frac{1}{2}d(r,s)$ and $d(i,s)<\frac{1}{2}d(r,s)$. Hence, the triangle inequality of $d$ gives a contradiction since $d(r,s)\le d(i,r)+d(i,s)<d(r,s)$.
\end{proof}
% \end{gptadded}

% \begin{gptadded}
\begin{lemma}
    For every $r\in \calR$, $B_r\ge R_r$. Consequently, 
    % \begin{equation} \label{eq:bundle-value}
        $\frac{1}{2}\le y(U_r)\le 1.$
    % \end{equation}
\end{lemma}
\begin{proof}
    If $\calR=\{r\}$, then $U_r=F_r$ and $y(U_r)=1$, so the lemma holds.
    Otherwise, let $s$ be a closest representative to $r$.
    We then have
    \[
        B_r = \frac{1}{2}d(r,s) > \max\{R_r,R_s\} \ge R_r,
    \]    
    where the first inequality is due to \Cref{obs:filter}. Consequently, we deduce
    \[
        \frac{1}{2} \le y( \{ i \in F_r: d(i,r) \le R_r \}) \le y(U_r) \le y(F_r)=1,
    \]
    where the first inequality follows from the admissibility of $R_r$.
    % \gptmargin{Corrected $F_j$ to $F_r$ and handled a single representative before choosing a different representative.}
\end{proof}
% \end{gptadded}

\paragraph{Dependent rounding.}

Before applying the dependent rounding, a matching $\cal{M}$ of $\calR$ is greedily constructed by iteratively matching two closest unmatched representatives. 
The matching $\calM$ is maximal, i.e., at most one representative is unmatched.

% \gptremovedfile{removed-text/robust1736-08}

% \begin{gptadded}

Next we apply the dependent rounding in \cite{srinivasan2001distributions} on a laminar family $\{F\}\cup \{U_r\cup U_s: \{r,s\}\in \calM\}\cup \{U_r: r\in \calR\}$. Let $Y_i\in \{0,1\}$ be a result from the dependent rounding for facility $i$. The rounding guarantees that
% \gptmargin{Corrected the total-opening identity: equality to $k$ requires $y(F)=k$. Rounding a noninteger total gives its floor or ceiling.}
\begin{align}
    \sum_{i\in F} Y_i & \in\{\lfloor y(F)\rfloor,\lceil y(F)\rceil\}, \label{eq:k-preserving} \\
    \E[Y_i] & = y_i, && \forall i \in F, \label{eq:marginal-preserving} \\
    Y(U_r) & \in \{\lfloor y(U_r)\rfloor, \lceil y(U_r) \rceil\}, && \forall r\in \calR, \label{eq:bundle-preserving} \\
    Y(U_r\cup U_s) & \ge 1, && \forall \{r,s\}\in \calM. \label{eq:pair-preserving}
\end{align} 
Furthermore, let $F_\open:= \{i\in F : Y_i=1\}$, and the rounding also has a negative-correlation property: 
\begin{equation} \label{eq:negative-correlation}
    \Pr(F_\open \cap S = \emptyset) \le \prod_{i\in S} (1-y_i) \le e^{-y(S)} \quad \forall S\subseteq F.
\end{equation}
% \end{gptadded}

We first show that every representative is guaranteed to be served deterministically by an open
facility within distance a constant factor of its bundle radius.

% \gptremovedfile{removed-text/robust1736-09}

% \begin{gptadded}
\begin{lemma} \label{lem:deterministic-backup}
    For any representative $r\in \calR$, $d(r,F_\open)\le 5B_r$.
\end{lemma}
\begin{proof}
    If $|\calR| = 1$, the lemma trivially holds since $B_r = \infty$.
    Hence, we assume $|\calR|\ge 2$.
    Let $S \subseteq \calR \setminus \{r\}$ denote the set of closest representatives to $r$, i.e., $S := \arg\min_{j \in \calR \setminus \{r\}} d(r, j)$.
    Note that $B_r = d(r, j) / 2$ for every $j \in S$ by definition.
    
    Suppose for the moment that there exists an unmatched representative in $S$ when $r$ is about to be matched.
    In this case, $r$ must be matched with a representative $j \in S$ due to the greediness of the matching procedure, and hence, there exists an open facility in the merged bundle $U_r \cup U_j$ due to \eqref{eq:pair-preserving}.
    We thus obtain
    \[
        d(r, F_\open)
        \leq d(r, j) + B_j
        = 2 B_r + B_j
        \leq 3 B_r,
    \]
    where the last inequality is due to the definition of $B_j := \frac{1}{2} \, d(j, \calR \setminus \{j\}) \leq d(j, r) / 2 = B_r$.

    On the other hand, if every representative in $S$ is matched before $r$, including the case that $r$ remains unmatched at termination, we consider an arbitrary $j \in S$ matched with another representative $j' \in \calR \setminus \{r, j\}$.
    Note that $d(j, j') \leq d(j, r)$ by the greediness of the matching procedure.
    Since the dependent rounding procedure guarantees to open a facility in the merged bundle $U_j \cup U_{j'}$ due to \eqref{eq:pair-preserving}, we deduce
    \[
        d(r, F_\open) \leq d(r, j) + d(j, j') + B_{j'}
        \leq d(r, j) + \frac{3}{2} d(j, j')
        \leq \frac{5}{2} d(r, j)
        = 5 B_r,
    \]
    where the second inequality is again due to the definition of $B_{j'} := \frac{1}{2} \, d(j', \calR \setminus \{j'\}) \leq d(j', j) / 2$.
\end{proof}
% \end{gptadded}

% \begin{gptadded}
Based on the above lemma, we can prove that the expected excess of connection cost for a representative beyond its threshold radius can be fully compensated by a constant factor of the partial expectation of $\calP_j$ beyond its threshold radius.
\begin{lemma} \label{lem:beyond-R}
    For any representative $r \in \calR$, $\E [(d(r,F_\open)-R_r)_+] \le 5T_r = 5 \, \E_{X \sim \calP_r} [X \mathsf{1}_{X > R_r}]$.
\end{lemma}
\begin{proof}
    If $\calR=\{r\}$, then $U_r=F_r$ has mass one and exactly one of its facilities opens by the dependent rounding procedure.
    We thus have
    \[
      \E[(d(r,F_\open)-R_r)_+]
      \le      
      \sum_{i\in F_r} y_i(d(i,r)-R_r)_+
      \le 
      \sum_{i \in F_r : d(i, r) > R_r} y_i d(i, r)
      =
      T_r.
    \]
    Let us now assume $|\calR| \ge 2$.
    % \gptmargin{Added the singleton case to avoid $0\cdot\infty$ in the backup calculation.}
    Since the dependent rounding procedure opens at most one facility from $U_r$ due to \eqref{eq:bundle-preserving}, we know that events $\{ i \in F_\open \}_{i \in U_r}$ partition the event $\{U_r \cap F_\open \neq \emptyset\}$.
    The marginal probability of a facility $i \in U_r$ being open is precisely $y_i$ due to \eqref{eq:marginal-preserving}, and under this event, we easily have
    \[
        (d(r, F_\open) - R_r)_+ \leq (d(i, r) - R_r)_+.
    \]
    Moreover, the event that none of $U_r$ are open happens with probability $1 - y(U_r)$, in which \Cref{lem:deterministic-backup} gives a backup bound as follows:
    \[
        (d(r, F_\open) - R_r)_+ \leq 5 B_r.
    \]
    We can therefore derive
    \begin{align*}
        \E[(d(r, F_\open) - R_r)_+]
        &
        \leq \sum_{i \in U_r} y_i \, (d(i, r) - R_r)_+ + 5 B_r (1 - y(U_r))
        \\&
        \leq \sum_{i \in U_r : d(i, r) > R_r} y_i \, d(i, r) + 5 B_r (1 - y(U_r))
        \\&
        \overset{(a)}{\leq} \sum_{i \in U_r : d(i, r) > R_r} y_i \, d(i, r) + 5 \sum_{i \in F_r \setminus U_r} y_i d(i, r)
        \\&
        \leq 5 \sum_{i \in F_r : d(i, r) > R_r} y_i \, d(i, r)= 5 T_r,
    \end{align*}
    where $(a)$ is due to the fact that every facility $i \in F_r \setminus U_r$ is at distance at least $B_r$ by definition.
\end{proof}
% \end{gptadded}

\subsection{Analysis}

In this section, we prove \Cref{lem:robust}, which is a key lemma for proving \Cref{thm:robust}. 
Note that \Cref{thm:robust} indeed follows from \Cref{lem:assurance-bound1,lem:assurance-bound2,lem:robust}.

\begin{lemma} \label{lem:robust}
    For every client $j \in C$,
    $\E[d(j,F_\open)]\le \left(1-\frac{1}{e}\right)V_j+\frac{3}{e}A_j$.
\end{lemma}

% \gptremovedfile{removed-text/robust1736-10}

% \begin{gptadded}
Following is a simple fact easily shown by convexity.
\begin{lemma} \label{lem:exp_inequality}
    For any $q \in [0, 1]$, $e^{-q} \leq \frac{1}{e} + (1 - \frac{1}{e})(1-q)$.
\end{lemma}
\begin{proof}
    The right-hand side is the line segment connecting $(0, 1)$ and $(1, 1/e)$.
    Observe that these two points are also on the left-hand side $e^{-q}$.
    The convexity of $e^{-q}$ then completes the proof.
\end{proof}
% \end{gptadded}

% \begin{gptadded}
We are now ready to prove the key lemma.
\begin{proof} [Proof of \Cref{lem:robust}]
    Let $r := r(j)$ denote the representative deleting $j$ in the filtering procedure, and let $L_j := d(j, r) + R_r$. 
    Observe that $L_j$ is an upper bound of $d(j,F_\open)$ if a facility is open in the ball of threshold radius $R_r$ around $r$. Otherwise, we use a backup facility, in which the distance is bounded above by \Cref{lem:beyond-R}.
    
    Since $d(j, F_\open)$ is nonnegative, we can write % (see, e.g., Lemma 2.9 in \cite{mitzenmacher2005probability})
    \begin{align*}
        \E[d(j, F_\open)]
        & 
        = \int_0^\infty \Pr[d(j, F_\open) > t] \, dt
        \\&
        = \underbrace{\int_0^{L_j} \Pr[d(j, F_\open) > t] \, dt}_{\text{(I)}}
        + \underbrace{\int_{L_j}^\infty \Pr[d(j, F_\open) > t] \, dt.}_{\text{(II)}}
    \end{align*}

    We first bound (I) above.
    For any $t > 0$, let $F_j(t) := \{i \in F_j : d(i, j) \leq t\}$ denote the facilities at distance no greater than $t$.
    We then have the following bound on the integrands:
    \[
        \Pr[d(j, F_\open) > t]
        \leq \Pr[F_\open \cap F_j(t)  = \emptyset]
        \leq e^{-y(F_j(t))}
        \leq \frac{1}{e} + \Big( 1 - \frac{1}{e} \Big) \Big( 1 - y \big( F_j(t) \big) \Big),
    \]
    where the first inequality uses $F_j(t)\subseteq\{i\in F:d(i,j)\le t\}$, the second is due to \eqref{eq:negative-correlation}, and the third is due to \Cref{lem:exp_inequality}.
    % \gptmargin{Changed equality to inequality: an open facility outside $F_j$ can also serve $j$ within radius $t$.}
    This inequality further gives
    \[
        \text{(I)}
        \leq \frac{L_j}{e} + \Big( 1 - \frac{1}{e} \Big) \int_0^{L_j} \Big( 1 - y \big( F_j(t) \big) \Big) \, dt
        \leq \frac{L_j}{e} + \Big( 1 - \frac{1}{e} \Big) V_j,
    \]
    where the last inequality comes from the fact that
    \[
        V_j 
        = \E_{X \sim \calP_j} [X]
        = \int_0^{\infty} \Pr(X > t) \, dt
        = \int_0^{\infty} \Big( 1 - y \big( F_j(t) \big) \Big) \, dt.
    \]

    For (II), observe that this term can be written as $\E[(d(j, F_\open) - L_j)_+]$. %\ysmargin{Do we have a good reference to cite for this fact: (II) $ = \E[(d(j, F_\open) - L_j)_+]$?}
    Let $i_j$ and $i_r$ be the nearest open facilities to clients $j$ and $r$, respectively, i.e., $d(j,F_\open)=d(j,i_j)$ and $d(r,F_\open)=d(r,i_r)$.
    % \gptmargin{Corrected the client and facility indices and specified that the closest facilities are open.}
    We thus have
    \begin{align*}
        \text{(II)}
        & 
        = \E[(d(j,F_\open) - L_j)_+]
        \\&
        = \E\Big[ \big(d(j,i_j) - (d(j,r)+R_r) \big)_+\Big] 
        \\& 
        \le \E\Big[ \big(d(j,i_r) - (d(j,r)+R_r) \big)_+\Big] 
        \\&
        \le \E[(d(F_\open,r)-R_{r})_+]
        \\& 
        \le 5T_r,
    \end{align*}
    where the last inequality follows from \Cref{lem:beyond-R}.

    Combining these two upper bounds for (I) and (II), we finally have
    \begin{align*}
        \E[d(j,F_\open)]
        &
        \leq \left(1-\frac{1}{e}\right) V_j + \frac{d(j,r)+R_{r}}{e}+5T_{r} 
        \\&
        \overset{(a)}{=} \left(1-\frac{1}{e}\right) V_j + \frac{d(j, r) + A_r}{e}
        \\&
        \overset{(b)}{\leq} \left(1-\frac{1}{e}\right) V_j + \frac{2 \max \{R_j, R_r\} + A_j}{e}
        \\&
        \overset{(c)}{\leq} \left(1-\frac{1}{e}\right) V_j + \frac{3 A_j}{e},
    \end{align*}
    where $(a)$ follows from the definition of $A_r := R_r + 5e \, T_r$, $(b)$ from \Cref{obs:filter}, and $(c)$ from the fact that
    $
        R_j \leq A_j \text{ and } R_r \leq A_r \leq A_j
    $
    due to the definitions of $A_j$ and $A_r$ and again \Cref{obs:filter}.
\end{proof}
% \end{gptadded}

\section{Combining the Roundings: An Approximation Below Two}
\label{sec:final-analysis}

In this section, we combine all the components to derive our final better-than-2 approximation. Our goal is to obtain the following pseudo-approximation: an $\alpha$ approximation for the $k$-median problem using $k + c$ facilities, for some $\alpha = 2 - \Omega(1)$ and $c = O(1)$.  This is sufficient as the reduction of Li and Svensson~\cite[Theorem 4]{LS16} converts such an algorithm into an $(\alpha+\tau)$-approximation in time $n^{O(c/\tau)}$ times the running time of the pseudo-approximation. 

We first solve the LP relaxation \eqref{lp:k-clustering} for $k$-median. Let $\varepsilon > 0$ be a sufficiently small constant. We perform the preprocessing step for this $\varepsilon$ in
Section~\ref{sec:preprocessing} to obtain a $\delta$-integral solution,
for some $\delta=\varepsilon^{O(1)}$.
By Theorem~\ref{thm:preprocessing}, the total opening is at most $k+3$
with probability at least $1-\exp(-\Omega(\varepsilon^{-10}))$.
We repeat preprocessing until this event holds. Conditioning on acceptance gives an overall
multiplicative loss in expected cost of at most
$$
 \frac{1+O(\varepsilon)}{1-\exp(-\Omega(\varepsilon^{-10}))}
 =1+O(\varepsilon).
$$

We condition on the accepted output in the analysis below. Abusing notations slightly, we still use $(\mathbf{x}, \mathbf{y})$ to denote the new LP solution. Define $V_j,\Disp_j,M_j$ for every $j \in C$ using
this solution.
% , and put
% \[
%  V=\sum_{j\in C}V_j,\qquad
%  \Disp=\sum_{j\in C}\Disp_j,\qquad M=\sum_{j\in C}M_j.
% \]
% The support weights sum to one. Jensen's inequality and the triangle inequality give
% \begin{equation}\label{eq:final-profile-relations}
%  M_j\le\Disp_j\le2V_j,\qquad M\le\Disp\le2V.
% \end{equation}
% These inequalities also cover $V_j=0$, in which case both $M_j$ and
% $\Disp_j$ vanish.

Let $\rho\in(0,1]$. By Theorem~\ref{thm:conversion-local-charges} and
Lemma~\ref{lem:conversion-disp-cert}, Algorithm~\ref{sec:conversion-algorithm} returns a random set $O_1$
such that
\begin{equation}\label{eq:final-A}
 |O_1|\le |\mathbf{y}|_1+s,\qquad
 \E[d(j,O_1)]\le(1+\rho)\left(2V_j-\frac{\Disp_j}{12}\right), \forall j \in C,
\end{equation}
where $s=O(\delta^{-2}\rho^{-2}\log(4/\rho))$. Here the expectation
already includes rejection of trials exceeding the facility budget.

The algorithm from Theorem~\ref{thm:robust} returns a random set $O_2$ with
\begin{equation}\label{eq:final-B}
 \begin{aligned}
 |O_2|&\le\lceil |\mathbf{y}|_1\rceil,\\
 \E[d(j,O_2)]&\le\left(1+\frac2e\right)V_j
       +\sqrt{\frac{90V_jM_j}{e}}+\frac{15}{2}M_j, \quad \forall j \in C.
 \end{aligned}
\end{equation}

Let $O$ be the cheaper of $O_1$ and $O_2$, breaking ties arbitrarily. Clearly, we have $|O|\le\max\{|\mathbf{y}|_1+s,\lceil |\mathbf{y}|_1\rceil\}$. The following lemma bounds the cost of $O$:
\begin{lemma}\label{lem:final-combination}
 $\E[\cost(O)]\le(1+\rho)\alpha^\star V$,
where
\begin{equation}\label{eq:final-alpha}
 \begin{aligned}
 a&=\frac{91}{12},\qquad b=\sqrt{\frac{90}{e}},\qquad c=1-\frac2e,\\
 u^\star&=\frac{2c}{b+\sqrt{b^2+4ac}},\qquad
 \beta^\star=(u^\star)^2,\qquad
 \alpha^\star=2-\frac{\beta^\star}{12}.
 \end{aligned}
\end{equation}
Numerically, $\beta^\star\approx0.00188669438974$ and
$\alpha^\star\approx1.99984277546752$.
\end{lemma}

\begin{proof}
Let $V = \sum_{j \in C} V_j, \Disp = \sum_{j \in C} \Disp_j$ and $M = \sum_{j \in C} M_j$. Summing~\eqref{eq:final-A} gives
\[\E[\cost(O_1)]\le(1+\rho)(2V-\Disp/12).\]
For the analysis of the second algorithm, 
Cauchy--Schwarz gives
\[
 \sum_j\sqrt{V_jM_j}\le\sqrt{VM}\le\sqrt{V\Disp},
\]
and hence
\[
 \E[\cost(O_2)]\le\left(1+\frac2e\right)V
       +\sqrt{\frac{90V\Disp}{e}}+\frac{15}{2}\Disp.
\]
If $V=0$, both expected costs are zero and the claim follows.
Otherwise let $\beta=\Disp/V\in[0,2]$ and define
\[
 f(\beta)=2-\frac{\beta}{12},\qquad
 g(\beta)=1+\frac2e+\sqrt{\frac{90\beta}{e}}+\frac{15}{2}\beta.
\]
Both functions are nonnegative on this interval. Therefore
\begin{align}
 \E[\cost(O)] &\leq \min\{\E[\cost(O_1)],\E[\cost(O_2)]\}\nonumber\\
 &\le V\min\{(1+\rho)f(\beta),g(\beta)\}
 \le(1+\rho)V\min\{f(\beta),g(\beta)\}. \label{eq:final-envelope}
\end{align}
The function $f$ is strictly decreasing and $g$ is strictly increasing,
with $g(0)<f(0)$ and $g(2)>f(2)$. Thus their lower envelope has its
unique maximum at their intersection. The equation $f(\beta)=g(\beta)$ is
\[
 \frac{91}{12}\beta+\sqrt{\frac{90\beta}{e}}=1-\frac2e.
\]
Substituting $u=\sqrt{\beta}$ gives $au^2+bu-c=0$, whose positive
root is~\eqref{eq:final-alpha}. This proves the cost bound.
The facility bound holds in every outcome by
\eqref{eq:final-A} and~\eqref{eq:final-B}.
\end{proof}

% For a numerical bound that does not depend on rounded decimal roots,
% take $\bar s=9/5000$. If $\beta\ge\bar s$, then
% $f(\beta)\le2-\bar s/12=1.99985$. If $\beta\le\bar s$, use
% $e>27/10$ and $\sqrt{3/50}<49/200$ to obtain
% \[
%  g(\beta)\le g(\bar s)
%  <1+\frac{20}{27}+\frac{49}{200}+\frac{27}{2000}
%  <1.99985.
% \]
% Consequently,
% \begin{equation}\label{eq:final-rational-bound}
%  \alpha^\star\le1.99985<1.9999.
% \end{equation}

When our $\varepsilon$ and $\rho$ are sufficiently small, the approximation ratio we obtain is at most $1.9999$. This proves Theorem~\ref{thm:intro-main}.

\section*{Acknowledgment}
The work of Shi Li and Zaixuan Wang is supported by State Key Laboratory for Novel Software Technology, New Cornerstone Science Foundation, and Fundamental and Interdisciplinary Disciplines Breakthrough Plan of the Ministry of Education of China (No.JYB2025XDXM118). 
The work of Jaros{\l}aw Byrka, Yongho Shin, and Phanu Vajanopath is supported by Polish National Science Centre grant 2020/39/B/ST6/01641. The work of Yang Hu is partially supported by NSF grant CNS-2504471, by NSF CAREER award CCF-2542165, and by a Jane Street Research Grant.

\paragraph{AI methodology:}
We used ChatGPT-6 Astra to assist in developing ideas and proofs in Sections 3, 4, and 7, and ChatGPT-6.1 Sol and ChatGPT-6 Astra to assist in developing the arguments in Section 5. ChatGPT-6 Astra also developed the local-charge certificate in Section 6 (\Cref{def:conversion-certificate}). ChatGPT-6 Astra generated the Lean formalization of \Cref{lem:median-certificate}, available at \url{https://github.com/phanuv/Euclidean-k-Median-Certificate}. 
We also used these tools, together with Claude Opus 5.5, to improve the presentation of the paper. The authors substantially revised the AI-generated material to clarify the mathematical ideas and make the proofs accessible to readers. The authors verified the correctness and originality of all content, including the Lean formalization and references, and take full responsibility for the paper.

\bibliographystyle{alpha}
\bibliography{refs}

\clearpage
\appendix
\section{Illustration of Atomization}\label{app:atoms}
\begin{figure}[htbp]
    \centering
    \includegraphics[width=\textwidth]{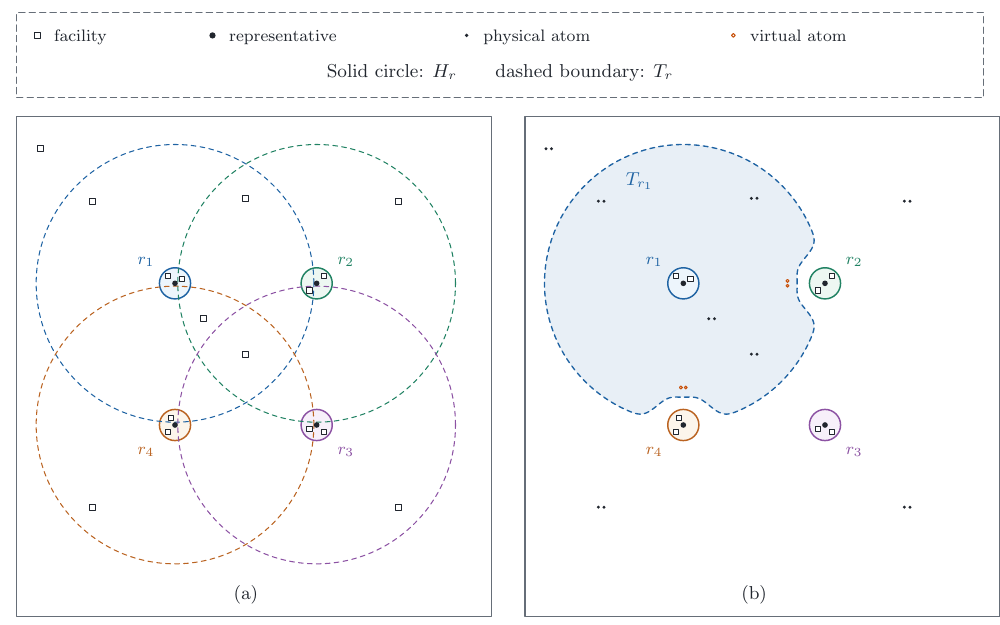}
%    \captionsetup{labelfont={color=blue},textfont={color=blue}}
    \caption{Illustration of the atomization in Section~\ref{sec:preprocessing}.
    (a) Facilities before atomization: squares denote facilities and solid dots
    denote representatives. Each small inner circle depicts a head $H_r$;
    the surrounding annulus inside the corresponding dashed circle depicts
    the original tail $T_r$.
    (b) After atomization: only the tail $T_{r_1}$ is shown, with a blue
    dashed boundary and light blue shading. Its boundary bends inward to
    exclude the other heads; the tail also excludes its own head $H_{r_1}$.
    The small orange circles show the virtual
    atoms belonging to $T_{r_1}$, arising from the facilities in
    $H_{r_2}\cap T_{r_1}$ and $H_{r_4}\cap T_{r_1}$ in (a).
    Virtual atoms are colocated with their original facilities in the
    source heads, as defined in Section~\ref{sec:preprocessing}; for clarity,
    they are drawn outside those heads and inside the schematic region
    representing $T_{r_1}$. Physical atoms are also colocated with their
    original facilities and are slightly offset in the drawing for visibility.
    Virtual atoms are auxiliary variables for tail rounding and do not
    contribute to the final facility openings.}
    \label{fig:atoms}
\end{figure}

\end{document}